\documentclass[10pt,journal,compsoc]{IEEEtran}

\usepackage{tikz}
\usepackage{amsmath}

\usepackage{filecontents}
\usepackage{amsmath}                                                                                                                                                                    
\usepackage{amssymb}
\usepackage{amsfonts}
\usepackage{amsthm}
\usepackage{algorithm}
\usepackage{algpseudocode}
\usepackage{graphicx}
\usepackage{textcomp}
\usepackage{xcolor}
\usepackage{soul}
\usepackage{enumitem}

\usepackage{enumitem}
\setlist[enumerate]{itemsep=0mm}
\usepackage{stackengine}
\usepackage{bm}
\usepackage{tikz}
\usepackage{filecontents}
\usepackage{multirow}
\usepackage{booktabs}
\usepackage{color}
\usepackage{algorithm}  
\usepackage{algorithmicx}  
\usepackage{subcaption}
\usepackage{caption}
\usepackage[normalem]{ulem}
\usepackage{threeparttable}
\usepackage{mathtools}
\usepackage{hyperref}
\makeatletter

\newcommand{\Rmnum}[1]{\expandafter\@slowromancap\romannumeral #1@}
\makeatother

\newcommand{\CD}{\ensuremath{\mathcal{D}}}

\newcommand{\CL}{\ensuremath{\mathcal{L}}}

\newcommand{\CN}{\ensuremath{\mathcal{N}}}

\newcommand{\CR}{\ensuremath{\mathcal{R}}}

\newcommand{\CT}{\ensuremath{\mathcal{T}}}
\newcommand{\CU}{\ensuremath{\mathcal{U}}}
\newcommand{\CV}{\ensuremath{\mathcal{V}}}
\newcommand{\CW}{\ensuremath{\mathcal{W}}}
\newcommand{\CX}{\ensuremath{\mathcal{X}}}

\newcommand{\CZ}{\ensuremath{\mathcal{Z}}}

\newcommand{\BR}{\ensuremath{\mathbb{R}}}
\newcommand{\BP}{\ensuremath{\mathbb{P}}}

\newcommand{\BSI}{\ensuremath{\boldsymbol{I}}}
\newcommand{\BSBETA}{\ensuremath{\boldsymbol{\beta}}}

\definecolor{darkblue}{RGB}{48, 84, 151}
\definecolor{lightblue}{RGB}{222, 235, 246}
\definecolor{ao(english)}{rgb}{0.0, 0.5, 0.0}
\definecolor{applegreen}{rgb}{0.55, 0.71, 0.0}

\useunder{\uline}{\ul}{}

\DeclareMathOperator*{\argmax}{arg\,max}

\usepackage[english]{babel}
\newtheorem{theorem}{Theorem}

\theoremstyle{definition}
\newtheorem{definition}{Definition}  

\begin{document}

\title{``Training robust watermarking model may hurt authentication!'' Exploring and Mitigating the Identity Leakage in Robust Watermarking}


\author{Xinyu Zhang$^{1,2,3}$, Ziping Dong$^{1,2}$, Qingyu Liu$^{1,2}$, Yuan Hong{$^{4,*}$}, Zhongjie Ba$^{1,2,*}$, and Kui Ren$^{1,2}$ \\
\textsuperscript{1}State Key Laboratory of Blockchain and Data Security, Zhejiang University \\
\textsuperscript{2}Hangzhou High-Tech Zone (Binjiang) Institute of Blockchain and Data Security \\
\textsuperscript{3}Tongyi Lab, Alibaba Group 
\textsuperscript{4}University of Connecticut
\thanks{{$^*$}Corresponding author: Zhongjie Ba and Yuan Hong. \\ Emails: \{xinyuzhang53, zhongjieba\}@zju.edu.cn, yuan.hong@uconn.edu.}
}






\maketitle


\begin{abstract}




The rapid advancement of generative AI has underscored the critical need for identifying image ownership and protecting copyrights. This makes post-processing image watermarking an essential tool---it involves embedding a specific watermark message into an image, with successful verification if a similar message can be decoded from the watermarked image. However, this method is susceptible to both adversarial attacks that manipulate the watermarked image to yield an unverified message upon decoding, and the proposed identity leakage-related attacks (e.g., forging watermarked images). The threat of identity leakage is particularly exacerbated in both empirical and certified robust watermarking methods.  

To defend against the aforementioned attacks, we propose W-IR, the first image watermarking framework that simultaneously incorporates identity protection and robustness. To enhance model robustness, we introduce a novel randomized smoothing technique as part of a robust watermarking, that offers certified robustness against perturbations across two distinct transformation spaces: pixel-level and coordinate-level.
Moreover, to further mitigate identity leakage, we propose a new strategy based on residual information loss, aimed at minimizing the mutual information between the residual and watermarked images. Our work strikes a superior balance between robustness and identity leakage mitigation. Extensive experiments demonstrate that our W-IR framework achieves high certified accuracy for authenticity while effectively reducing identity leakage.
\footnote{The code is available at \url{https://github.com/holdrain/W-I-R}.}

\end{abstract}

\begin{IEEEkeywords}
Post-processing image watermarking, Adversarial attack, Identity leakage attacks, Randomized smoothing, Mutual information, Identity protection.
\end{IEEEkeywords}




\section{Introduction}
With the rapid advancement and widespread adoption of generative AI technologies like DALL·E~\cite{DALL·E3O73:online} and Stable Diffusion~\cite{StableDi13:online}, 
the generation of high-quality, commercially valuable images has become increasingly prevalent.
In this scenario, deep learning-based post-processing image watermarking~\cite{zhu2018hidden,stegastamp,zhang2020udh} emerges as a crucial tool for copyright protection, ownership tracking, and identification~\cite{MetaGoog27:online,ray2020recent}, and is widely adopted in real-world applications~\cite{rombach2022high,dathathri2024scalable}.  %
Among these techniques, the industry~\cite{Watermar98:online,Announci38:online} favors invisible watermarks~\cite{wolfgang1996watermark,zhu2018hidden} due to their ability to preserve visual fidelity while resisting straightforward removal attempts~\cite{zhao2024invisible}.
Invisible image watermarking can assist in verifying the image copyright through a three-step process--- 
embedding a unique and imperceptible copyright watermark message into an image with an \textit{encoder network} to produce a watermarked image, extracting the message with a \textit{decoder network}, and verifying the match with the original watermark message through a \textit{verification function}. 
The integration of the decoder network and the verification function is termed \textit{watermarking authentication}. 
Similarly, invisible image watermarking enables creators to track illegal distribution channels by embedding and authenticating unique secret watermarks for each channel. 


\begin{figure}[!t]
  \centering
  \includegraphics[width=\linewidth]{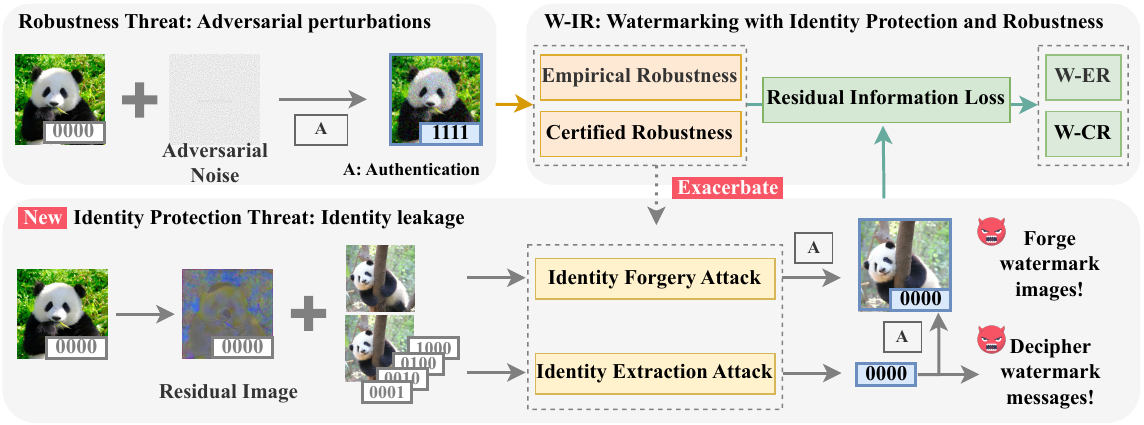}
  \vspace{-4mm}
  \caption{Our main contributions are: (1) the discovery of identity leakage and corresponding attacks in post-processing watermarking (which are exacerbated in the robust case), and (2) W-IR: robust watermarking with mitigated identity leakage.} 
  \label{fig:intro_framework}
  \vspace{-4mm}
\end{figure}

Recent studies have revealed that deep learning-based post-processing watermarking techniques are susceptible to adversarial attacks~\cite{goodfellow2014explaining,engstrom2019exploring,saberi2023robustness}. Adversaries can alter watermarked images in ways that are imperceptible to humans, but cause the decoder network to extract the message incorrectly, leading to authentication failure~\cite{jiang2023evading}. 
However, these attacks focus solely on robustness vulnerabilities, leaving the vulnerabilities related to identity leakage (especially critical in invisible watermarks) and the relationship between robustness and identity leakage largely unexplored.

\begin{table*}[!t]
\setlength\tabcolsep{2pt}
\footnotesize
\caption{Comparison of different types of image watermarking defenses against identity leakage attacks.}
\vspace{-2mm}
\begin{threeparttable}
\resizebox{\linewidth}{!}{
\begin{tabular}{lccccl}
\toprule
Method & Identity leakage threat & Detection Type~\cite{wan2022comprehensive} & Required storage & Key management & Computational Cost \\
\midrule
Baseline (no protection) & Linking \& forgery \& extraction & Blind & Watermarked image only & Universal watermark message~\cite{jiang2024watermark} & Low \\
One watermark per image & Forgery \& extraction & Non-blind & Original \& watermarked image & Image-specific watermark message & Low \\
Crypto-based$^\dagger$ & Forgery \& extraction & Non-blind & Original \& watermarked image & Universal key~\cite{sanivarapu2022digital} & High$^*$ \\
W-IR (Ours) & Mitigated all attacks & Blind & Watermarked image only & Universal watermark message & Low \\
\bottomrule
\end{tabular}
}
\begin{tablenotes}
    \item {$^\dagger$ For example, watermark message $t=\mathtt{func}(\mathtt{HMAC}(\text{key}, \text{original image}))$.}
    \item {$^*$Additional watermark message computation cost for encoding and authentication.}
\end{tablenotes}
\end{threeparttable}
\label{tab:intro}
\vspace{-4mm}
\end{table*}

In this paper, we uncover a new serious threat to watermarking:
identity leakage. This vulnerability allows 
adversaries to forge watermarked images carrying the genuine watermark message of a target user without ever decoding the watermark itself (i.e., \textit{identity forgery attacks}), using only a single image. Different from existing attacks that alter watermarked images directly, this attack crafts fake watermarked images to \textit{deceive watermark authenticity models}. Further, adversaries can decipher the specific watermark message bit by bit by reasonably accessing the watermark embedding services (i.e., \textit{identity extraction attacks}). This attack enables more accurate forgeries of watermarked images that \textit{deceive watermark authenticity verification}, severely compromising ownership tracking and identification, particularly in critical scenarios like police evidence authentication. 
%
This phenomenon arises primarily due to the decoder network's proficiency in accurately \textit{extracting watermark messages from watermarked and residual images}, as illustrated in Figure~\ref{fig:intro_framework}. 
Here, the residual image is defined as the difference between the watermarked image and the original one, and can be estimated using methods such as low-pass filtering~\cite{kutter2000watermark} and Variational Autoencoder (VAE) \cite{rombach2022high}.

\begin{figure}[!t]
  \centering
  \includegraphics[width=\linewidth]{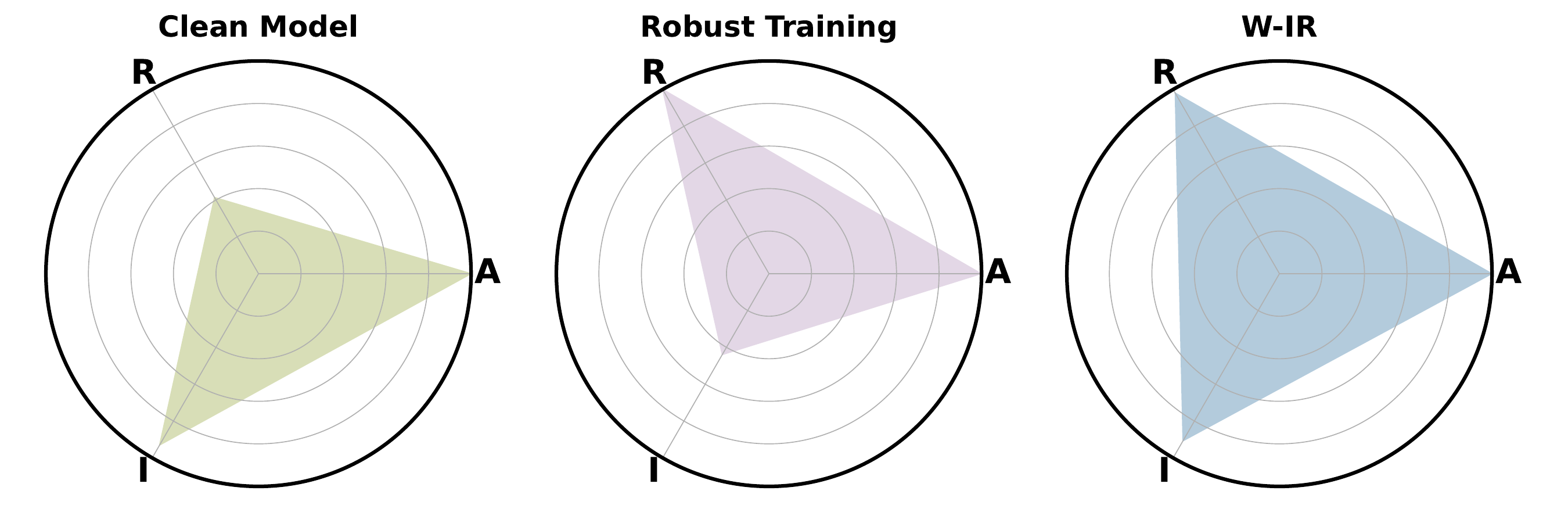}
  \vspace{-4mm}
  \caption{Visualization for the three-facet performance (A -- Authenticity, R -- Robustness, I -- Identity Protection) of three types of watermarking strategies. 
  }  
  \label{fig:visualize}
   \vspace{-4mm}
\end{figure}

To address the aforementioned vulnerabilities, 
we introduce W-IR, the first image watermark framework that integrates identity protection with robustness. W-IR can be universally combined with both empirical robustness (W-ER) and certified robustness (W-CR), as illustrated in Figure~\ref{fig:intro_framework}.


To enhance watermark robustness, researchers have developed various empirical defenses, including the integration of image augmentation~\cite{stegastamp,jia2021mbrs} and adversarial training~\cite{zhu2018hidden} into the training phases of encoder and decoder networks. 
However, they are broken by adaptive or stronger attacks~\cite{jiang2023evading,saberi2023robustness}.
Certified defenses~\cite{li2023sok} end the cat-and-mouse game between attacks and defenses primarily for classification models, and offers provable robustness guarantees against adversarial perturbations. Among them, randomized smoothing (RS)~\cite{cohen2019certified,zhang2020black} is the state-of-the-art and widely adopted due to its applicability to any model and its ability to achieve acceptable accuracy and efficiency on large-scale datasets (see Section~\ref{sec:RS}). 
Motivated by the success of RS in image classification, we aim to adapt it to counter adversarial perturbations in image watermarking, a largely unexplored field. 

We formulate the watermark authenticity, which encompasses the decoder network and verification function, as a classification model. Building on this model, we propose a novel randomized smoothing method to ensure certified robustness for image watermarking against two forms of image perturbations: modifying the pixel values and/ or pixel coordinates. Specifically, we model image perturbations as combinations of pixel and coordinate transformations and then offer provable robustness guarantees within the respective pixel and coordinate spaces. 


More crucially, 
while robust training bolsters model resistance against adversarial attacks, it simultaneously greatly exacerbates identity leakage (see Figure \ref{fig:visualize}).
This is because robust training methods (e.g., adversarial training and randomized smoothing) incorporate noise to counter adversarial perturbations, inadvertently enhancing the decoder network's capacity for extracting messages. 
For instance, we observe that the residual images obtained from robust models contain more watermark messages than those from normally trained (clean) models.  
As also illustrated in Figures~\ref{fig:id_link} and \ref{fig:id_link_robust}, 
residual images corresponding to the same message are closer on robust models than clean models.

To further enhance identity protection in robust watermarking, we theoretically analyze the mutual information to sufficiently preserve the identity information in watermarked images while reducing the watermark message in residual images. Although crypto-based methods (e.g., watermark message $t=\texttt{func}(\texttt{HMAC}(\mathrm{key}, \mathrm{original\ image}))$) can mitigate identity linking attacks, they remain vulnerable to identity forgery and extraction attacks while incurring high storage and computational overhead (see Table~\ref{tab:intro}). The mutual information estimation, however, is still a well-known challenging problem~\cite{poole2019variational,belghazi2018mutual}. Inspired by ~\cite{tian2021farewell}, we further use variational inference to derive \textit{residual information loss} (see Section~\ref{sec:ril}) to fit this mutual information objective but without explicitly estimating it. Figure \ref{fig:visualize} visualizes the performance improvement of our W-IR on authenticity, robustness, and identity protection/leakage. 



Therefore, our main contributions are as follows:


\begin{itemize}[leftmargin=*]

    
    \item We explore various vulnerabilities in post-processing watermarking techniques, particularly the risks of identity leakage that may be exacerbated by robust models. 


    \item We introduce W-IR, the first image watermarking framework that ensures both identity protection and robustness. W-IR can be universally integrated with different robustness methods, e.g., empirical robustness (W-ER) and certified robustness (W-CR). 
    
    Specifically, we develop a certified robust watermarking method that addresses pixel and coordination-level image perturbations. To protect identity, we mathematically formulate a mutual information objective that preserves necessary identity information in watermarked images while effectively discarding it in residual images.

    





    \item We evaluate the proposed identity leakage attacks and the W-IR defenses (W-ER and W-CR) on multiple datasets and post-processing watermarking methods.
    The results show that the three attacks can achieve high success rates. Further, our W-IR largely reduces the attack performance (providing identity protection) without compromising robustness.


    
\end{itemize}

\section{Preliminaries}
\subsection{Post-processing Image Watermarking}

\begin{figure}[t]
  \centering
  \includegraphics[width=\linewidth]{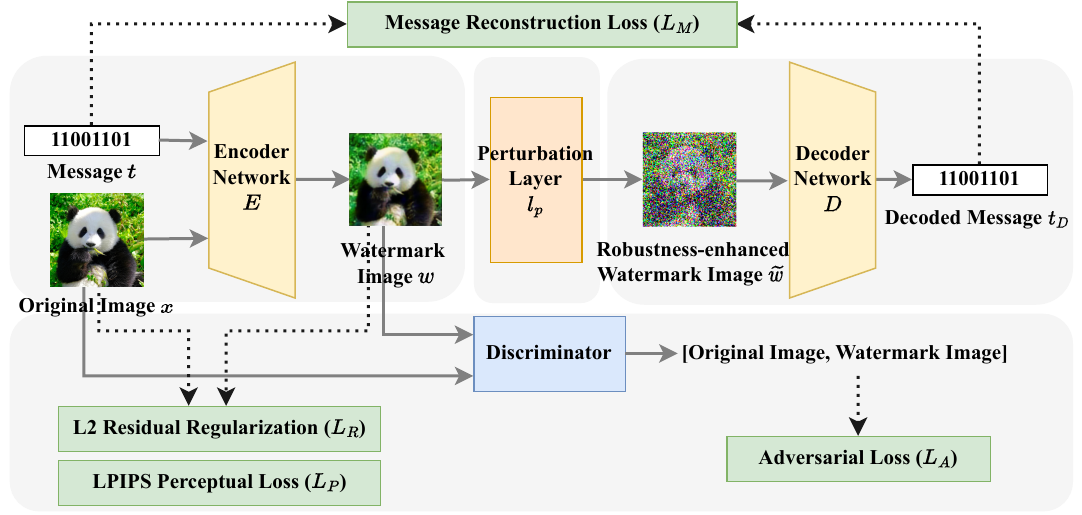}
  \vspace{-4mm}
  \caption{Training phase of image watermarking.} 
  \label{fig:watermarking}
   \vspace{-4mm}
\end{figure}

Image watermarking is a technique employed to embed a unique identifier (e.g., signal, pattern,  information) into an image. This process asserts ownership, protects intellectual property, tracks provenance, and provides authentication proof. 
%
The generic neural network-based post-processing watermarking process comprises four modules: encoder, image-enhancement module, decoder, and visual quality augmentation module (see Figure~\ref{fig:watermarking}). The encoder 
embeds a watermark message into an original image, transforming it into a watermarked image. 
The robustness-enhancement module performs data augmentation to empirically increase the watermark's resilience against various perturbations (e.g., geometric transformations, filters, compression). The decoder 
retrieves the embedded message from the robustness-enhanced watermarked image. Finally, the visual quality augmentation module maintains the imperceptibility of the watermark by minimizing its impact on the original image's visual fidelity, thereby preserving its aesthetic quality for human observers or computational analysis.

\vspace{0.05in}
\noindent\textbf{Training Phase: Encoder-Decoder Optimization.} 
Formally, given the original image $x$ and an $n$-bit binary string representing the watermark message $t$ (e.g., $11001101$), the encoder network $E$ obtains the watermarked image as $w=E(x, t)$. The perturbation layer $l_p$ then obtains the image-enhanced watermarked image $\tilde{w}=l_p(w)$ and the decoder network $D$ extracts the message from it. After deriving the decoded message, the decoder and encoder use a message reconstruction loss $\CL_M$ to derive the loss $\CL_M(D(l_p(E(x, t))),t)$. Besides the reconstruction loss, there is also the visual quality loss to improve the image's visual fidelity. Usually, $\CL_2$ residual regularization loss $\CL_R(E(x, t), x)$ and LPIPS perceptual loss~\cite{johnson2016perceptual,zhang2018unreasonable} $\CL_P(E(x, t), x)$ are used to calculate the difference between the watermarked image and the original image. It is also possible to construct an additional discriminator for distinguishing the original watermarked image and calculating the adversarial (discriminator) loss $\CL_A(E(x, t), x)$~\cite{zhu2018hidden}. The weighted sum of these loss components is used to update the encoder and decoder:

{
\begin{equation}
    \CL = \eta_M \CL_M + \eta_R \CL_R + \eta_P \CL_P + \eta_{A} \CL_{A}   
\label{eq:L_loss}
\end{equation}
}


\par
\noindent\textbf{Authentication Phase: Extracting and Verifying the Watermark.}
\label{sec:inference}
Once the encoder and decoder have been trained, given an original image $x$ and watermark message $t$ (i.e., $n$-bit $01$ string), the encoder can embed the watermark message into the image to produce a watermarked image $w=E(x, t)$, as shown in Figure~\ref{fig:authentication_watermarking}. During the authentication phase, when the watermarked image is fed into the decoder, a string of information $t_D$ with the same length as $t$ is obtained as $t_D=D(w)$. 
The bit accuracy is calculated by counting the number of matching positions between $t_D$ and $t$. If the bit accuracy exceeds a certain threshold $\tau$, the watermark is considered to be decoded correctly and verified successfully, which can be formulated as $V(t_D, t, \tau) = 1$. Otherwise, the decoding is deemed incorrect, $V(t_D, t, \tau) = 0$. 
The threshold $\tau$ is often determined by the false positive rate and the length of $t$~\cite{lukas2023ptw,stablesignature,jiang2023evading}.

\begin{figure}[t]
  \centering
  \includegraphics[width=\linewidth]{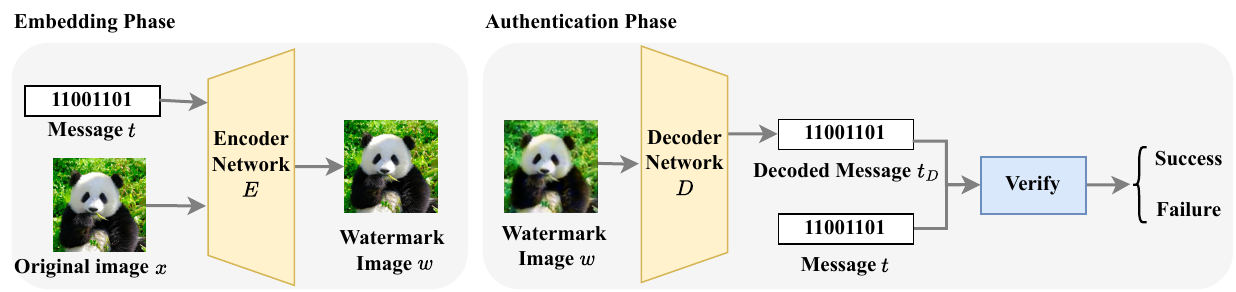}
  \vspace{-4mm}
  \caption{Authentication phase of image watermarking. 
  } 
  \label{fig:authentication_watermarking}
   \vspace{-4mm}
\end{figure}


\subsection{Watermarking Adversarial Perturbation}

Adversarial attacks on image watermarking involve slight modifications to a watermarked image 
to generate a perturbed watermarked image 
(aka. adversarial example) during the authentication phase~\cite{jiang2023evading}. These modifications are imperceptible to human observers but can deceive the decoder, resulting in low bit accuracy and failed verification. 
A direct way for generating adversarial perturbation on watermarked images is adding random noise, such as \textit{Gaussian noise}, to the watermarked image. 
Another type of adversarial perturbation is created through transforming the images, while preserving their semantics~\cite{an2024benchmarking}, e.g.,  
\textit{affine} transformation~\cite{an2024benchmarking, ma2022towards}. 
A third type of method is based on more refined and purposeful perturbations, e.g., precisely adjusting the perturbation on each image pixel via an optimization algorithm~\cite{jiang2023evading}. 


Rule-based adversarial perturbation, including random noise and semantic transformation, is intuitive and easy-to-implement, 
which is the main focus of  the paper. 



\subsection{Randomized Smoothing for Classification}
\label{sec:RS}

Randomized smoothing is a widely adopted defense strategy that offers provable robustness guarantees to classifiers against adversarial examples~\cite{cohen2019certified,lecuyer2019certified,jia2020certified,wang2021certified,zhang2023text}. This technique is distinguished by its applicability across various classifier architectures and its ability to maintain satisfactory accuracy with large models and large-scale datasets. The rationale behind this method involves incorporating noise into the classifier's training and inference phases, which effectively smooths the decision boundaries and mitigates steep transition areas exploited by adversarial examples.

Initially, the randomized smoothing 
defines a base classifier $h$, which assigns an input $x$ to a label $y$  from the label set $\mathcal{Y}$. Then a \textit{smoothed classifier} $g$ is constructed based on 
$h$ by incorporating random noise $\epsilon$, sampled from an application-dependent distribution,  
into 
$x$. 
Mathematically, $g(x) = \argmax_{c\in\mathcal{Y}}\mathbb{P}(h(x + \epsilon) = c)$. Finally, the certified robustness, in the form of \emph{certified radius}, can be derived based on 
$g$. 
Particularly, $p_A, p_B\in [0,1]$, the 
highest and second highest probabilities outputted by $\mathbb{P}(h(x + \epsilon)) $, are first calculated, and the corresponding classes are denoted as $y_A$ and $y_B$. 
Then $g$ consistently has an accurate prediction 
$y_A$ for 
$x$ with the adversarial perturbation $\delta$ falling within the \textit{certified radius} $R$, i.e., $g({x} + \delta) = y_A, \forall ||\delta||_p \leq R$, where $||\cdot||_p$ is an $\ell_p$ norm. $R$ is often determined by the variance of the noise $\epsilon$ and the $p_A, p_B$.  
A larger $R$ implies that $h$ (or $g$) is more certified robust. 
\section{Identity Protection Threat: Identity Leakage}


In this section, we present critical identity leakage vulnerabilities. Adversaries may execute 
an \textit{identity forgery attack}, that uses one image from a target user to create counterfeit watermarked images without decoding the secret watermark and a more potent \textit{identity extraction attack}, that accurately retrieves the secret watermark. These attacks compromise watermark-based image copyright protection and ownership traceability, particularly in invisible watermarks.


\begin{figure}[!t]
\centering
  \begin{subfigure}[b]{0.32\linewidth} 
    \includegraphics[width=\linewidth]{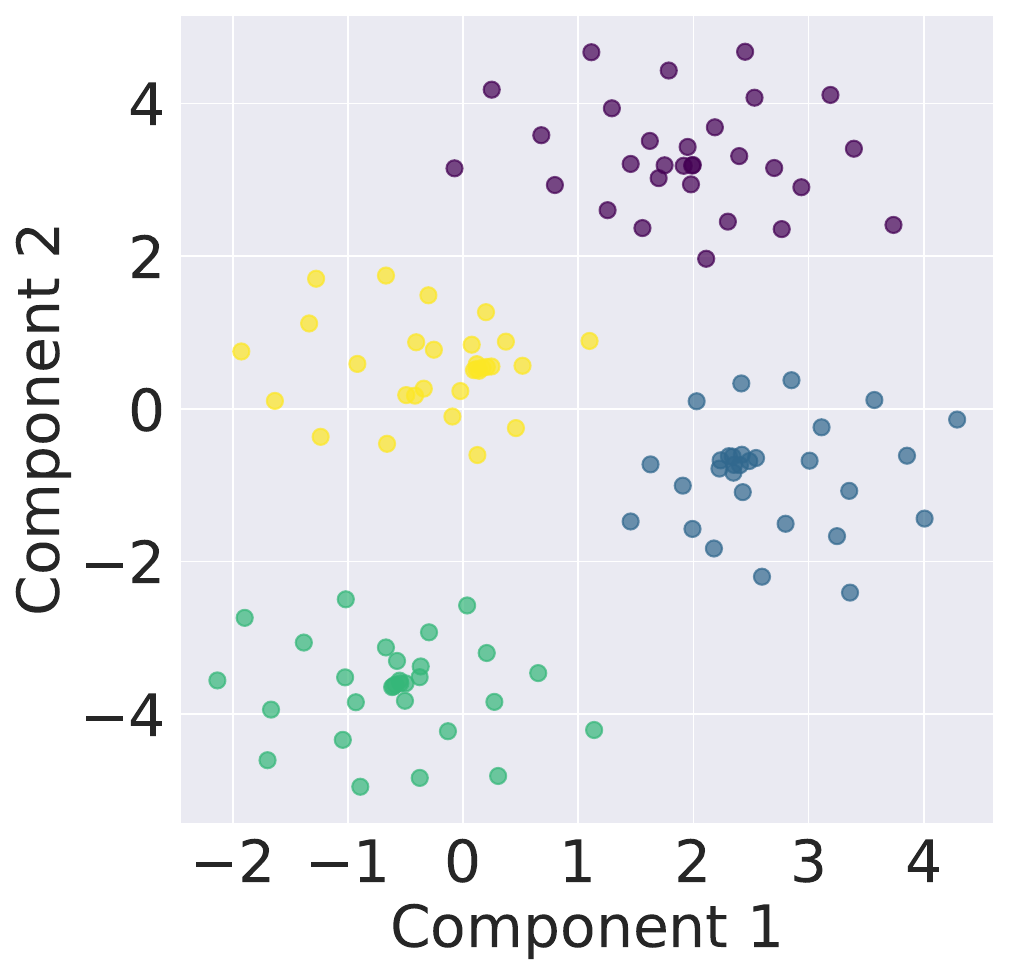}
    \label{fig:clean_coco}
  \end{subfigure}
  \begin{subfigure}[b]{0.32\linewidth}
    \includegraphics[width=\linewidth]{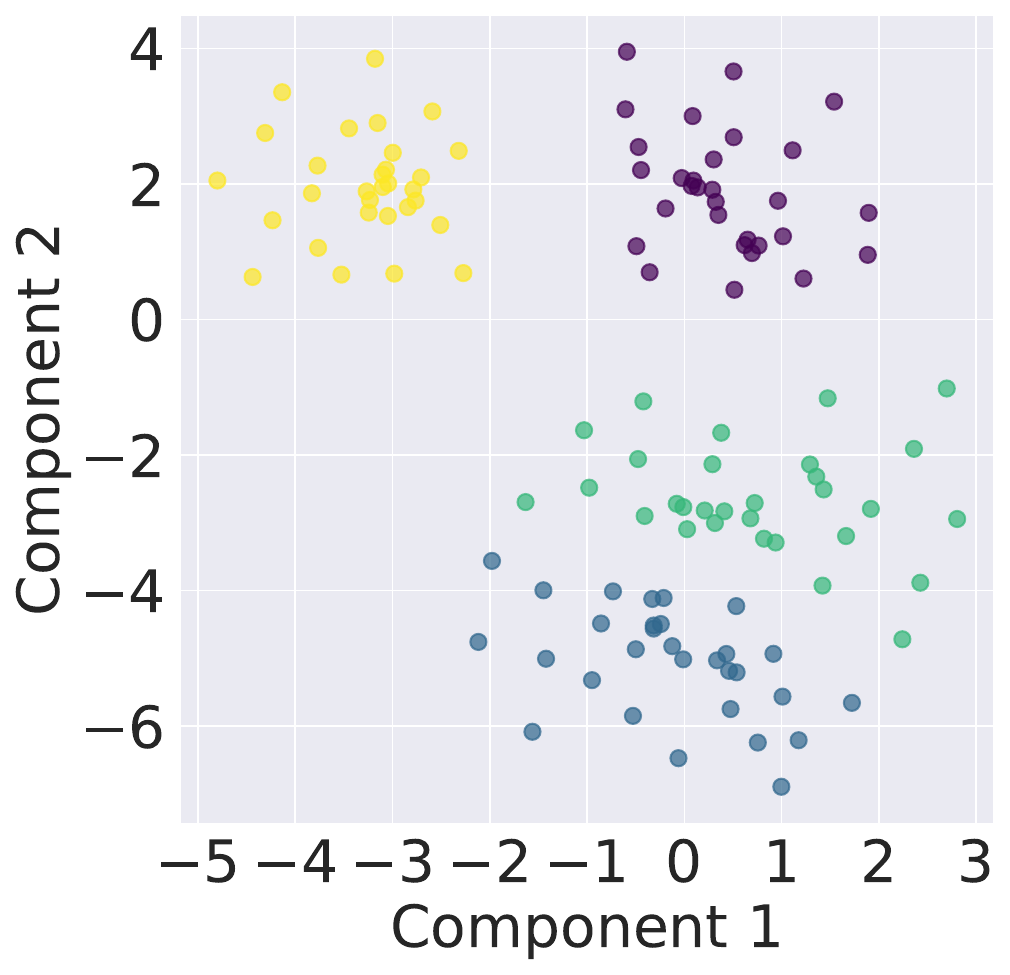}
    \label{fig:clean_celeb}
  \end{subfigure}
  
\vspace{-4mm}
\caption{Cluster results of residual images from four users on HiDDeN. Datasets: left-COCO and right-CelebA.
The intra-cluster distances of the four secret watermarks are (a)-$[0.78,0.83,0.94,0.78]$ and (b)-$[1.17,0.86,0.98,1.13]$.
}
\label{fig:id_link}
\vspace{-4mm}
\end{figure}

\subsection{System Model}
We consider a practical scenario in which a user or entity uses watermarking technology for copyright protection and ownership identification of their images~\cite{ding2021generalized}. The process generally involves two parties: the user (or entity), who seeks to safeguard their images, and a third-party watermarking service provider that exclusively offers black-box watermark embedding and authentication services. 
Initially, the user selects a secret watermark message \( t \) (i.e., $n$-bit binary string) to serve as a unique identifier and submits \( t \) with the clean image \( x \) to the third-party service provider. The provider then embeds \( t \) into \( x \) and returns the watermarked image \( w \) to the user. Upon receiving \( w \), the user can publish \(w\) on the Internet. When authentication for copyright or ownership is required, the user sends both the watermarked image \( w \) and the secret watermark \( t \) to the third-party provider. The provider subsequently returns a result indicating whether the secret matches the watermarked image, with responses categorized as "yes"/ "no". \textit{Note that users typically maintain consistent secret watermarks for practical management purposes, with secret watermark $t$ serving as their identity representation.} Additionally, we evaluate scenarios involving secret watermark modifications to ensure comprehensive analysis.



\subsection{Threat Model} \label{sec:threat_model_new}
We consider an adversary aiming to forge watermarked images and extract the target user's identity using only a single image. While we focus on the challenging single-image scenario, using multiple images from the same target user can improve attack success rates.
The \textbf{adversary's goals} are twofold: 
1) to perform an identity forgery attack by crafting watermarked images that appear to carry the user's secret watermark; and 2) to execute an identity extraction attack to recover user's secret watermark to produce more accurate forged watermarked images.
Both attacks fundamentally undermine watermark authentication by maliciously misattributing images to users who did not generate them.
The adversary possesses the same black-box encoder access \textbf{capability} as benign users but lacks knowledge of users' secret watermarks.
We assume the adversary can access the watermarked images from users, each with their unique secret watermark, e.g., downloaded from the Internet. 
The adversary can obtain the corresponding 
residual images $z$, 
using watermark removal techniques such as VAE~\cite{rombach2022high} and low-pass filters\cite{kutter2000watermark}. 



For defense, we consider the third-party watermarking provider as the defender with legitimate ownership of watermarking models~\cite{jiang2023evading}. The \textbf{defender's capabilities} encompass the training of encoder-decoder networks, for which we propose an identity protection-enhanced training strategy (see Section~\ref{sec:ril}) to mitigate user identity leakage.


\subsection{Identity Leakage Attacks} \label{sec:leakage_clean} 

\begin{figure}[!t]
  \centering
  \includegraphics[width=\linewidth]{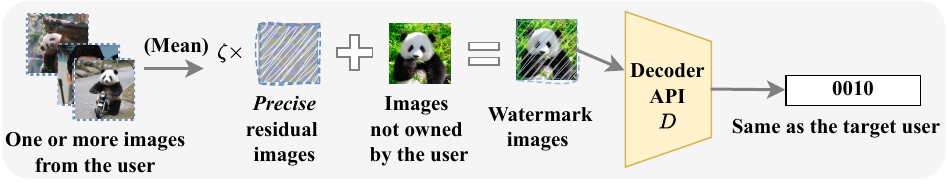}
  \vspace{-4mm}
  \caption{Identity forgery attacks. }
  \vspace{-4mm}
  \label{fig:id_attack}
\end{figure}

\noindent\textbf{Observation.}
We observe that residual images from the same user exhibit consistent watermark patterns. To demonstrate this, we extract salient features from
residual images from four users using $t$-SNE~\cite{van2008visualizing} dimensionality reduction, and subsequently apply $k$-means clustering~\cite{krishna1999genetic} to identify image clusters that share similar watermark characteristics. 
As shown in Figure~\ref{fig:id_link}, images from the same user form coherent clusters, indicating that the residual images preserve stable and user-specific watermark patterns rather than mere random noise. This reveals that substantial watermark-related information is exposed in the residual domain, which we later exploit to mount practical attacks.

\vspace{0.05in}
\noindent\textbf{Identity Forgery Attacks.} 
Given a set of $m$ ($m\geq 1$) images from the target user, 
the adversary can launch identity forgery attacks. Since residual images contain both secret watermarks and image content, 
directly overlaying a single residual image onto a clean image can already produce a forged watermarked image, but often at the cost of noticeable visual artifacts. 
Drawing from previous collusion attacks~\cite{doerr2005collusion}, averaging multiple images embedded with the same secret watermark 
can estimate and remove the watermark from the image. Thus, adversaries can further neutralize image content by averaging residual images from the same user, resulting in a more \textit{precise} residual image. As depicted in Figure~\ref{fig:id_attack}, the adversary overlays $\zeta$ times 
of \textit{precise} residual images on an original image to obtain the forged watermarked image. Figure~\ref{fig:id_forgery} illustrates the effect of overlaying varying numbers of residual images on a clean original image and the corresponding bit accuracy of watermark decoding and authentication. 
The results imply that even a single residual image retains the primary secret watermark and allows the decoder network to correctly decode much of the embedded message. Moreover, as the number $m$ of averaged residual images increases, the visual quality of the forged images improves.



\begin{figure}[!t]
  \centering
  \begin{subfigure}[b]{0.23\columnwidth}
        \centering
        \includegraphics[width=\textwidth]{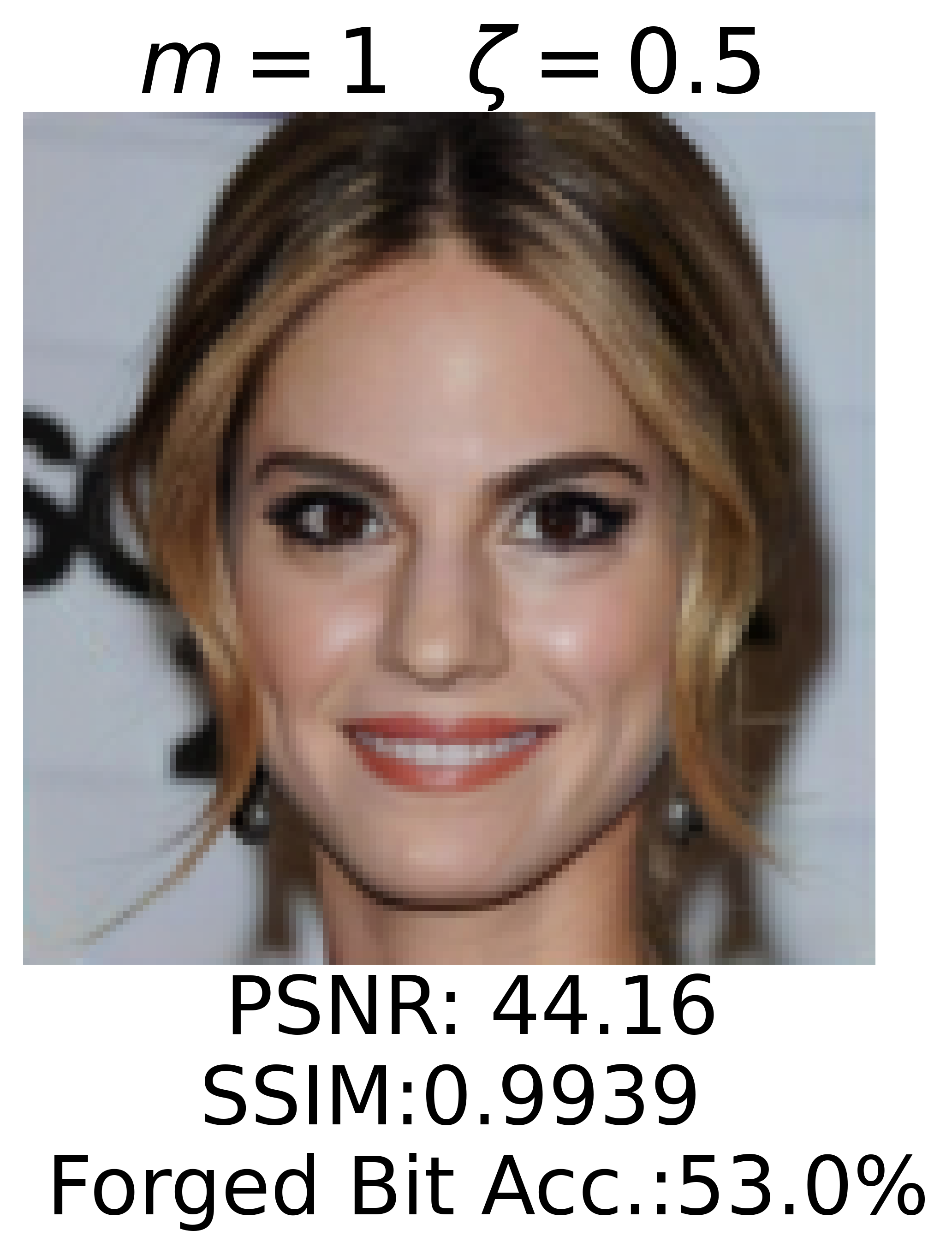}
    \end{subfigure}
    \begin{subfigure}[b]{0.23\columnwidth}
        \centering
        \includegraphics[width=\linewidth]{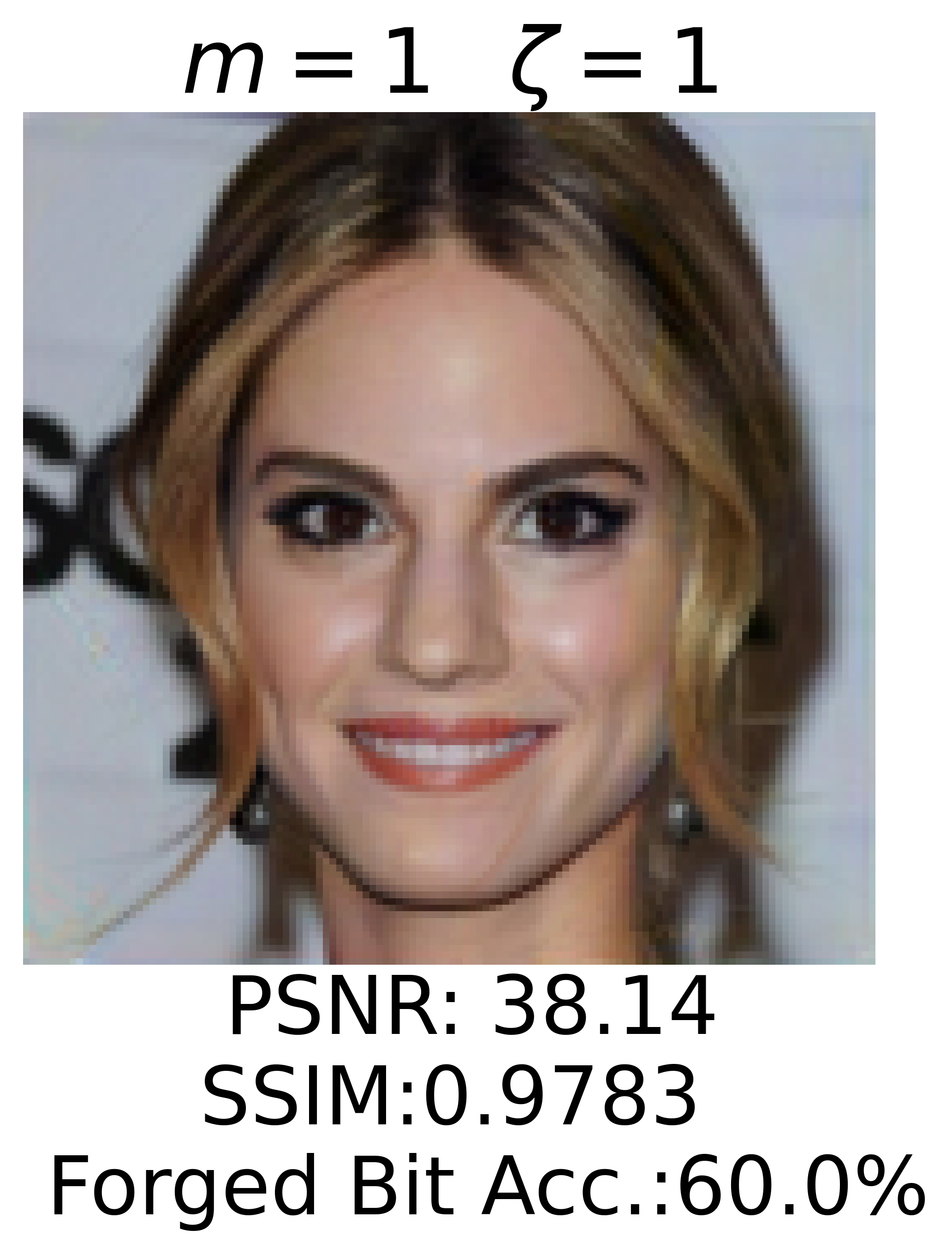}
    \end{subfigure}
    \begin{subfigure}[b]{0.23\columnwidth}
        \centering
        \includegraphics[width=\textwidth]{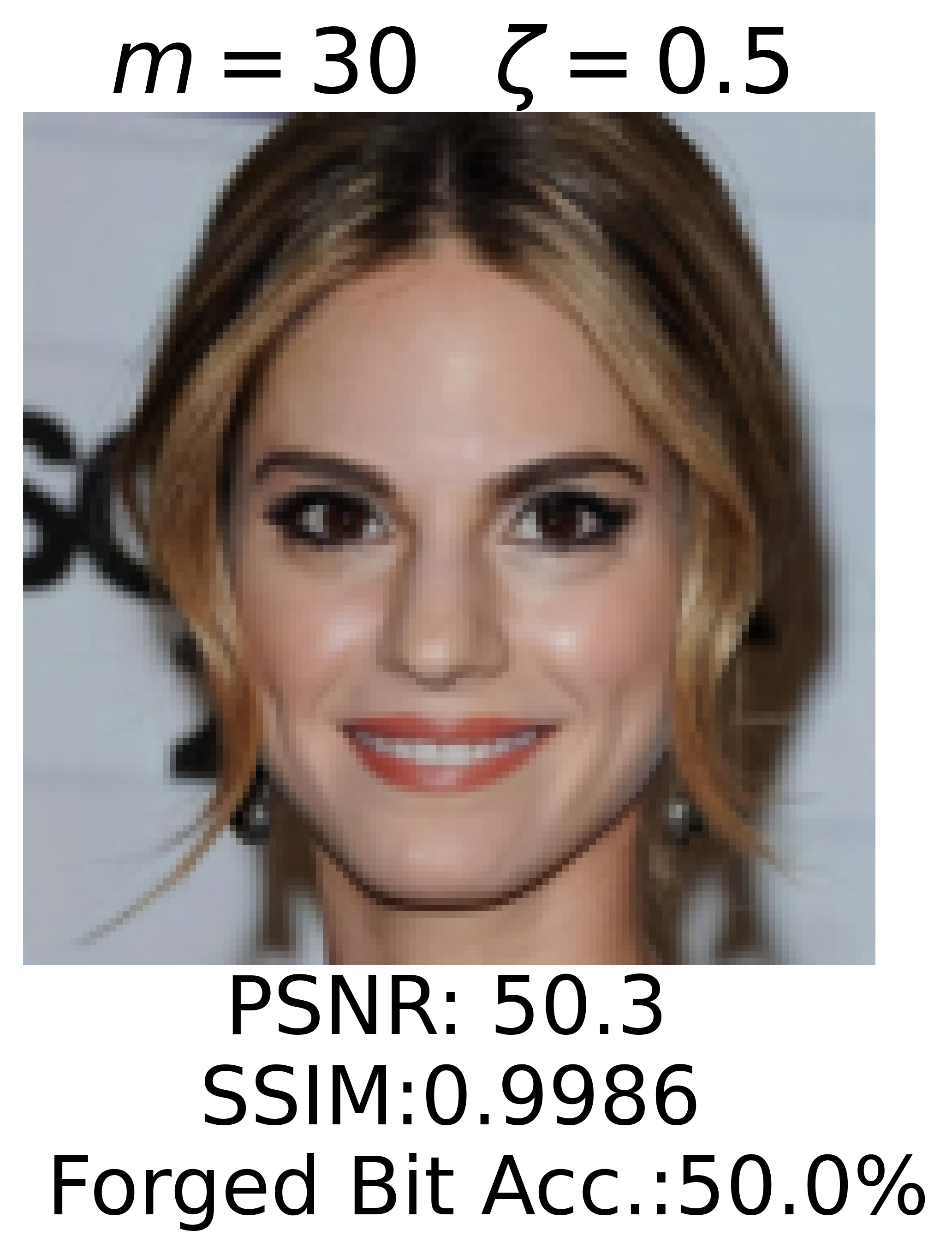}
    \end{subfigure}
    \begin{subfigure}[b]{0.23\columnwidth}
        \centering
        \includegraphics[width=\linewidth]{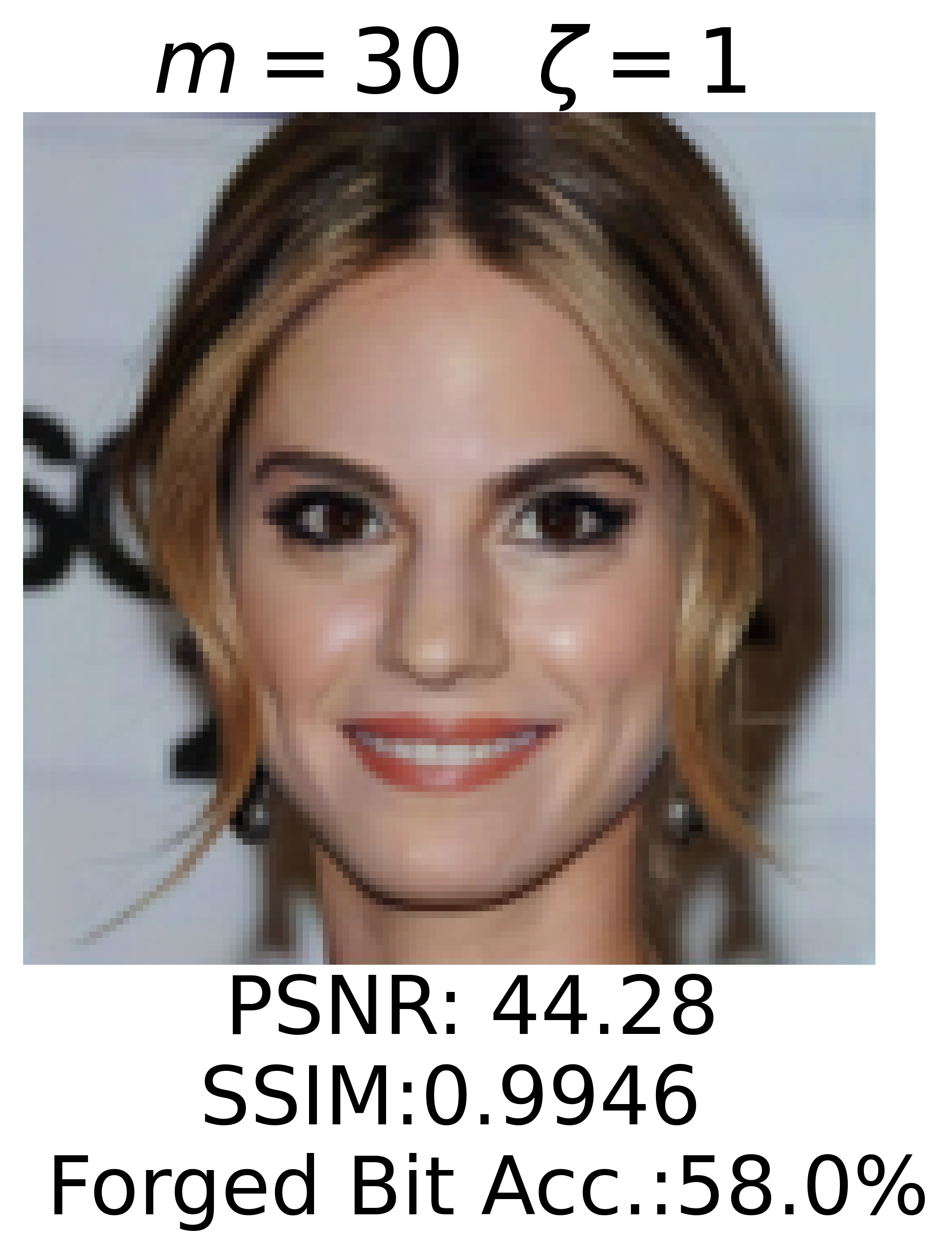}
    \end{subfigure} 
\vspace{-2mm}
  \caption{Illustration of forging watermarked images and attack effect on StegaStamp. PSNR(↑) and SSIM(↑) represent the image quality. Forged bit accuracy represents the accuracy of secret watermark decoding from the forgery image. 
}
\vspace{-4mm}
  \label{fig:id_forgery}
\end{figure}

\begin{algorithm}[!b]
\small
\caption{Identity Extraction Algorithm}  
\begin{algorithmic}[1]
\Require Target user's watermarked image $w_\texttt{tg}$ and its residual image $z_\texttt{tg}$, clean image $x$, secret watermark length $n$, encoder $E$.
\State $t\gets$ a bit string of length $n$ initialized to all `$0$' 
\For {$\iota = 1$ to $n$}  \Comment{Iterate over each bit}
    \State Set $t' \gets t$, and $t'[\iota] \gets 1$
    \State $w_{t}\gets E(x, t)$, $w_{t'}\gets E(x, t')$ 
        \Comment{Watermarked image}
    \State $z_t \gets w_{t} -x$, $z_{t'} \gets w_{t'} - x$ 
        \Comment{Residual image}
    \State $\textrm{dist}_{t} \gets \|z_t,\ z_\texttt{tg}\|,\ \textrm{dist}_{t'} \gets \|z_{t'},\ z_\texttt{tg}\|$
        \Comment{Euclidean distance}
    \If{$\textrm{dist}_{t'} < \textrm{dist}_{t}$} 
    \State $t \gets t'$ \Comment{Update $t[\iota]\gets 1$}
    \EndIf
\EndFor
\State \textbf{Output:} Extracted secret watermark message $t$.
\end{algorithmic}
\label{alg:identity_detection}
\end{algorithm}

\vspace{0.05in}
\noindent\textbf{Identity Extraction Attacks.}
Given one watermarked image from the target user, the adversary can illicitly decode a $n$-bit secret watermark by accessing the watermark embedding services fewer than $2\times n$ times. This attack arises from the inadvertent inclusion of watermark information in the residual image; the closer the secret watermark is to a given image, the more similar the corresponding residual images become. For instance, when applying $100$ secret different watermarks to the same image, the Pearson correlation coefficient ($\in[-1,1]$)~\cite{cohen2009pearson} between the secret watermark distance and the residual image distance series may reach as high as $0.95$ on StegaStamp (CelebA), indicating a strong positive correlation between these two sequences. 
As illustrated in Figure~\ref{fig:id_extract} and Algorithm~\ref{alg:identity_detection}, to determine whether the $\iota$-th bit of the target user’s watermark is $0$ or $1$, the adversary sends two pairs, $(x, t)$ and $(x, t')$, to the watermark service provider, where $x$ is a clean image, and $t$ and $t'$ differ only in the $\iota$-th bit. After obtaining the residual images $z_t$ and $z_{t'}$ for their own watermarked image, the adversary can decide whether the $\iota$-th bit of the target user’s watermark corresponds to $t$ or $t'$ by evaluating which of the residual images $z_t$ or $z_{t'}$ is closer to the target user’s residual image.

\begin{figure}[!t]
  \centering
  \includegraphics[width=\linewidth]{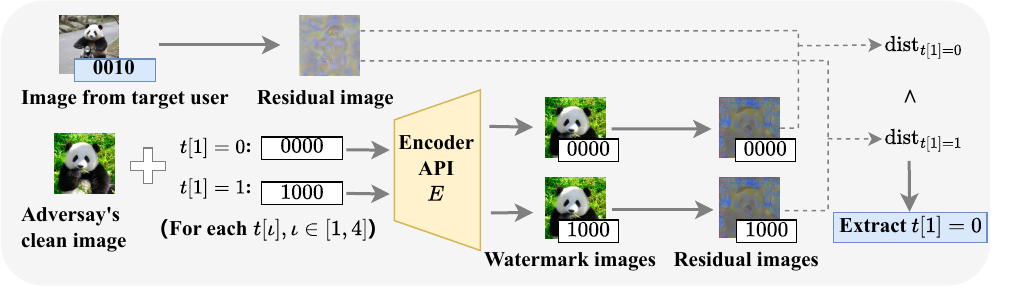}
  \vspace{-4mm}
  \caption{Identity extraction attack.}
  \vspace{-4mm}
  \label{fig:id_extract}
\end{figure}

\section{Robustness Exacerbates Identity Leakage} \label{sec:leakage_distort}

In this section, we first present two robust watermarking approaches: an empirically robust method (see Section~\ref{sec:W-ER}) and a certified robust method (see Section~\ref{sec:W-CR}), both designed to defend against adversarial attacks. Next, we discuss how robust watermarking may inadvertently exacerbate identity leakage (see Section~\ref{sec:id_leakage_robust}).

\subsection{Empirically Robust Watermarking}
\label{sec:W-ER}
These methods typically integrate adversarial training~\cite{zhu2018hidden,stegastamp,CIN,jia2021mbrs} into the encoder and decoder's training phases to defend against adversarial attacks. We use two representative open-sourced robust watermarking, referred to as W-ER: 1) StegaStamp$^+$\footnote{We use \textit{StegaStamp} and \textit{HiDDeN} to identify clean watermarking methods that do not incorporate empirically robustness training. We use \textit{StegaStamp$^+$} and \textit{HiDDeN$^+$} to signify watermarking methods that include empirically robustness training, as detailed in \cite{stegastamp} and \cite{zhu2018hidden}.}\cite{stegastamp} improves model robustness through a perturbation layer including various distortions (e.g., perspective warp, Gaussian noise, JPEG compression) to approximate distortions in real printing and photography. 2) HiDDeN$^+$\cite{zhu2018hidden} enhances invisibility via a discriminator and incorporates a noise simulation layer (e.g., Gaussian noise, JPEG compression) to bolster robustness against real-world perturbations.



\subsection{Certified Robust Watermarking}
\label{sec:W-CR}

In this section, 
we first 
introduce the notations used in the paper, followed by the formulation of 
image watermark authentication as a classification model. We then formalize adversarial watermarks as a transformation of image pixels or/and coordinates. 
Finally, we introduce W-CR as the solution to certify watermarking against such adversarial watermarks.

\subsubsection{Watermark Authenticity Formulation} 

\begin{table}[]
\setlength\tabcolsep{2pt}
\centering
\footnotesize
\caption{Frequently used notations}
\vspace{-2mm}
\begin{tabular}{@{}ll@{}}
\toprule
Term      & Description               \\
\midrule
$\CX\subseteq \BR^d$ & Original image space \\
$\CT\subseteq \{0,1\}^{n}$ & Secret watermark message space \\
$\CW\subseteq \BR^d $ & Watermarked image space \\
$\CZ\subseteq \BR^d$ & Residual image space \\
$\CV\subseteq \BR^{2d}$ & Coordinate space \\
$\CW\oplus\CV \subseteq \BR^d$ & Pixel space situated at coordinate space \\
$\tau\subseteq \BR$ & Verification (tolerance) threshold \\
$D: \CW\oplus\CV \to \CT$ & The decoder model \\
$f_V: \CT \times \CT \times \BR \to \{0,1\}$ & The watermark verification function \\
$h\!:\!(\CW\oplus\CV) \times \CT \times \BR \!\to\! \{0,1\}$ & The watermark authentication model \\
$\phi(w, u): \CW \times \CU$ & Pixel trans. with parameter $u$ on $w$ \\ 
$\theta(v, r): \CV \times \CR$ & Coordinate trans. with parameter $r$ on $v$ \\ 
$f_I: \CW\times \CV\to \CW$ & The interpolation function \\
$\langle i, j \rangle$ & $x$-axis and $y$-axis of coordinates \\
$n$ & Length of the secret watermark message \\
$d$ & Size of original (watermarked) image \\
$\delta$ & Perturbations of pixel/coordinate space \\
$R$ & Certified radius \\

\bottomrule
\end{tabular}
\label{tab:notation}
 \vspace{-4mm}
\end{table}

\textbf{Notations.}
We denote the set of all possible original images as $\CX$, the set of all possible secret watermark messages as $\CT$, and the set of all watermarked images as $\mathcal{W}$. For each original image $x \in \CX$ and secret watermark $t \in \CT$, there exists a corresponding watermarked image $w \in \mathcal{W}$. 
We use $\mathcal{V} = \{v=\langle i, j \rangle\}$ to denote the image's coordinate space, where $i$ and $j$ represent the coordinates along the $x$-axis and $y$-axis of a pixel in the image. 
Then, a coordinated watermarked image can be defined as $w\oplus v$, which means the pixels $w$ positioned at coordinates $v$\footnote{For ease of description, we will interchangeably use $w$ and $w\oplus v$ to denote a watermarked image.}. We introduce a decoder model $D: \CW\oplus\CV \to \CT$, where for any watermarked image $w$, the decoder extracts an estimated secret watermark denoted as $t_D = D(w\oplus v)$.  
We denote the verification function as $f_V: \CT \times \CT \times \BR \to \{0,1\}$. 
The function $V(t_D, t, \tau)$ evaluates the match between the decoded secret watermark $t_D$ and the original secret watermark $t$. This function outputs $1$ (indicating a verification success) if the correspondence between $t_D$ and $t$ meets or exceeds the threshold $\tau$, and $0$ otherwise, signifying a verification failure. We formalize the watermark authentication phase as follows.


\begin{definition}[Watermark Authentication]
\label{def:watermark_det}
We define the watermark authentication model as $h: (\CW\oplus\CV) \times \CT \times \BR \to \{0,1\}$, which combines both the decoding and verification steps. Specifically, it first applies the decoder $D$ to a watermarked image $w$, to extract a decoded message. The decoded message is then assessed in conjunction with the original secret watermark $t$ and a predefined tolerance threshold $\tau$ via the verification function $V$. Consequently, the model $h$ evaluates the congruence between the decoded message and $t$, factoring in the threshold $\tau$, denoted as below: 
{
\begin{equation}
    h(w\oplus v, t, \tau) = f_V(D(w\oplus v), t, \tau), 
\label{eq:auth_def}
\end{equation}
}

\vspace{-1mm}

\par
\noindent where a value of $1$ signifies a successful match, and a value of $0$ indicates a mismatch.

\end{definition}

Note that the input image space of $h$ ($\CW\oplus\CV \subseteq \BR^d$) is consistent with $\CW$, and the training process of $h$ (including decoder $D$ and verification function $f_V$) aligns with the classic training of the encoder-decoder watermarking framework. Table~\ref{tab:notation} shows the commonly used notations.



\vspace{0.05in}
\subsubsection{Adversarial Watermarking Formulation}
\label{sec:perturbation}
Recall that a watermarked image comprises image pixels $w$ and coordinates $v$. Perturbing images can be the pixel-level ($\phi$) and coordinate-level ($\theta$) modifications. Considering the pixel of a given image ($w$), the effect of a perturbation on transformation is mathematically denoted as:
\begin{equation}
    w'\oplus v = \phi(w)
\end{equation}


In this scenario, without alterations to the coordinates, the resultant image pertains solely to the pixel values. Conversely, image coordinate transformation 
produces new pixel positions ($v'$), potentially not aligning precisely with the original image's coordinates. We introduce the interpolation function $f_I(w, v')$~\cite{perez20223deformrs, alfarra2022deformrs}, which is commonly employed to estimate the pixel values based on the neighboring pixels. This coordinate transformation is expressed as:
{
\begin{equation}
    w\oplus v'= f_I(w, \theta(v))
\end{equation}
}
%

\noindent Next, we present examples depicted in Figure \ref{fig:noise} and offer precise definitions for $\phi$ and $\theta$ corresponding to the perturbations. 

\begin{figure}[t]
  \centering
  \begin{subfigure}[b]{0.15\columnwidth}
        \centering
        \includegraphics[width=\textwidth]{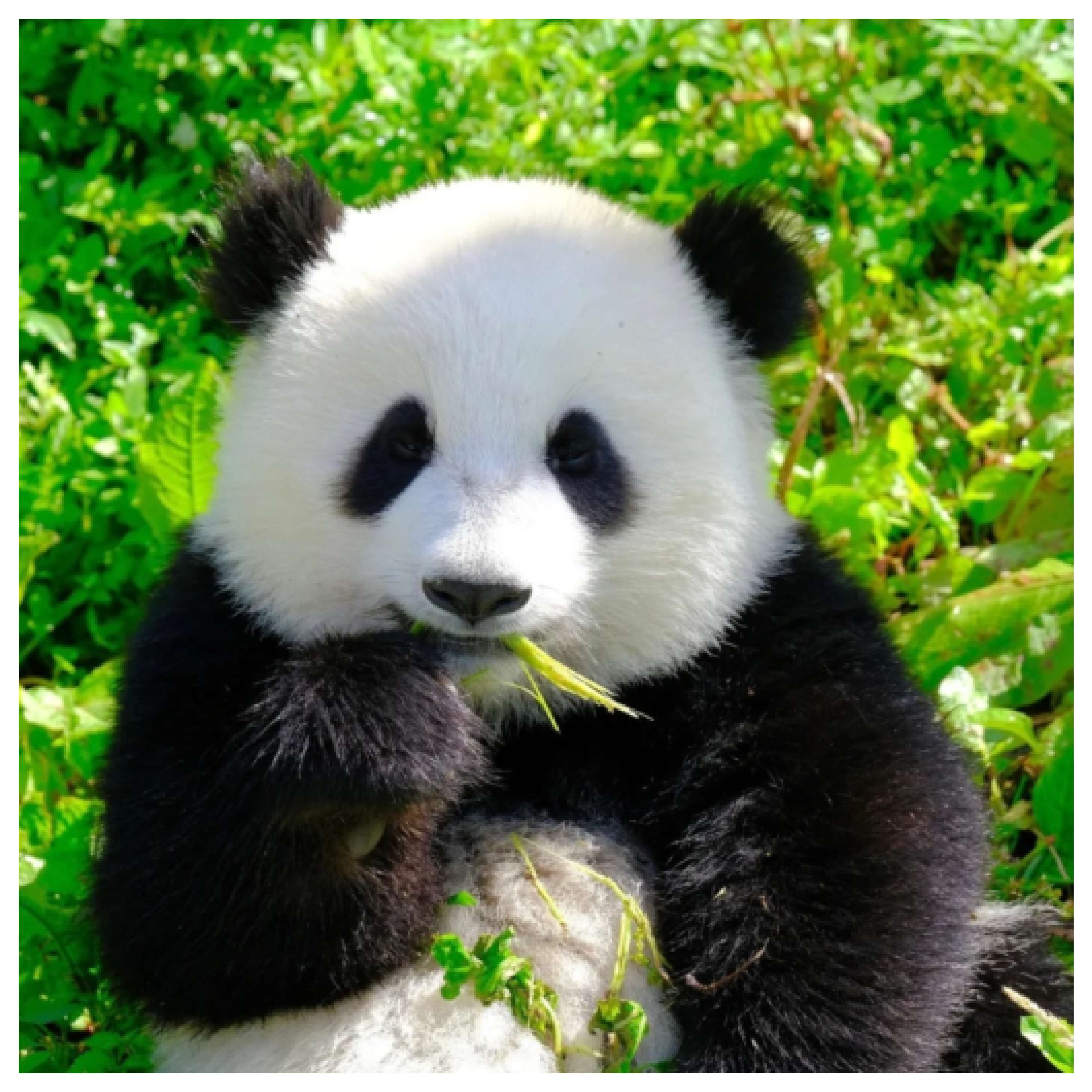}
    \end{subfigure} 
\begin{subfigure}[b]{0.15\columnwidth}
        \centering
        \includegraphics[width=\textwidth]{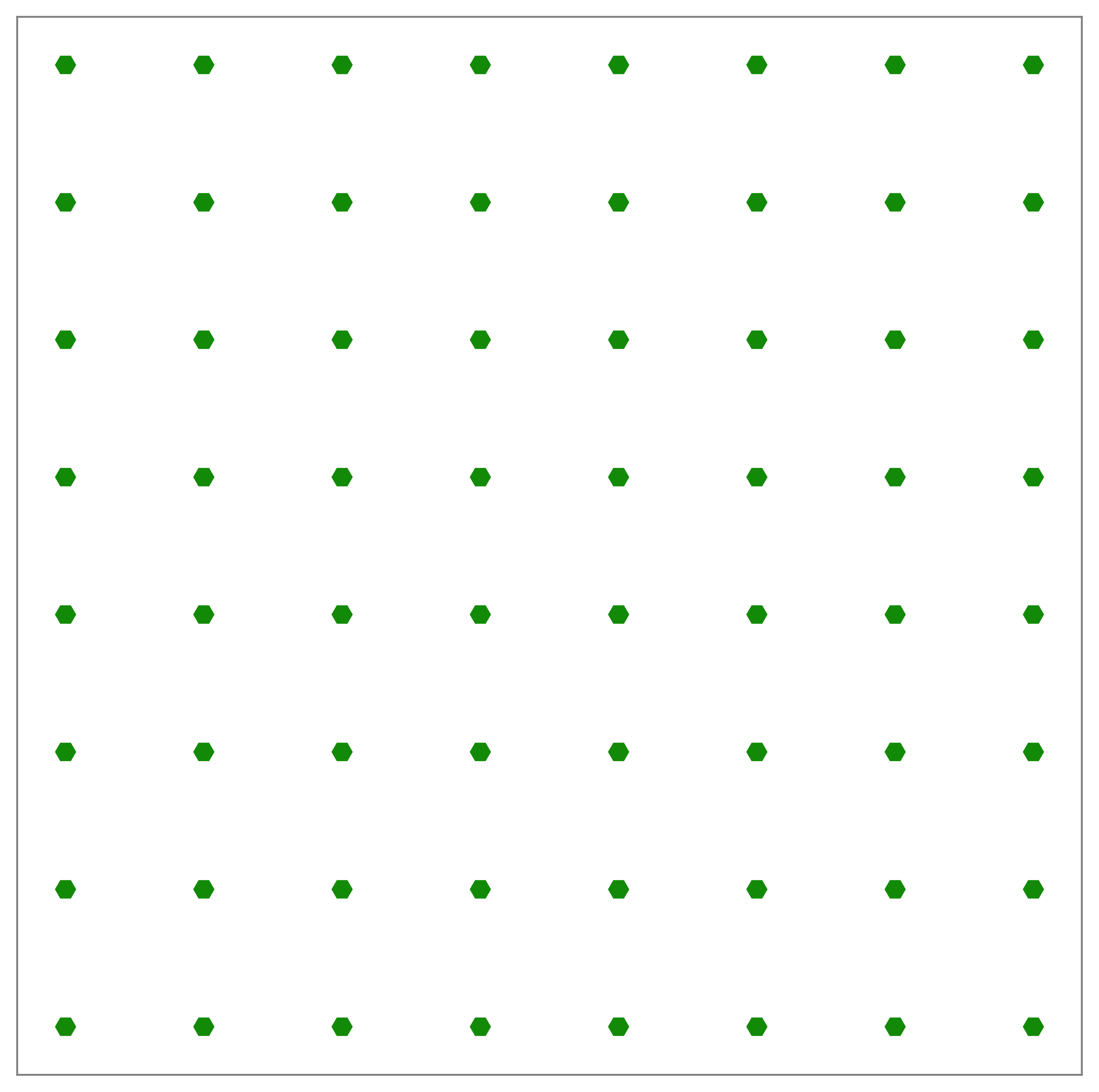}
    \end{subfigure}  
    \begin{subfigure}[b]{0.15\columnwidth}
        \centering
        \includegraphics[width=\linewidth]{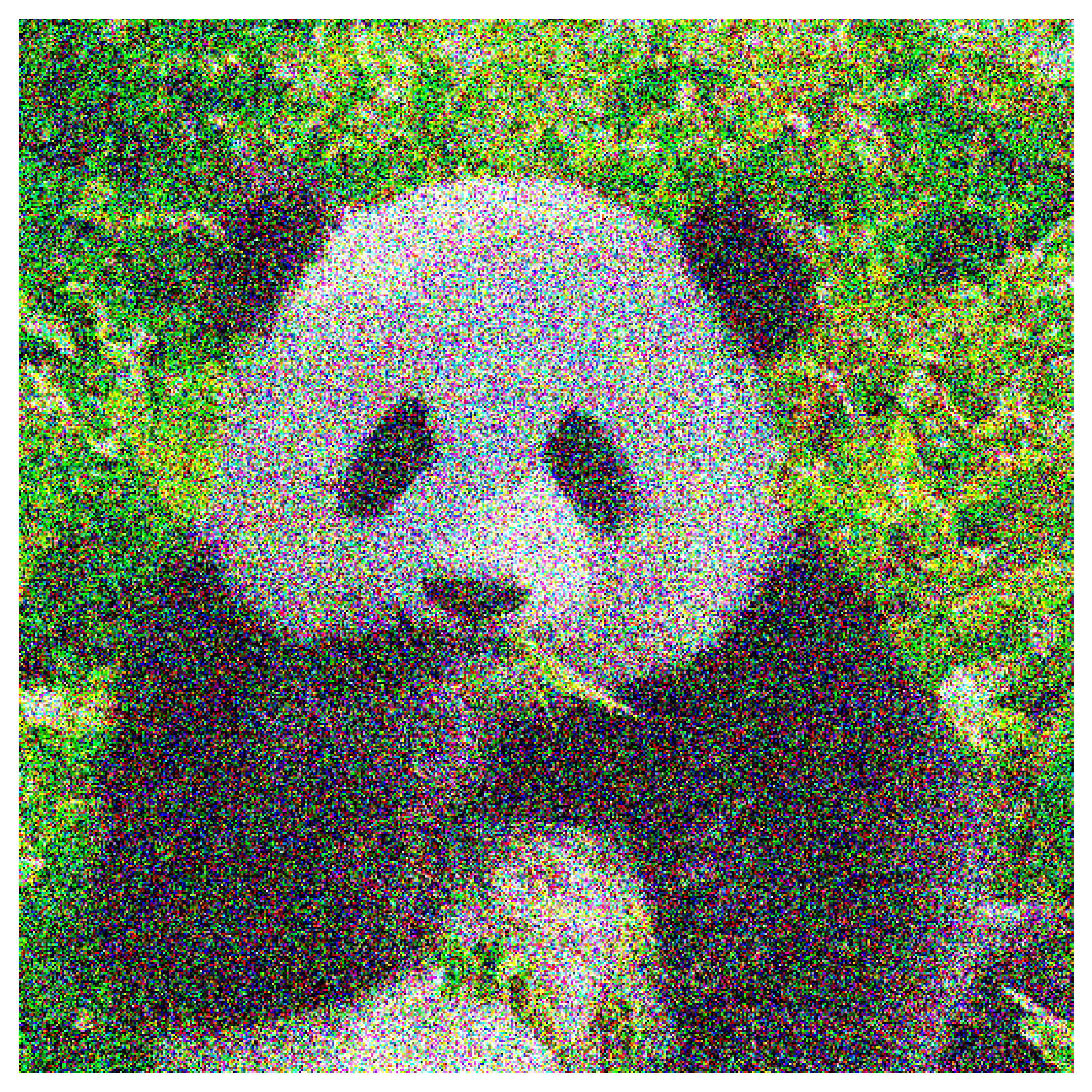}
    \end{subfigure}  
\begin{subfigure}[b]{0.15\columnwidth}
        \centering
        \includegraphics[width=\linewidth]{figures/rule-based-pic/GN_vfield.png}
    \end{subfigure}  
    \begin{subfigure}[b]{0.15\columnwidth}
        \centering
        \includegraphics[width=\linewidth]{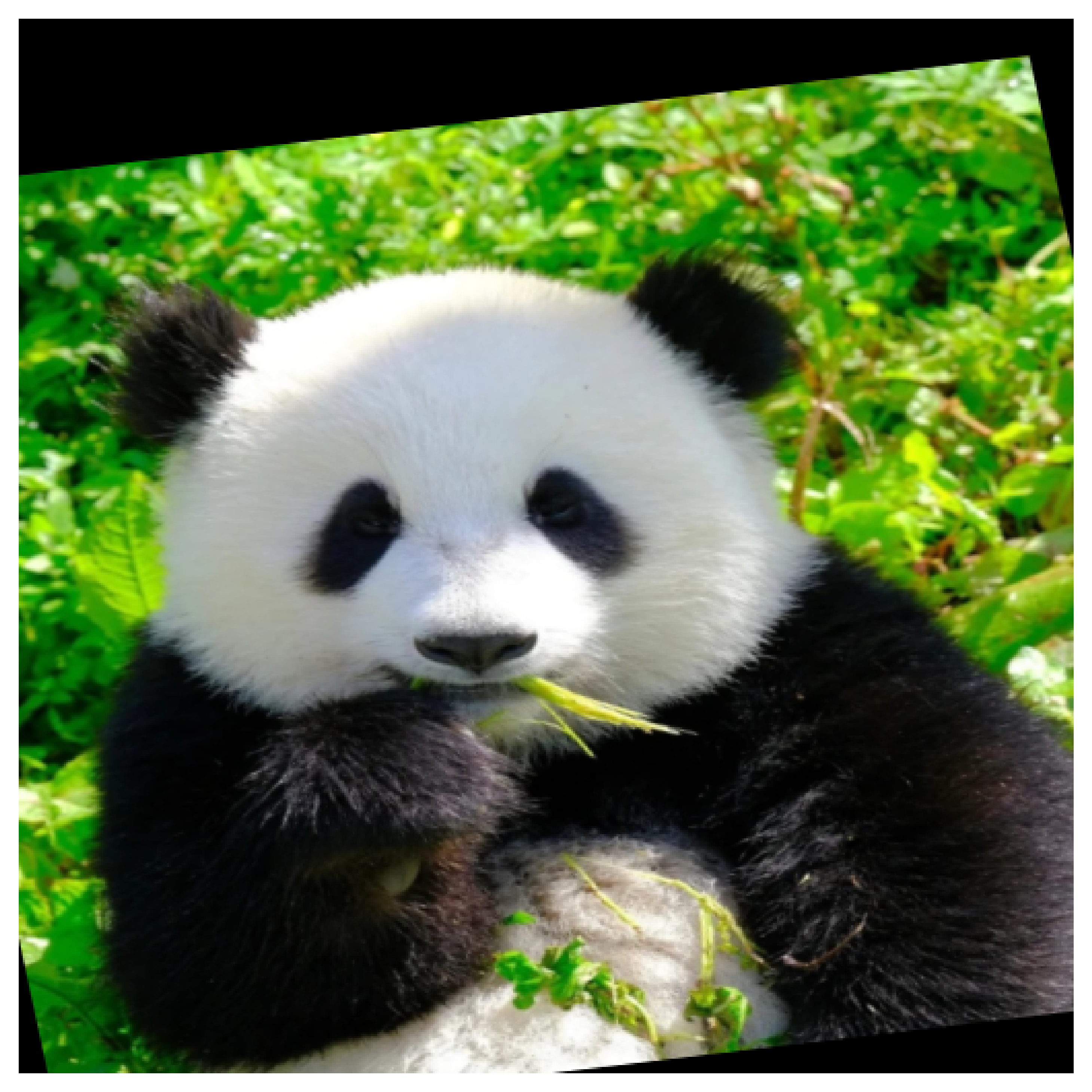}
    \end{subfigure} 
    \begin{subfigure}[b]{0.15\columnwidth}
        \centering
        \includegraphics[width=\linewidth]{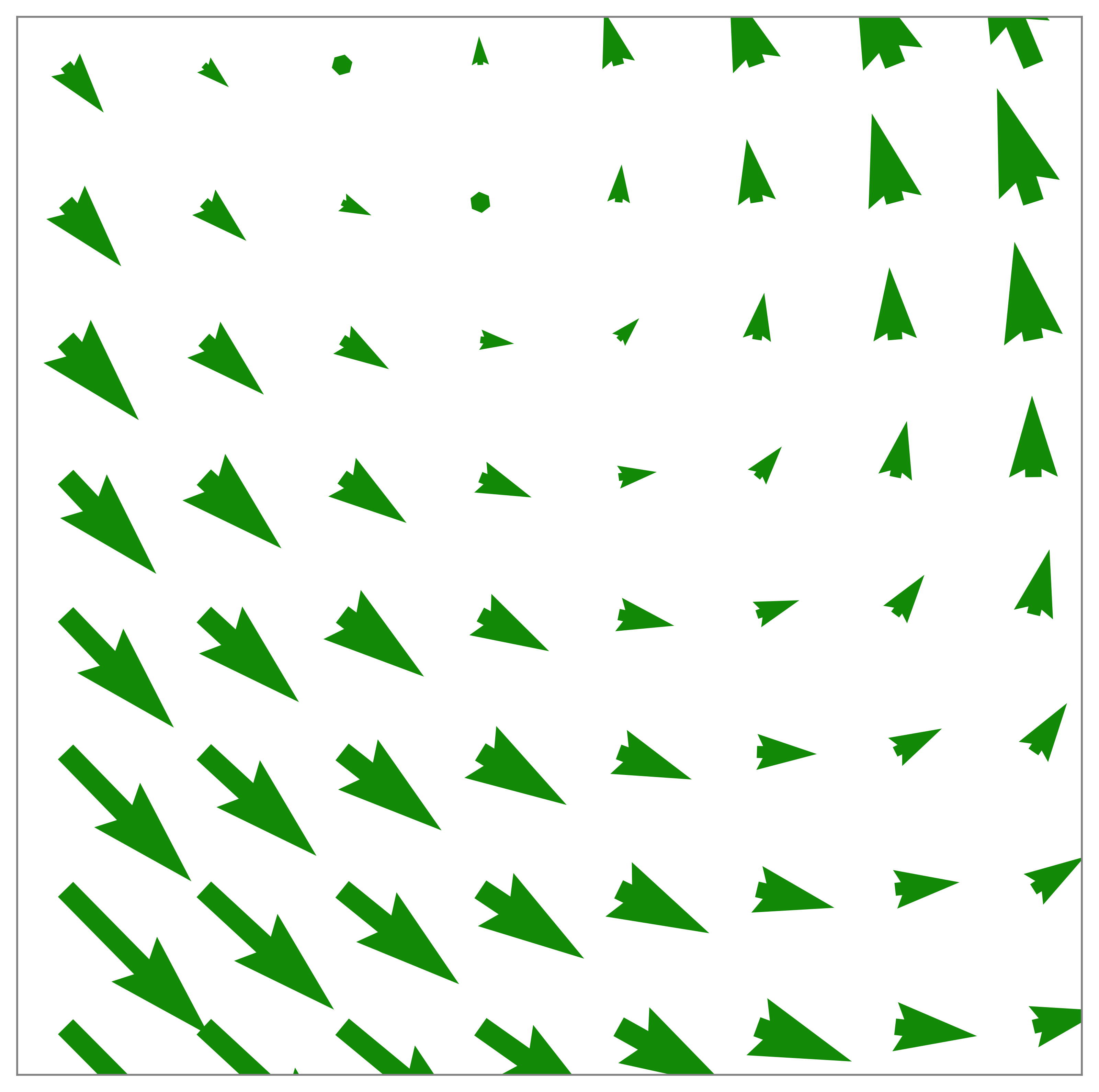}
    \end{subfigure} 
    \vspace{-2mm}
  \caption{Ruled-based perturbations. Left to right: original image and coordinates, Gaussian noisy image and coordinates, and affine transformed image and coordinates. 
  } 
  \vspace{-4mm}
  \label{fig:noise}
\end{figure}

\vspace{0.05in}

\noindent\textbf{Gaussian Noise.}
It adds Gaussian noise as the perturbation 
to image pixels $w$, denoted as 
{
\begin{equation}
\phi_G(w, \sigma) = w + \alpha,\ \alpha\sim \CN(\mathbf{0}, \sigma^2 \BSI)
\label{eq:trans_gaussian}
\end{equation}
}
\noindent where $\alpha\sim \CN(\mathbf{0}, \sigma^2 \BSI)$ denotes Gaussian distribution with mean $0$ and variance $\sigma$. The process generates and inject the noise for each pixel, which adhere to $\CN(\mathbf{0}, \sigma^2 \BSI)$.




\vspace{0.05in}

\noindent\textbf{Affine Transformation.}
It relocates image coordinates $v$ to $v'$ via a combination of linear transformations (such as scaling, rotation, shearing, and translation). This can be mathematically represented as $\theta_F( v, \BSBETA)$ as follows
{
\begin{equation}
\theta_F(v, \BSBETA) 
= \left(\begin{array}{c}\beta_1 i+\beta_2 j+\beta_3\\ \beta_4 i+\beta_5 j+\beta_6\end{array}\right),\ \BSBETA\sim \CN(\mathbf{0}, \sigma^2 \BSI_6)
\label{eq:trans_affine}
\end{equation}
}
where {
$\CN(\mathbf{0}, \sigma^2 \BSI_6)$ denotes that each component in $\BSBETA$ is independently distributed according to a Gaussian distribution with a mean of $0$ and a variance of $\sigma$. 


We consider modifications within either the pixel or coordinate space, as each of them represents a typical adversarial perturbation on watermarking in practice\cite{jiang2023evading,an2024benchmarking}.



\subsubsection{W-CR Framework}
\label{sec:W-CR-robust}

\textbf{\Rmnum{1}. Overview and Certification Goal.}
By framing watermark authentication as a binary classification model, we can apply the certified robustness framework from image classification. This involves transforming perturbations to watermarked images into pixel and coordinate spaces (see Section~\ref{sec:perturbation}) to ensure certified robustness for various operations in each space (see Section~\ref{sec:distort_framework}).

Since operations that perturbed watermarked images affect either pixel space or coordinate space, any adversary is essentially adjusting pixel transformation $u$ or coordinate transformation $r$.  Our goal is to ensure the watermark authentication model's robustness against attacks defined by $u$ or $r$. Specifically, for pixel or coordinate transformation, we aim to identify a set of pixel parameters $S_{adv}^u \subseteq \CU$ or coordinate parameters $S_{adv}^r \subseteq \CR$, such that the model $h$'s prediction remains consistent for any $u$ or $r$:
\par

\vspace{-4mm}

{
\begin{align}
    h(w\oplus v, t, \tau)&=h(\phi(w,u), t, \tau) &\forall u\in S_{adv}^u \\
    h(w\oplus v, t, \tau)&=h(f_I(w,\theta(v, r)), t, \tau) &\forall r\in S_{adv}^r
\end{align}
}


\par
\noindent\textbf{\Rmnum{2}. Certified Robustness.}
\label{sec:distort_framework}
By leveraging pixel or coordinate transformations, we construct a new type of transformation smoothed watermark authentication model $g$ from an arbitrary base watermark authentication model $h$. Specifically, this transformation-smoothed model $g$ is designed to predict the class with the highest probability returned by $h$ when the input for watermark authentication, i.e., $w \oplus v$, is perturbed via the pixel or coordinate transformation. The definitions of pixel and coordinate transformation smoothed watermark authentication models are shown below.


\begin{definition}[Pixel Transformation Smoothed Watermark Authentication]
\label{def:smooth_watermark_w}
Let $\phi\!:\! \CW\!\times\! \CU\!\to\! \CW$ be a pixel transformation, and let $h\!:\! (\CW\!\oplus\!\CV)\!\times\!\CT\!\times\!\BR \!\to\! \{0,1\}$ be an arbitrary base watermark authentication model. Taking random variables $\epsilon\!\sim\! \BP_\epsilon$ from $\CU$, we define pixel transformation smoothed watermark authentication model $g_\phi\!:\!(\CW\oplus\CV)\times\CT\times\BR \!\to\! \{0,1\}$ as
{
\begin{equation}
\begin{aligned}
    g_\phi(w\oplus v, t, \tau) 
    =\argmax_{y\in \{0,1\}}\BP(h(\phi(w, \epsilon), t, \tau))
\label{eq:smooth_watermark_w}
\end{aligned}
\end{equation}
}
\end{definition}
\begin{definition}[Coordinate Transformation Smoothed Watermark Authentication]
\label{def:smooth_watermark_v}
Let $\theta: \CV\times \CR\to \CV$ be a coordinate transformation, $f_I: \CW\times \CV\to \CW$ be an interpolation function, and let $h: (\CW\oplus\CV)\times\CT\times\BR \to \{0,1\}$ be an arbitrary base watermark authentication model. Taking random variables $\rho\sim \BP_\rho$ from $\CR$, we define coordinate transformation smoothed watermark authentication model $g_\theta: (\CW\oplus\CV)\times\CT\times\BR\to \{0,1\}$ as
{
\begin{equation}
    g_\theta(w\oplus v, t, \tau)=\argmax_{y\in \{0,1\}}\BP(h(f_I(w, \theta(v, \rho)), t, \tau))
\label{eq:smooth_watermark_v}
\end{equation}
}
\end{definition}

To ensure certified robustness for the transformation-smoothed watermark authentication model against the two perturbations, we refer to methods proven for the image classification~\cite{cohen2019certified, perez20223deformrs, alfarra2022deformrs, li2021tss}. 
We adopt common types of noise, such as Gaussian and uniform distributions, to formulate pixel transformations $\phi$ and coordinate transformations $\theta$. We then describe the certification theorems against these perturbations.

\vspace{0.05in}
\noindent {\bf Certified Robustness to  Gaussian Noise.}
 Gaussian noise only changes the image pixel values, without affecting the pixel coordinates. We invoke Theorem 1 (binary case) in \cite{cohen2019certified} and adapt it to our pixel transformation smoothed watermark authentication, as detailed below:

\begin{theorem}
Let $\phi_G: \CW\times\CU \to \CW$ be the pixel transformation based on Gaussian noise $\epsilon\sim\CN(\mathbf{0}, \sigma^2 \BSI)$. Let $g_{\phi_G}$ be the smoothed watermark authentication model from a base watermark authentication model $h$ as in Eq.~\ref{eq:smooth_watermark_w}, and suppose $\underline{p_A}\in (\frac{1}{2}, 1]$ satisfies $\BP(h(\phi(w,\epsilon), t, \tau)) \leq \underline{p_A}$. Then $g_{\phi_G}(\phi_G(w,\delta_G)\oplus v,t,\tau) = c_A$ for all $\|\delta_G\|_2 < R_G$, where
{
\begin{equation}
     R_G = \sigma \Phi^{-1}(\underline{p_A})
\label{eq:gaussian_r}
\end{equation}
}
\label{thm:gaussian}
\end{theorem}

\vspace{-6mm}

\par
Per Eq.~\ref{eq:trans_gaussian}, $\phi_G$ indicates the addition of noise to the pixels. Thus $\delta_g$ represents the distance at the pixel level between the original and perturbed images.
Theorem~\ref{thm:gaussian} states that $g_{\phi_G}$ can defend against adding Gaussian noise perturbations as long as the addition noise $\|\delta_G\|_2<R_G$ in Eq.~\ref{eq:gaussian_r}. We observe that the certified radius $R_G$ is larger when the noise level $\sigma$ is higher and/ or $\underline{p_A}$ is larger.



\vspace{0.05in}
\noindent {\bf Certified Robustness to Affine Transformation.}
Affine transformation only alters the image coordinate with the parameter owning six components (i.e., $\BSBETA=[\beta_1, \cdots, \beta_6]$) applying to the coordinate. 
We invoke Theorem 1 in \cite{alfarra2022deformrs} and adapt it to our (binary) coordinate smoothed watermark authentication, as detailed below:
\begin{theorem}
Let $\theta_F: \CV\times\CR \to \CV$ be the coordinate transformation based on Gaussian noise $\rho\sim\CN(\mathbf{0}, \sigma^2 \BSI_6)$. Let $g_{\theta_F}$ be the smoothed watermark authentication model from a base model $h$ as in Eq.~\ref{eq:smooth_watermark_v}, and suppose $\underline{p_A}\in (\frac{1}{2}, 1]$ satisfies $\BP(h(f_I(\theta_F(v, \rho)), t, \tau)) \leq \underline{p_A}$. Then $g_{\theta_F}(w\oplus\theta_F(v,\delta_F), t, \tau) = c_A$ for all $\|\delta_F\|_2 < R_F$, where
{
\begin{equation}
     R_F = \sigma \Phi^{-1}(\underline{p_A})
\label{eq:affine_r}
\end{equation}
}
\label{thm:affine}
\end{theorem}
\vspace{-6mm}

As indicated by Eq.~\ref{eq:trans_affine}, $\theta_F$ represents the linear transformation of $x$-axis and $y$-axis, i.e., $v=\langle i, j \rangle$, by the parameter list $\BSBETA$. Then $\delta_F$ denotes the distance of the coordinate transformation coefficients, denoted as $\|\delta_F\|_2=\sqrt{\sum_{\kappa} \beta_\kappa^2}$~\cite{alfarra2022deformrs}.
Theorem~\ref{thm:affine} states that $g_{\theta_F}$ can defend against affine perturbations as long as the condition about $\delta_F$ in Eq.~\ref{eq:affine_r} satisfies. We observe that the certified radius $R_F$ is larger when the uniform noise level $\sigma$ is higher and/ or $\underline{p_A}$ is larger.


\subsubsection{Practical Algorithms}

\vspace{0.05in}
\noindent\textbf{Training.} 
Note that, though we combine 
decoding and verification into a single function $h$, 
the verification function is parameter-free. During the training phase, we update only the decoder $D$ (and encoder $E$).
The training process is akin to classical robust watermarking training~\cite{stegastamp,zhu2018hidden}. We introduce perturbations solely within the perturbation layer. Receiving a watermarked image output from the encoder, we obtain the corresponding pixel value $w$ and pixel coordinate $v$. Then we perturb images through transformations applied to either the pixel transformation $\phi(w, u)$ or the coordinate transformation $\theta(v, r)$. Finally, we feed these new images into the decoder to train and update the watermark's encoder and decoder. 

\vspace{0.05in}

\noindent\textbf{Certification.}
The overall process of the certification is similar to the classical randomized smoothing in \cite{cohen2019certified}. As described in Algorithm~\ref{alg:certify}, we first obtain the watermarked image, which comprises image pixels $w$ and coordinate $v$. Then we utilize pixel transformation $\phi_P$ or coordinate transformation $\theta_P$ on the image and draw $N_0$ samples to formulate an initial guess for $y_A$. The input parameters for $h$, i.e., $w, v, t, \tau$, are as specified in Eq.~\ref{eq:auth_def}.
We then employ a larger number of sample sets to calculate an estimate of $\underline{p_A}$. Finally, we certify robustness on $N$ samples and output the robust prediction.
\subsection{Identity Leakage Exacerbation} \label{sec:id_leakage_robust}

\noindent\textbf{Clean v.s. Robust:} 
Both empirically and certified robust watermarking methods
are more susceptible to identity leakage risks than clean models. 
%
As illustrated in Figure~\ref{fig:id_link_robust}, the clustering effect on residual images under robust models (i.e., the effectiveness of identity linking attacks) outperforms that of clean models in Figure~\ref{fig:id_link}, characterized by the reduced intra-cluster distances and increased inter-cluster distances. Similarly, as Figure~\ref{fig:id_forgery_robust} shows, the capacity of residual images under robust models to forge watermarked images also exceeds that of clean models. The decoded accuracy of the forged bit accuracy is significantly higher than that of the clean model in Figure~\ref{fig:id_forgery}. Moreover, the Pearson correlation coefficients for empirical and certified robust models are $0.96$ and $0.97$, which is higher than that for clean models. 
This indicates a stronger correlation between the residual image distance and the secret watermark distance, which results in a higher success rate for watermark extraction attacks.

\begin{figure}[!t]
\centering
  \begin{subfigure}[b]{0.24\linewidth} 
    \includegraphics[width=\linewidth]{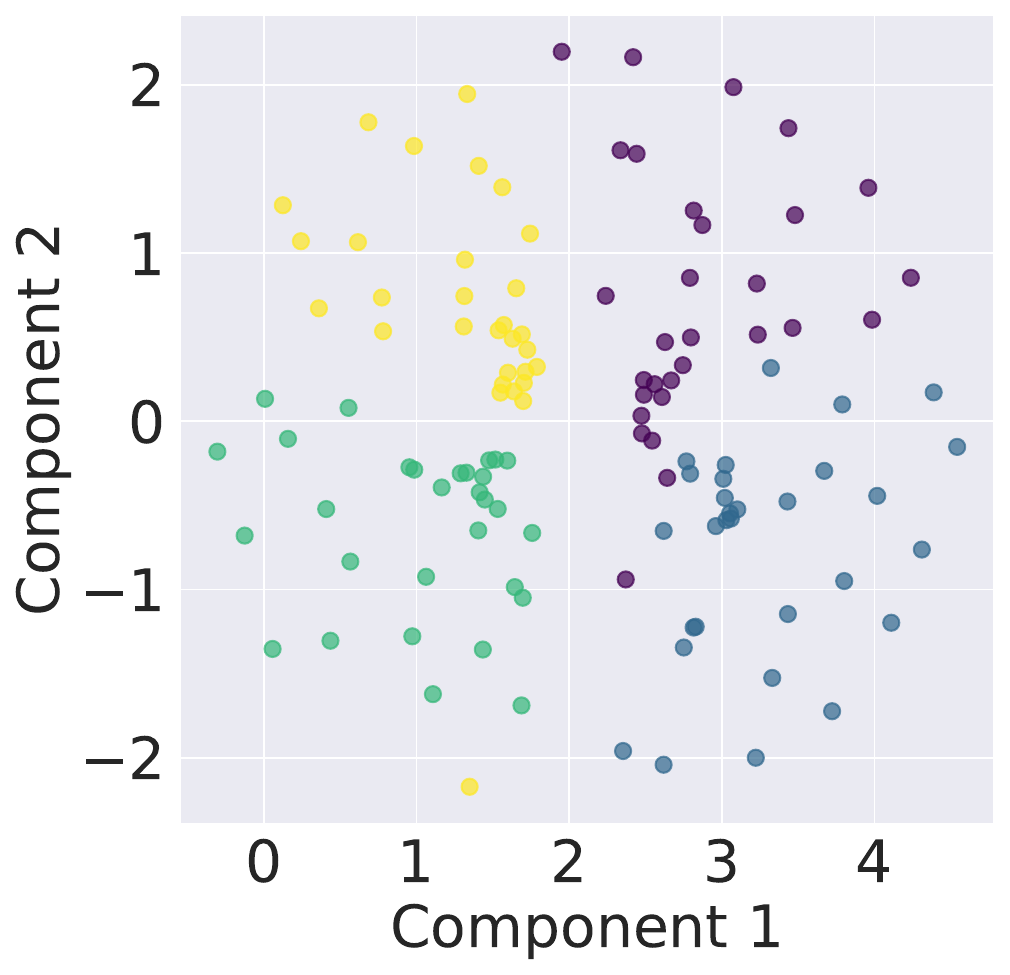}
    \caption{COCO}
    \label{fig:emprical_coco}
  \end{subfigure}
  \begin{subfigure}[b]{0.24\linewidth}
    \includegraphics[width=\linewidth]{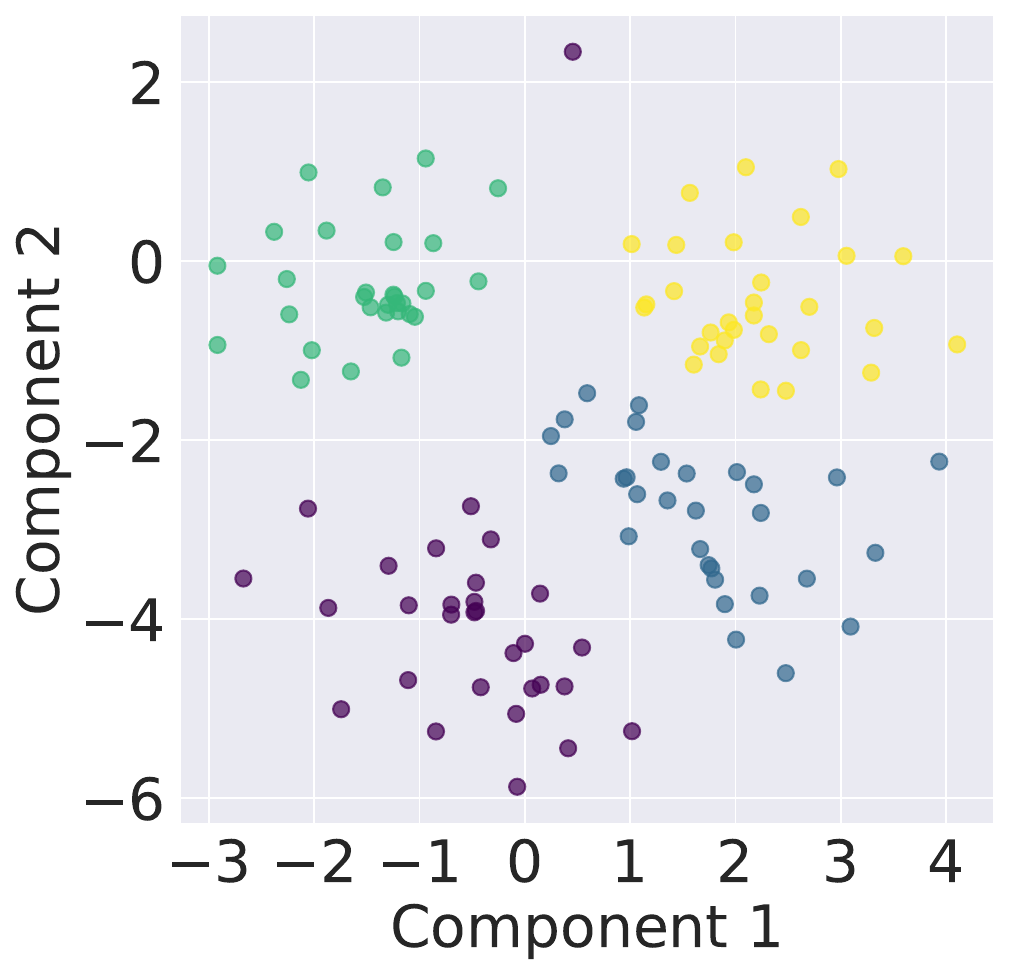}
    \caption{CelebA}
    \label{fig:emprical_celeb}
  \end{subfigure}
    \begin{subfigure}[b]{0.24\linewidth} 
    \includegraphics[width=\linewidth]{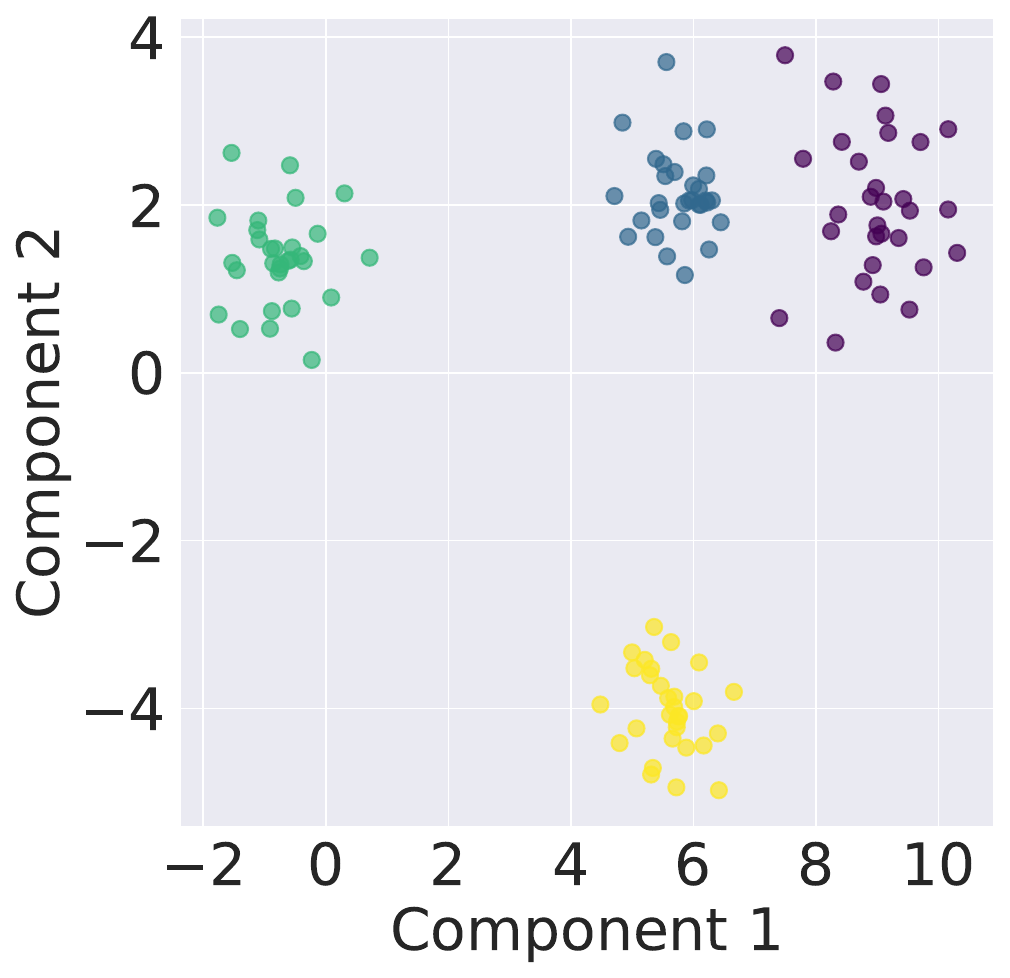}
    \caption{COCO}
    \label{fig:distortrs_coco}
  \end{subfigure}
  \begin{subfigure}[b]{0.24\linewidth}
    \includegraphics[width=\linewidth]{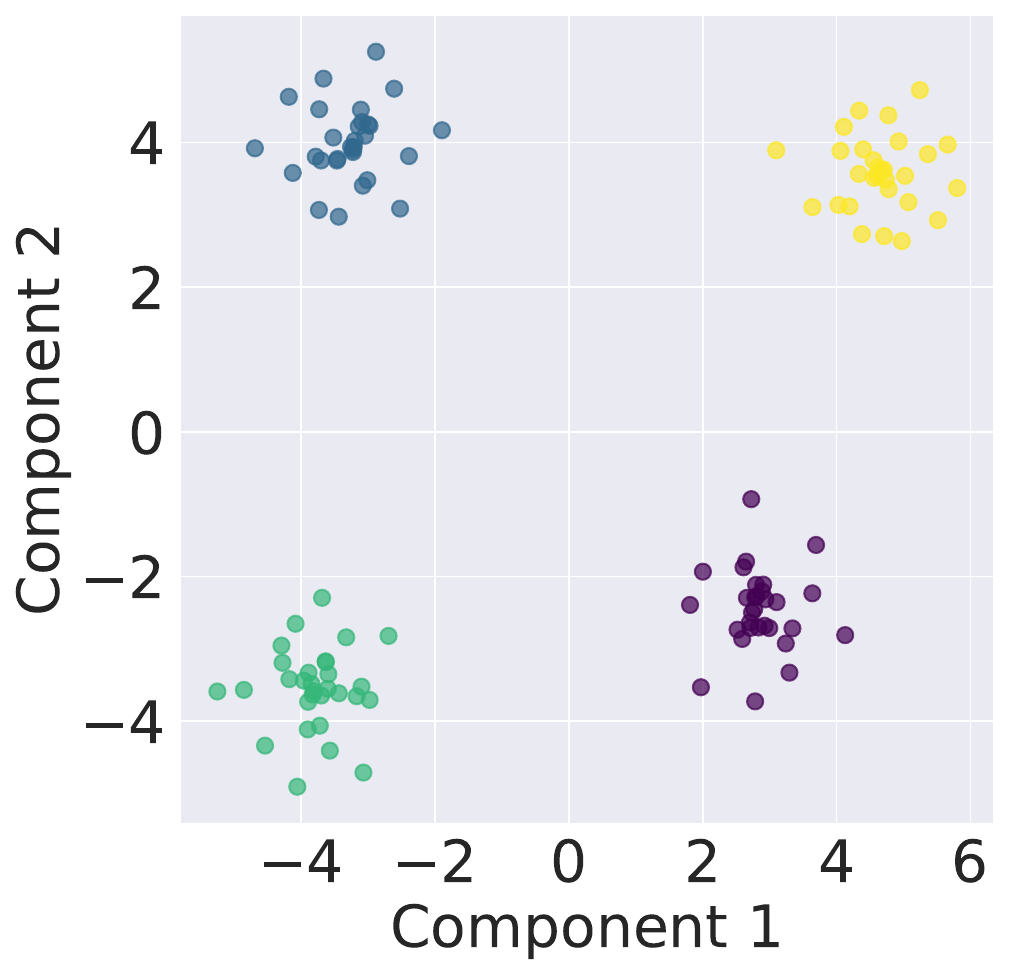}
    \caption{CelebA}
    \label{fig:distortrs_celeb}
  \end{subfigure}
\vspace{-3mm}
\caption{The cluster results for residual images on HiDDeN under empirical (\ref{fig:emprical_coco} and \ref{fig:emprical_celeb}) and certified (\ref{fig:distortrs_coco} and \ref{fig:distortrs_celeb}) robust models. The intra-cluster distances 
are as follows: (a)-$[0.77,0.79,0.80,0.72]$, (b)-$[0.85,1.06,1.00,0.89]$, (c)-$[0.66,0.68,0.60,0.86]$, (d)-$[0.64,0.62,0.58,0.62]$. 
}
\vspace{-4mm}
\label{fig:id_link_robust}
\end{figure}


These observations indicate that in order to mitigate the impact of adversarial perturbations on watermarked images, the encoder network embeds a greater number of secret watermarks into the images, leading to residual images that contain a greater quantity of secret watermarks. Furthermore, incorporating noise during the robust training process enhances the decoder network's capacity to decode perturbed images correctly. This inadvertently enhances the decoder's ability to extract messages from residual images.

\vspace{0.05in}
\noindent\textbf{Empirically Robust (W-ER) v.s. Certified Robust (W-CR):} 
The degree of identity leakage differs between W-ER and W-CR methods due to variations in training noise magnitude and combination, making it inappropriate to rank their robustness uniformly. 
Figure~\ref{fig:id_link_robust} shows that the W-CR leads to greater identity leakage than the W-ER (HiDDeN$^+$). Conversely, Figure~\ref{fig:id_forgery_robust} shows that the W-CR yields less identity leakage than the W-ER (StegaStamp$^+$).

\begin{figure}[t]
  \centering
  \begin{subfigure}[b]{0.24\columnwidth}
        \centering
        \includegraphics[width=\textwidth]{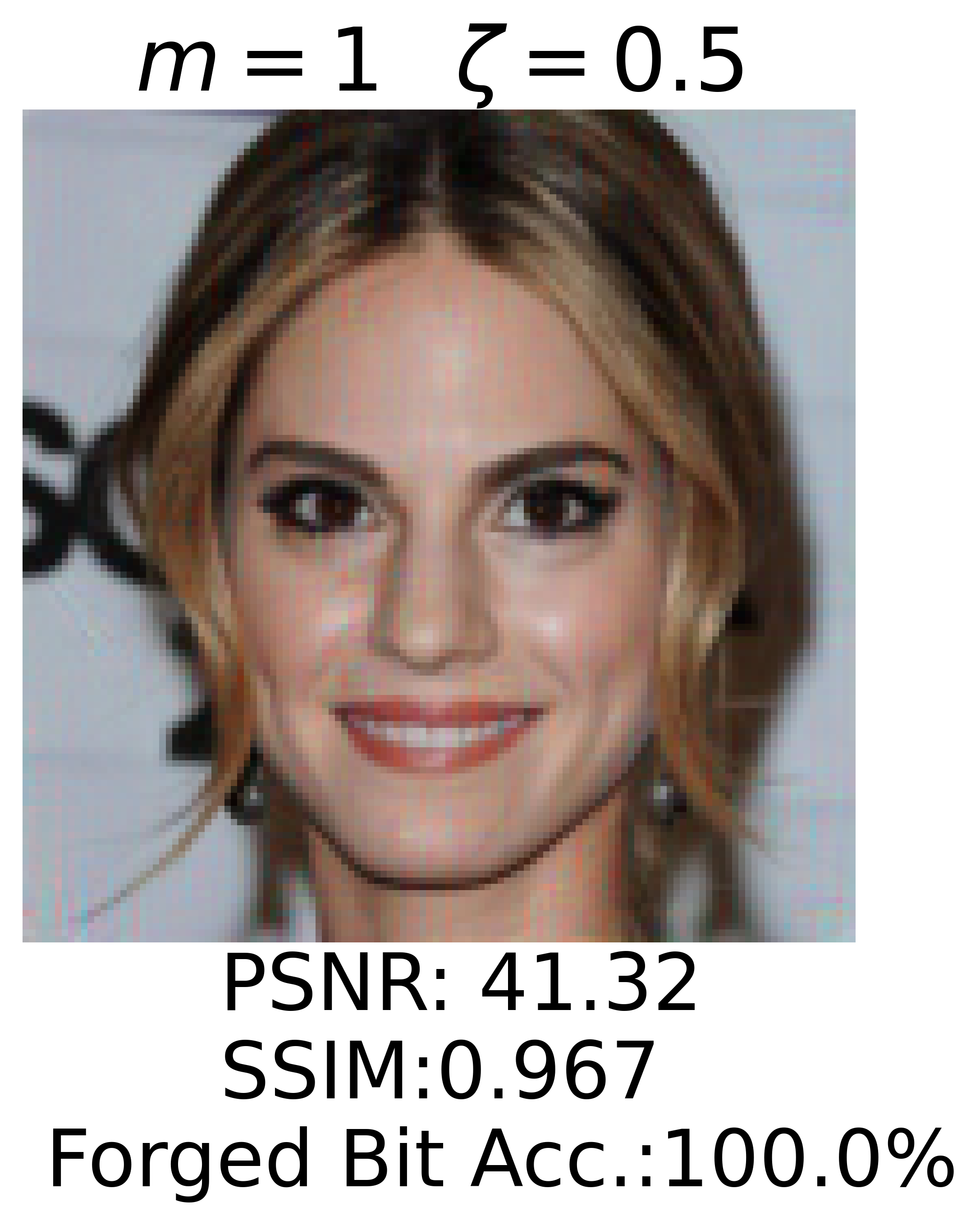}
    \end{subfigure}
    \begin{subfigure}[b]{0.24\columnwidth}
        \centering
        \includegraphics[width=\linewidth]{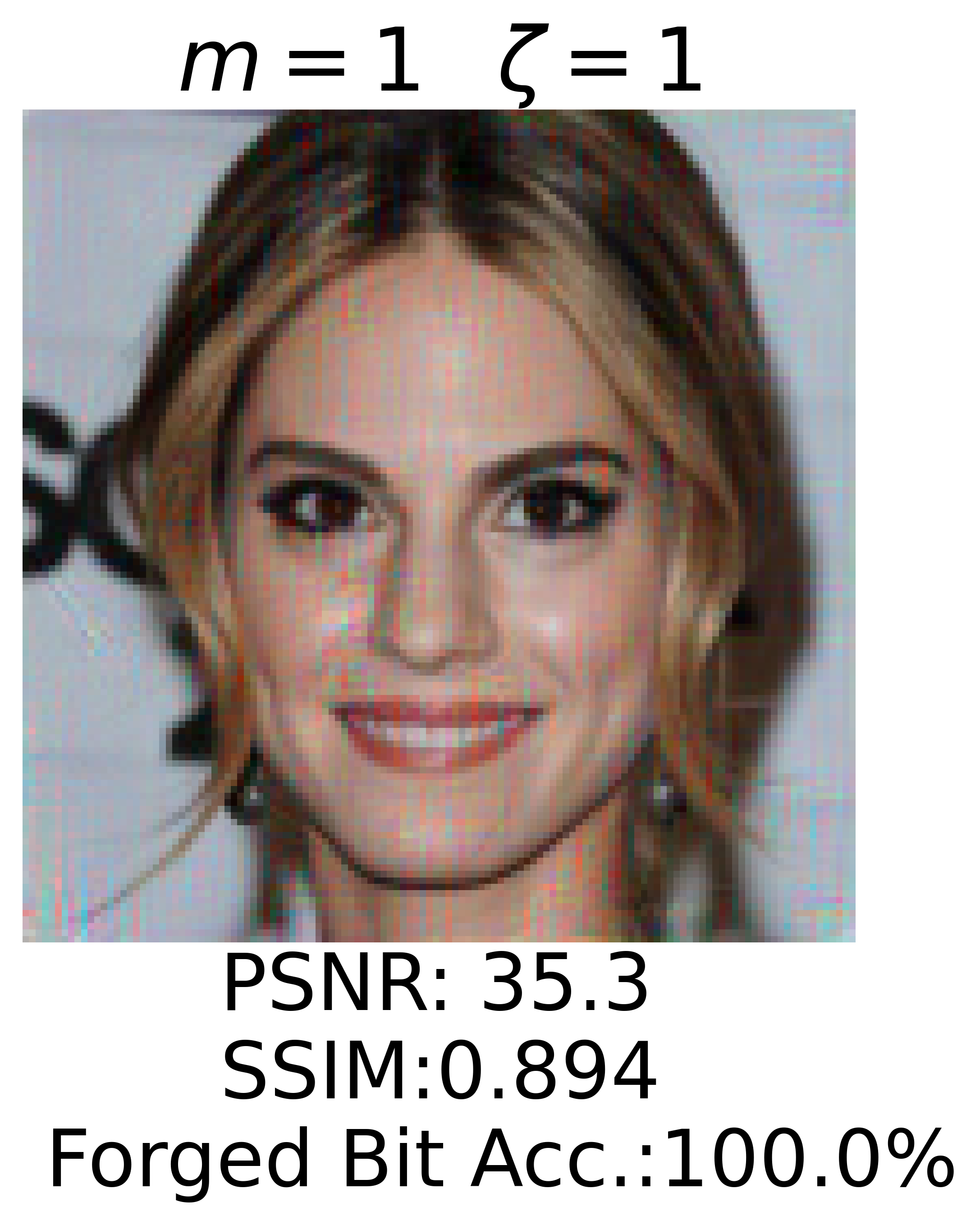}
    \end{subfigure}
    \begin{subfigure}[b]{0.24\columnwidth}
        \centering
        \includegraphics[width=\textwidth]{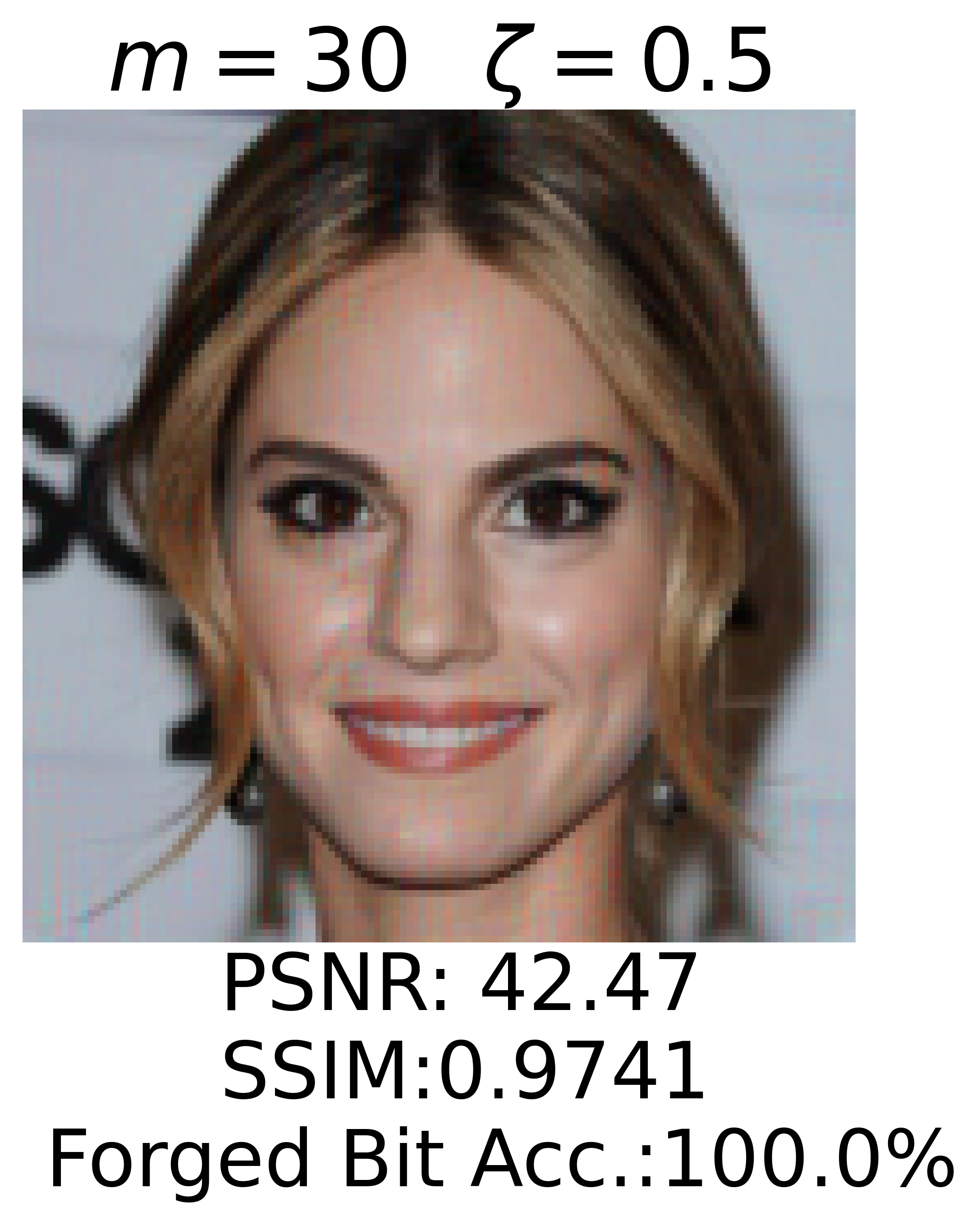}
    \end{subfigure}
    \begin{subfigure}[b]{0.24\columnwidth}
        \centering
        \includegraphics[width=\linewidth]{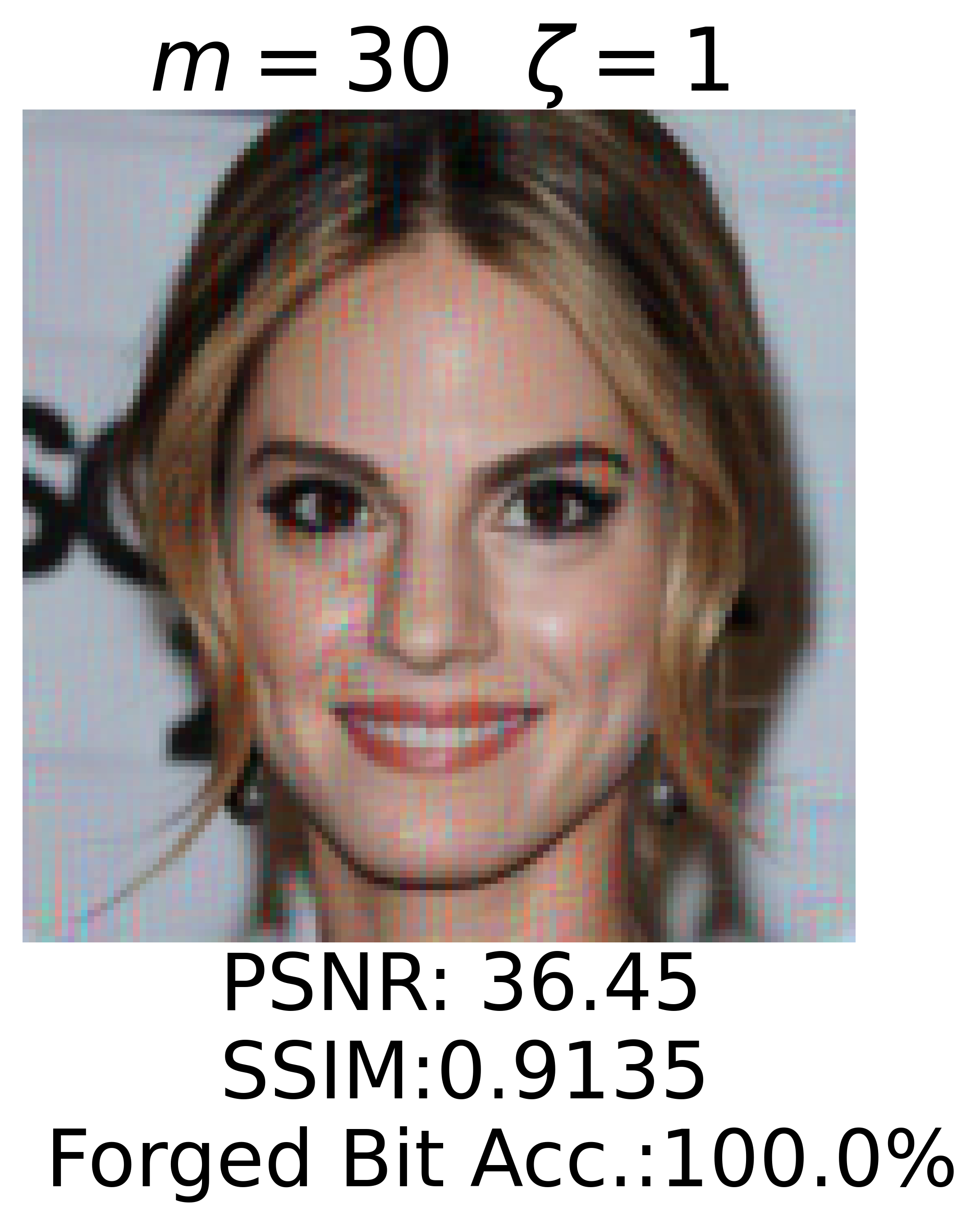}
    \end{subfigure} 

      \begin{subfigure}[b]{0.24\columnwidth}
        \centering
        \includegraphics[width=\textwidth]{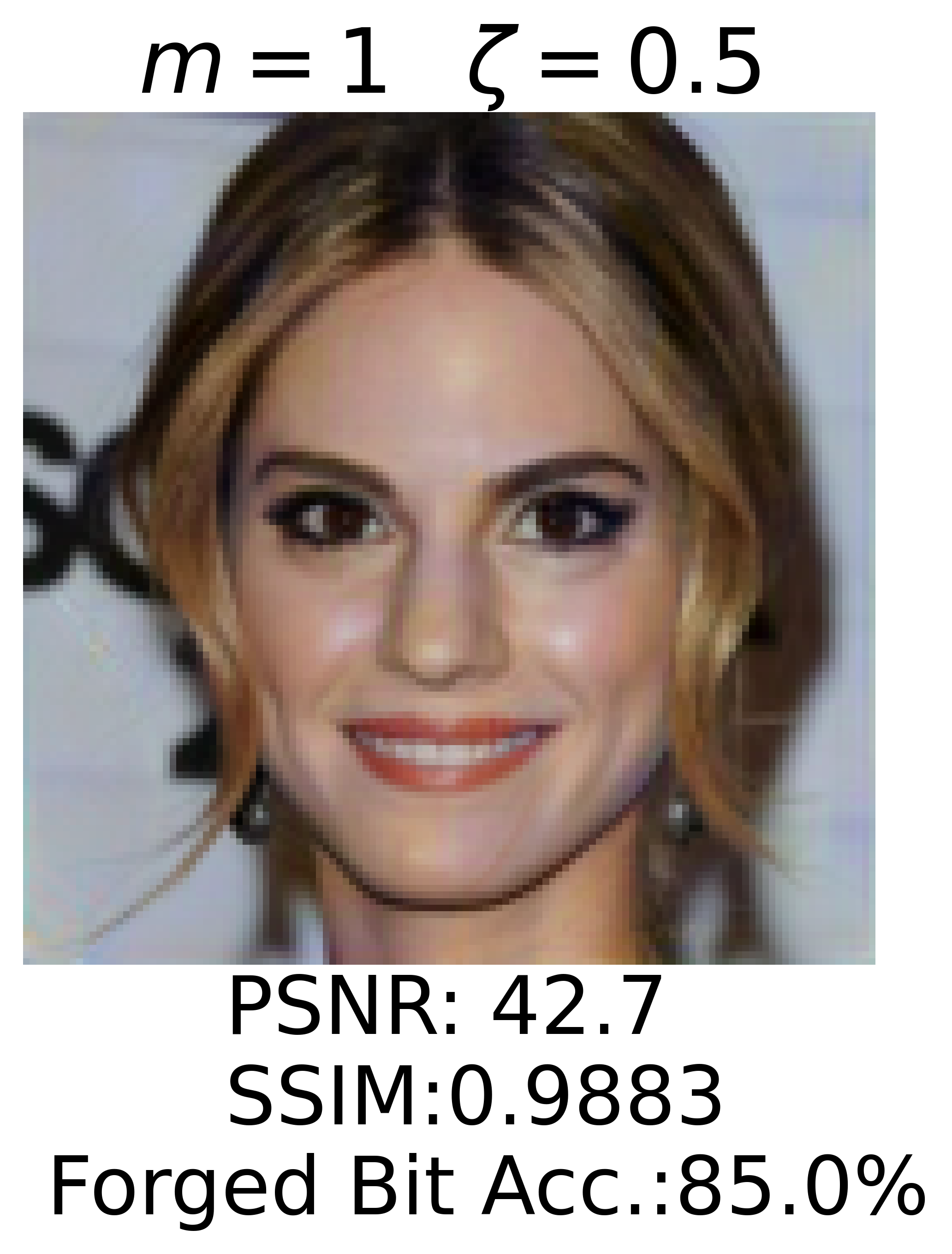}
    \end{subfigure}
    \begin{subfigure}[b]{0.24\columnwidth}
        \centering
        \includegraphics[width=\linewidth]{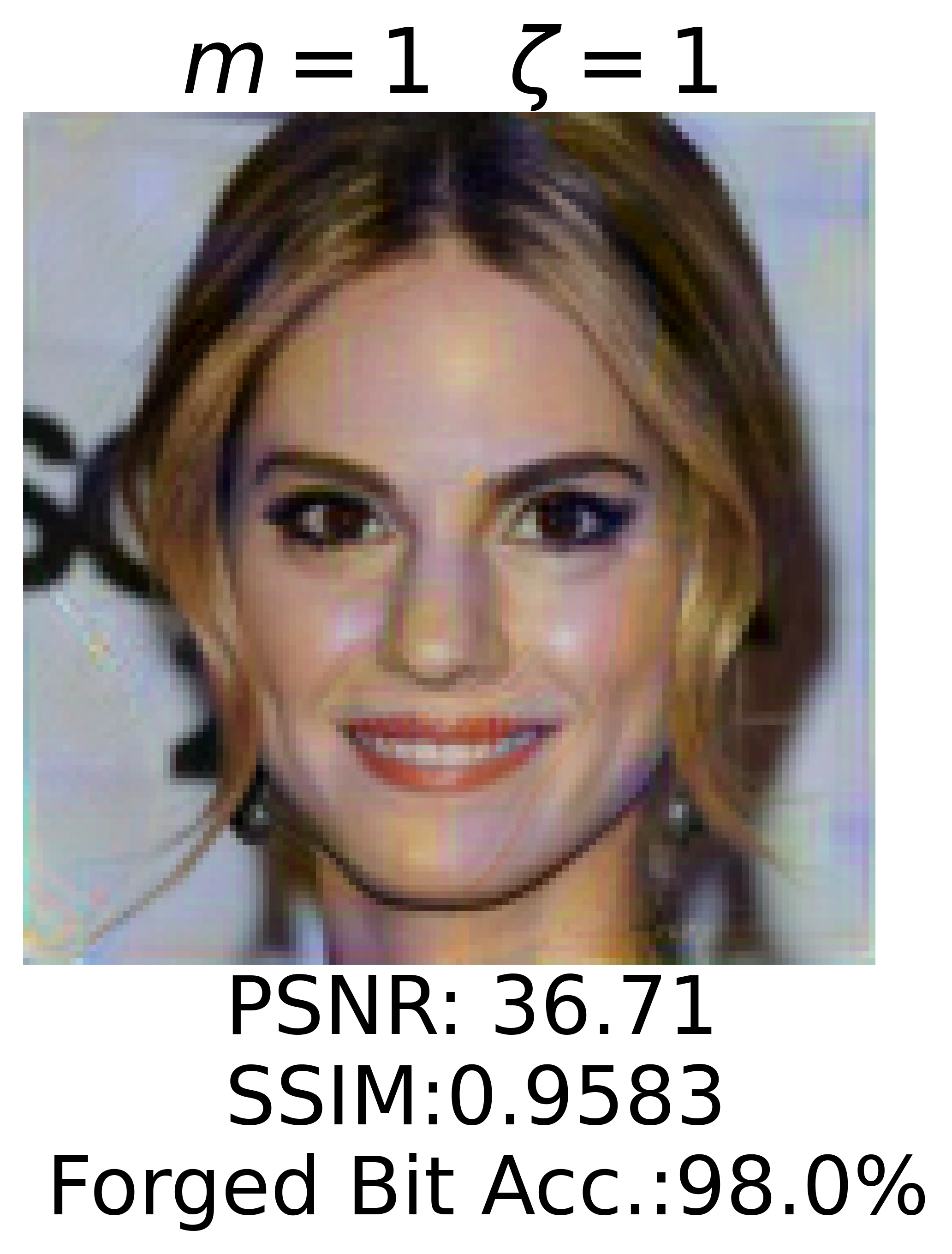}
    \end{subfigure}
    \begin{subfigure}[b]{0.24\columnwidth}
        \centering
        \includegraphics[width=\textwidth]{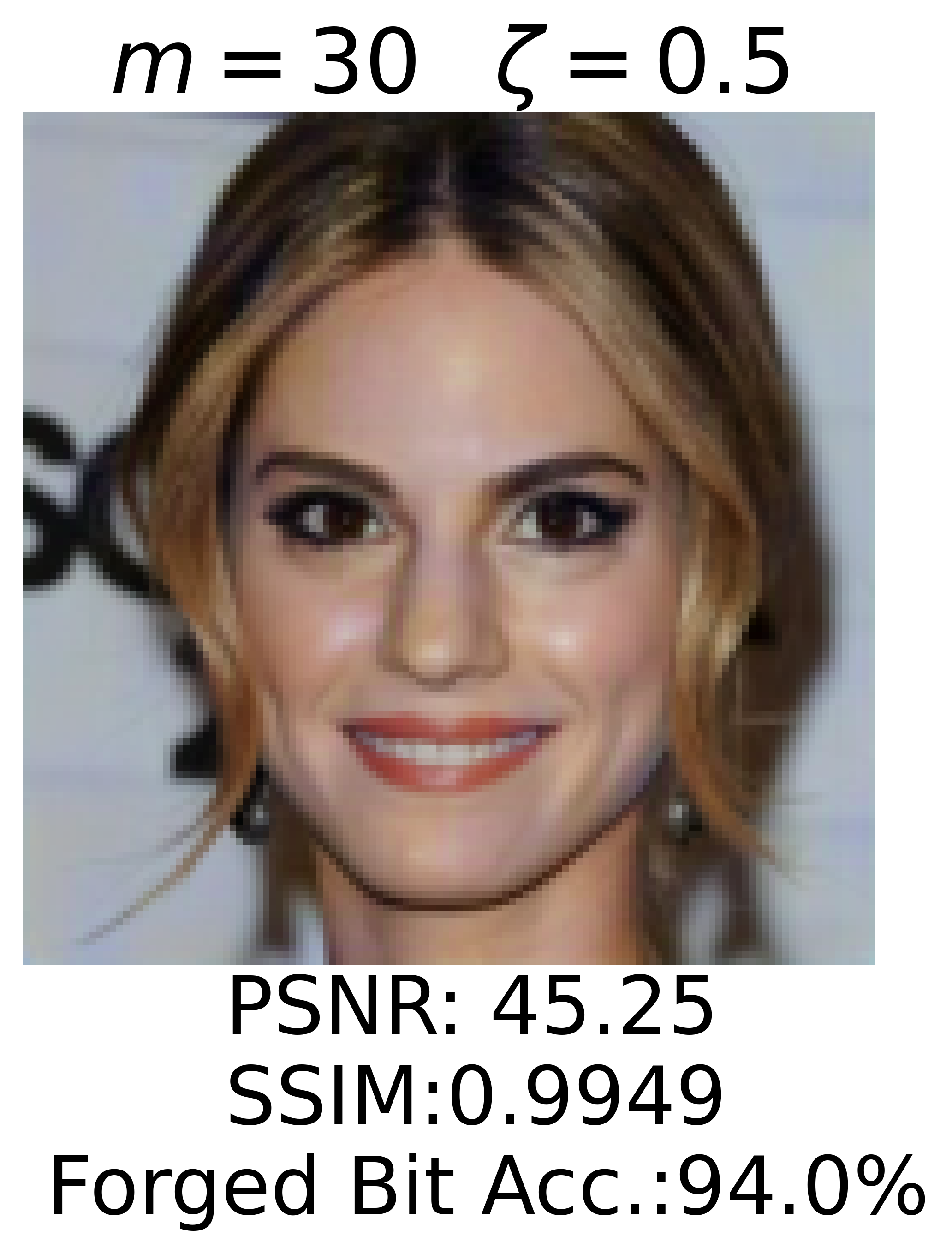}
    \end{subfigure}
    \begin{subfigure}[b]{0.24\columnwidth}
        \centering
        \includegraphics[width=\linewidth]{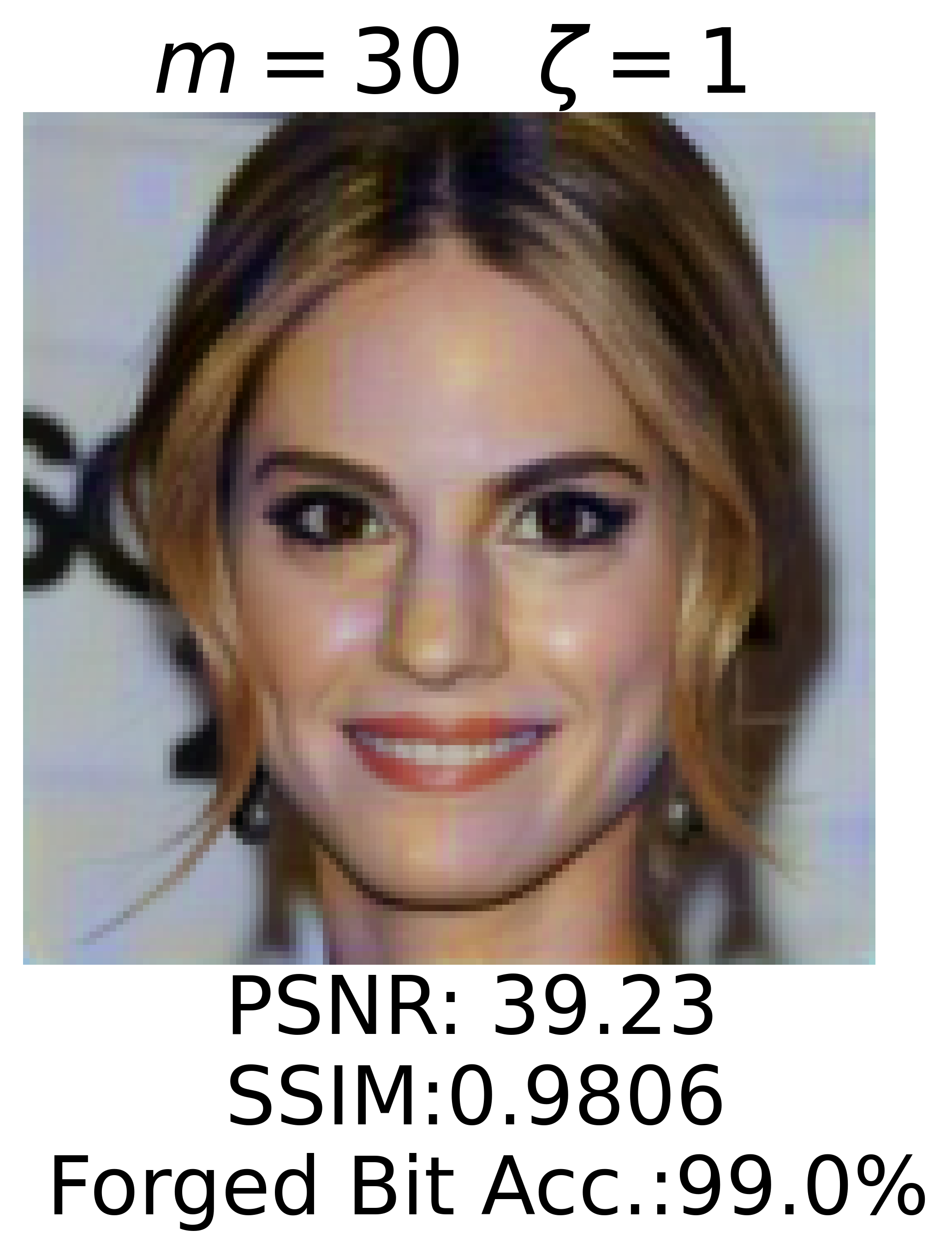}
    \end{subfigure} 
    \vspace{-3mm}
  \caption{Forging watermarked images on StegaStamp under empirical (top) and certified (bottom) robust models. 
}\vspace{-4mm}
  \label{fig:id_forgery_robust}
\end{figure}






\section{Mitigating Identity Leakage} \label{sec:ril}


During the embedding, the secret watermark $t$ is incorporated into the watermarked image $w$ via the encoder network. Ideally, per the Information Bottleneck (IB) principle, $w$ should retain all the information from both the secret watermark $ t $ and the clean image $ x $. However, identity extraction attack results (see Section~\ref{sec:leakage_t}) reveal that secret watermark $t$ can be deciphered by a simple operation: the residual difference between $w$ and $x$. \textit{This means that the watermark information $t$ inadvertently flows into the residual image $z$.}



\begin{figure}[!h]
  \centering
  \includegraphics[width=0.98\linewidth]{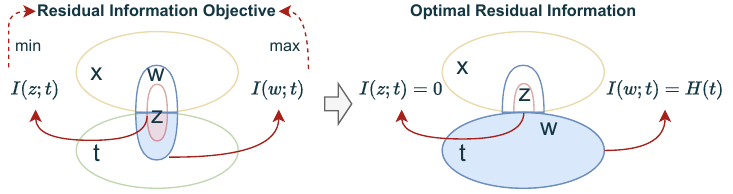}
  \vspace{-3mm}
  \caption{Information content of feature representations.} 
  \label{fig:vene}
   \vspace{-2mm}
\end{figure}

To address the issue, 
we introduce the residual information objective from the information theory perspective, as shown in Figure~\ref{fig:vene}. Particularly, our two tasks in mitigating identity leakage 
are: 1) ensuring that $ w $ retains all of $ t $ (sufficiency of $ w $ for $ t $), which involves maximizing the mutual information between $ w $ and $ t $; 2) ensuring that $ z $ does not contain any information about $ t $, which involves minimizing the mutual information between $ z $ and $ t $.
Note that the absence of \( t \) in \( z \) does not indicate that \( t \) is also absent in \( w \). This distinction arises from the fact that the transmission of information does not occur solely through subtraction.
By combining these two tasks, our residual information objective will be converted to: 
{
\begin{equation}
     \max I(w;t) - I(z;t),
     \label{eq:residual_objective}
\end{equation}
}

\par
\noindent where the mutual information $I(A;B)$ between two variables $ A $ and $ B $ is defined as: 
\vspace{-1mm}
{
\begin{equation}
I(A;B) = \int p(A, B) \log \frac{p(A, B)}{p(A) p(B)} d A d B . 
\end{equation}
}

\par
However, many studies have shown that mutual information cannot be directly calculated in high-dimensional space. Alternatively, it is 
approximated using neural estimators~\cite{poole2019variational,belghazi2018mutual,nguyen2010estimating}. 
To deal with the estimation of mutual information in Eq.~\ref{eq:residual_objective}, we introduce the following theorem:

\begin{theorem} 
Eq.~\ref{eq:residual_objective} can be estimated as follows:
{\small
\begin{equation}
\begin{aligned}
\max I(w;t) - I(z;t)  \Leftarrow \max f_{KL}[\mathbb{P}_w \| \mathbb{P}_z] - f_{KL}[\mathbb{P}_z \| \mathbb{P}_w],
\end{aligned}
\end{equation}
}

\par
\noindent where $\mathbb{P}_w=p(t|w), \mathbb{P}_z=p(t|z)$ denote the probability distributions, and $f_{KL}$ represents the KL divergence.
\label{theorem:mutual_info}
\end{theorem}

\vspace{-2mm}
 \par
\begin{proof}
    Detailed proof is deferred to Appendix~\ref{appendix:RIL_proof}.
\end{proof}

 With it, we define the \textit{Residual Information Loss} as:
 {\small
\begin{equation}
\begin{aligned}
     \CL_{RIL}=\min_{E} \mathbb{E}_{w \sim \mathbb{E}_E(\CW \mid \CT,\CX)}\{f_{KL}[\mathbb{P}_z \| \mathbb{P}_w] - f_{KL}[\mathbb{P}_w \| \mathbb{P}_z] \},
\end{aligned}
\label{eq:RIL_loss}
\end{equation}}

\par
\noindent where $E$ denotes the parameters of the encoder network, as shown in Figure~\ref{fig:residual_information_loss}. $\mathbb{P}_w$ and $\mathbb{P}_z$ are the probabilities of extracting identity information (via secret $t$) from $w$ and $z$, respectively. Both $\mathbb{P}_w$ and $\mathbb{P}_z$ can be parameterized via a decoder $D$, and the expectation is approximated via a batch of samples. 
Due to the range and asymmetry of KL divergence, the estimate (in Eq.~\ref{eq:residual_objective}) ranges from $[0, +\infty]$. Upon continuously optimizing for this estimate and updating the encoder $E$, the identity information in $w$ will be retained to the maximum extent, while that in $z$ will approach 0. In this way, we can obtain a sufficient watermarked image with minimized residual identity information, i.e., an optimal watermarked image.

\begin{figure}[!t]
  \centering
  \includegraphics[width=0.98\linewidth]{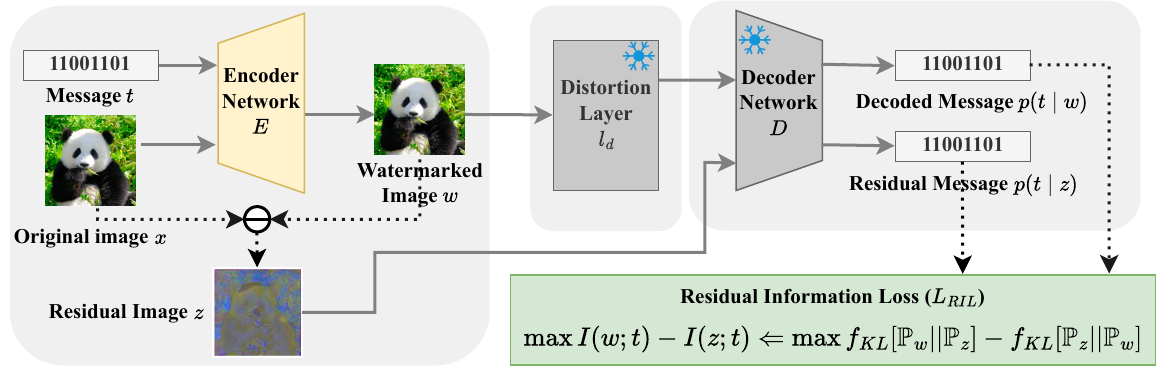}
  \vspace{-2mm}
  \caption{Identity leakage mitigation in watermarking.} 
  \label{fig:residual_information_loss}
   \vspace{-4mm}
\end{figure}

In Appendix~\ref{sec:W-IR_train}, we describe the complete training process of the W-IR, which incorporates the designed loss $\CL_{RIL}$.




\section{Experiments}

\subsection{Experimental Setup}

\noindent\textbf{Datasets.} 
We utilized two widely recognized datasets for training and testing our watermarking models: CelebA-HQ (abbreviation CelebA)~\cite{karras2018progressive} and MS-COCO (abbreviation COCO)~\cite{MS-COCO}. CelebA is a high-resolution face image dataset, 
while COCO is a large-scale dataset supporting tasks such as object detection, segmentation, key-point detection, and captioning.
On CelebA, we allocate $24,000$ images for training, $6,000$ for validation, and $2,000$ for testing. On COCO, 
we select $10,000$ images for training, 
$5,000$ images for validation, and 
$1,000$ images for testing. 
All images were rescaled to a resolution of $d=128\times128$ pixels.

\vspace{0.05in}

\noindent\textbf{Watermarking Methods.}
We use two representative open-source invisible watermarking schemes for our experiments: 
StegaStamp\cite{stegastamp} and HiDDeN\cite{zhu2018hidden}. 
For 
the secret watermark, we set its length to be $n=100$ bits and bit values are random. 
During the clean model training, we did not use the noise layer. 
All other parameters follow the default settings unless otherwise noted~\cite{zhu2018hidden,stegastamp}. 



\vspace{0.05in}

\noindent\textbf{Robust Training.} 
For empirically robust methods StegaStamp$^+$ and HiDDeN$^+$, we adhered to the noise settings specified in \cite{stegastamp} and \cite{zhu2018hidden}.
For the certified robust method W-CR, we train each scheme against additive Gaussian noise or affine transformation. 
For additive Gaussian noise, we set the variance $\sigma\in\{0.1,0.25,0.5\}$. 
For affine transformation, following the setup in~\cite{alfarra2022deformrs} on the large-scale dataset (i.e., ImageNet), we set the noise intensity to be $\sigma \in \{0.01, 0.02, 0.03\}$ in $\BSBETA \sim \mathcal{N}(\mathbf{0}, \sigma^2 \BSI_6)$ \footnote{We observed challenges on achieving convergence during adversarial training with larger perturbations, possibly due to this not being used 
in 
the original adversarial training procedures of 
HiDDeN and StegaStamp.}.

\noindent We first train our watermarking model without noise for $500$ epochs to obtain the clean vanilla model. Subsequently, we additionally train $100$ epochs on the noisy data 
to obtain the robust model (i.e., W-CR), except with the highest noise magnitude (i.e., $\sigma=0.5$ in additive Gaussian noise and $\sigma=0.03$ in affine transformation) where we trained $200$ epochs.
\emph{Without otherwise mentioned, we consider that W-CR includes the designed residual information loss.}

\vspace{0.05in}
\noindent\textbf{Identity Leakage Attacks.}
For identity forgery attacks, we evaluate two settings: each user has either $m=1$
or $m=30$ watermarked images. We set the multiplicative coefficient to $\zeta=1$. The 
$m=1$ setting demonstrates that a single watermarked image is already sufficient to forge a watermark, while the $m=30$ setting further improves the visual quality of the forged images.
For identity forgery attacks targeting individual watermarked images and identity extraction attacks, we set $m=1$, where each user possesses a single watermarked image while maintaining other parameters unchanged.


\vspace{0.05in}

\noindent\textbf{Evaluation Metrics.} 
We use \textit{PSNR} (Peak Signal-to-Noise Ratio) and \textit{SSIM} (Structural Similarity Index) to measure the visual quality after embedding secret watermarks.
For authentication, we report the \textit{bit accuracy} and \textit{accuracy} of the model in the clean test set used for vanilla training (\textit{clean vanilla}) and certified robust training (\textit{clean (bit) accuracy}).
For robustness, we evaluate \textit{certified accuracy}, defined as the fraction of the test samples verified correctly and certified robust. We randomly select $500$ examples from the test set in each dataset. For each example, we used $N_{0} = 100$ samples to select the most likely prediction $y_A$, and $N=10^{5}$ samples to estimate the lower confidence limit $\underline{p_A}$. We set $\alpha = 0.001$ for certification, meaning at least $99.9\%$ confidence.

For identity linking attacks, we use the \textit{silhouette score}, the average of the silhouette coefficients for all samples~\cite{rousseeuw1987silhouettes}. The silhouette coefficient $\in[-1,1]$ evaluates the degree of similarity between an object and its cluster compared to other clusters. A high value indicates better overall clustering performance.
To evaluate the success of identity forgery attacks, we utilize the \textit{forged bit accuracy} of forgery images and the \textit{forged accuracy} of forgery images. 
To evaluate the success of the identity extraction attack, we use \textit{attack bit accuracy} to measure the rate at which the extracted secret watermark matches the target secret watermark.



\subsection{Results on Identity Leakage Attacks}
\label{sec:leakage_t}
In this section, we compare the identity leakage results under clean models, empirically robust models (W-ER w/o $\CL_{RIL}$), and certified robust models (W-CR w/o $\CL_{RIL}$). 

\begin{table}[!t]
\centering
\scriptsize
\setlength\tabcolsep{2pt}
\caption{Forged bit accuracy under the $m=30$ setting. A lower value ($\downarrow$) indicates the model leaks fewer identities. The \textbf{bold} indicates the model with the most identity leakage.
}
\vspace{-2mm}
\resizebox{\linewidth}{!}{
\begin{tabular}{ccccccc}
\toprule
 &  &  & \multicolumn{4}{c}{W-CR w/o $\CL_{RIL}$} \\
 \cmidrule(r){4-7} 
\multirow{-2}{*}{\begin{tabular}[c]{@{}c@{}}Dataset\\ (Model)\end{tabular}} & \multirow{-2}{*}{\begin{tabular}[c]{@{}c@{}}Clean \\ vanilla\end{tabular}} & \multirow{-2}{*}{\begin{tabular}[c]{@{}c@{}}W-ER \\ w/o $\CL_{RIL}$\end{tabular}} & Noise($\sigma$) & Gaussian noise & Noise($\sigma$) & Affine \\
\midrule
 &  &  & 0.10 & 91.10\% & 0.01 & 81.34\% \\
 &  &  & 0.25 & 95.22\% & 0.02 & 77.97\% \\
\multirow{-3}{*}{\begin{tabular}[c]{@{}c@{}}COCO\\ (StegaStamp)\end{tabular}} & \multirow{-3}{*}{66.54\%} & \multirow{-3}{*}{\textbf{98.99\%}} & 0.50 & 96.66\% & 0.03 & 81.79\% \\
\midrule
 &  &  & 0.10 & 98.61\% & 0.01 & 94.29\% \\
 &  &  & 0.25 & 98.74\% & 0.02 & 95.39\% \\
\multirow{-3}{*}{\begin{tabular}[c]{@{}c@{}}CelebA\\ (StegaStamp)\end{tabular}} & \multirow{-3}{*}{84.80\%} & \multirow{-3}{*}{\textbf{100.0\%}} & 0.50 & 92.13\% & 0.03 & 91.56\% \\
\midrule
 &  &  & 0.10 & 56.23\% & 0.01 & \textbf{64.90\%} \\
 &  &  & 0.25 & 55.94\% & 0.02 & 54.56\% \\
\multirow{-3}{*}{\begin{tabular}[c]{@{}c@{}}COCO\\ (HiDDeN)\end{tabular}} & \multirow{-3}{*}{54.46\%} & \multirow{-3}{*}{55.93\%} & 0.50 & 58.97\% & 0.03 & 56.50\% \\
\midrule
 &  &  & 0.10 & \textbf{61.21\%} & 0.01 & 57.73\% \\
 &  &  & 0.25 & 57.58\% & 0.02 & 57.97\% \\
\multirow{-3}{*}{\begin{tabular}[c]{@{}c@{}}CelebA\\ (HiDDeN)\end{tabular}} & \multirow{-3}{*}{57.48\%} & \multirow{-3}{*}{58.78\%} & 0.50 & 59.28\% & 0.03 & 59.09\% \\
\bottomrule
\end{tabular}
}
\label{tab:forged_bit_acc}
\vspace{-2mm}
\end{table}

\vspace{0.05in}

\noindent\textbf{Identity Forgery Attack.}
Table~\ref{tab:forged_bit_acc} illustrates the effectiveness of identity forgery attacks across various models and datasets, revealing that nearly all models are susceptible to some extent of identity forgery. 
Under StegaStamp, both W-ER and W-CR show a significant degree of identity leakage, which is predominantly higher than that observed in clean models. Conversely, under HiDDeN, although models employing W-ER and W-CR show increased robustness, their forgery bit accuracy resembles that of the clean model, particularly at high-level noise. This phenomenon is due to the introduction of noise during training, which undermines the authentication performance (see clean bit accuracy in Table~\ref{tab:certified_acc}), leading to decreased accuracy of the decoder compared to the clean model.

Table~\ref{tab:forged_bit_acc_single} demonstrates the effectiveness of identity forgery attacks using a single watermarked image per user. We evaluate the forgery bit accuracy averaged over 1000 experimental iterations. 
The results reveal that even with just one watermarked image, almost all models exhibit some degree of identity leakage. Moreover, models under robustness training generally show more severe identity leakage compared to clean models. At the model level, consistent with previous findings, StegaStamp demonstrates more significant privacy leakage than HiDDeN.



\begin{table}[!t]
\centering
\scriptsize
\setlength\tabcolsep{2pt}
\caption{Forged bit accuracy under the $m=1$ setting. A lower value ($\downarrow$) indicates the model leaks fewer identities. The \textbf{bold} indicates the model with the most identity leakage.}
\vspace{-2mm}
\resizebox{\linewidth}{!}{
\begin{tabular}{ccccccc}
\toprule
 &  &  & \multicolumn{4}{c}{W-CR w/o $\CL_{RIL}$} \\
 \cmidrule(r){4-7} 
\multirow{-2}{*}{\begin{tabular}[c]{@{}c@{}}Dataset\\ (Model)\end{tabular}} & \multirow{-2}{*}{\begin{tabular}[c]{@{}c@{}}Clean \\ vanilla\end{tabular}} & \multirow{-2}{*}{\begin{tabular}[c]{@{}c@{}}W-ER \\ w/o $\CL_{RIL}$\end{tabular}} & Noise($\sigma$) & Gaussian noise & Noise($\sigma$) & Affine \\
\midrule
 &  &  & 0.10 & 79.82\% & 0.01 & 73.12\% \\
 &  &  & 0.25 & 91.23\% & 0.02 & 69.93\% \\
\multirow{-3}{*}{\begin{tabular}[c]{@{}c@{}}COCO\\ (StegaStamp)\end{tabular}} & \multirow{-3}{*}{67.93\%} & \multirow{-3}{*}{\textbf{95.44\%}} & 0.50 & 95.07\% & 0.03 & 69.35\% \\
\midrule
 &  &  & 0.10 & 97.95\% & 0.01 & 91.60\%  \\
 &  &  & 0.25 & 97.94\% & 0.02 & 93.57\% \\
\multirow{-3}{*}{\begin{tabular}[c]{@{}c@{}}CelebA\\ (StegaStamp)\end{tabular}} & \multirow{-3}{*}{83.32\%} & \multirow{-3}{*}{\textbf{100.0\%}} & 0.50 & 87.65\% & 0.03 & 89.26\% \\
\midrule
 &  &  & 0.10 & 62.92\% & 0.01 & 55.72\% \\
 &  &  & 0.25 & \textbf{64.01\%} & 0.02 & 57.37\% \\
\multirow{-3}{*}{\begin{tabular}[c]{@{}c@{}}COCO\\ (HiDDeN)\end{tabular}} & \multirow{-3}{*}{53.3\%} & \multirow{-3}{*}{62.93\%} & 0.50 & 62.35\% & 0.03 & 59.84\% \\
\midrule
 &  &  & 0.10 & \textbf{74.46\%} & 0.01 & 58.01\% \\
 &  &  & 0.25 & 68.48\% & 0.02 & 57.89\% \\
\multirow{-3}{*}{\begin{tabular}[c]{@{}c@{}}CelebA\\ (HiDDeN)\end{tabular}} & \multirow{-3}{*}{55.15\%} & \multirow{-3}{*}{68.93\%} &  0.50 & 60.15\% & 0.03 & 59.57\% \\
\bottomrule
\end{tabular}
}
\label{tab:forged_bit_acc_single}
\vspace{-4mm}
\end{table}

\vspace{0.05in}

\noindent\textbf{Identity Extraction Attack.}
Table~\ref{tab:attack_bit_acc} illustrates the effects of identity extraction attacks under various settings (using only one watermarked image). The results indicate that all models are at risk of secret watermark leakage, with attack bit accuracy (reflecting the extent of identity leakage) increasing under robust training compared to the clean model. Moreover, under Stegastamp, the watermark extraction bit accuracy of W-ER and W-CR exceeds $90\%$, highlighting a substantial risk of identity leakage. 
Notably, the bit accuracy presented in Table~\ref{tab:attack_bit_acc} is generally higher than the values shown in Table~\ref{tab:forged_bit_acc}. The difference is because identity extraction attacks directly compare the extracted secret watermark to the ground truth watermark, without relying on the decoder for decoding. As a result, bit accuracy remains high in instances of reduced decoder efficiency (e.g., under Gaussian noise on HiDDeN).

\begin{table}[!t]
\centering
\scriptsize
\setlength\tabcolsep{2pt}
\caption{Attack bit accuracy under different settings. A lower value ($\downarrow$) indicates the model leaks fewer identities. The \textbf{bold} indicates the model with the most identity leakage.
}
\vspace{-2mm}
\resizebox{\linewidth}{!}{
\begin{tabular}{ccccccc}
\toprule
 &  &  & \multicolumn{4}{c}{W-CR w/o $\CL_{RIL}$} \\
 \cmidrule(r){4-7} 
\multirow{-2}{*}{\begin{tabular}[c]{@{}c@{}}Dataset\\ (Model)\end{tabular}} & \multirow{-2}{*}{\begin{tabular}[c]{@{}c@{}}Clean \\ vanilla\end{tabular}} & \multirow{-2}{*}{\begin{tabular}[c]{@{}c@{}}W-ER \\ w/o $\CL_{RIL}$ \end{tabular}} & Noise($\sigma$) & Gaussian noise & Noise($\sigma$) & Affine \\
\midrule
 &  &  & 0.10 & 90.00\% & 0.01 & 93.00\% \\
 &  &  & 0.25 & 92.00\% & 0.02 & 89.00\% \\
\multirow{-3}{*}{\begin{tabular}[c]{@{}c@{}}COCO\\ (StegaStamp)\end{tabular}} & \multirow{-3}{*}{83.00\%} & \multirow{-3}{*}{\textbf{97.00\%}} & 0.50 & 96.00\% & 0.03 & 82.00\% \\
\midrule
 &  &  & 0.10 & 88.00\% & 0.01 & 88.00\%  \\
 &  &  & 0.25 & 91.00\% & 0.02 & 81.00\% \\
\multirow{-3}{*}{\begin{tabular}[c]{@{}c@{}}CelebA\\ (StegaStamp)\end{tabular}} & \multirow{-3}{*}{76.00\%} & \multirow{-3}{*}{\textbf{98.00\%}} & 0.50 & 93.00\% & 0.03 & 78.00\% \\
\midrule
 &  &  & 0.10 & 69.00\% & 0.01 & 65.00\% \\
 &  &  & 0.25 & \textbf{75.00\%} & 0.02 & 58.00\% \\
\multirow{-3}{*}{\begin{tabular}[c]{@{}c@{}}COCO\\ (HiDDeN)\end{tabular}} & \multirow{-3}{*}{63.00\%} & \multirow{-3}{*}{66.00\%} & 0.50 & 69.00\% & 0.03 & 59.00\% \\
\midrule
 &  &  & 0.10 & \textbf{80.00\%} & 0.01 & 68.00\% \\
 &  &  & 0.25 & 74.00\% & 0.02 & 69.00\% \\
\multirow{-3}{*}{\begin{tabular}[c]{@{}c@{}}CelebA\\ (HiDDeN)\end{tabular}} & \multirow{-3}{*}{67.00\%} & \multirow{-3}{*}{69.00\%} & 0.50 & \textbf{80.00\%} & 0.03 & 60.00\% \\
\bottomrule
\end{tabular}
}
\label{tab:attack_bit_acc}
\vspace{-4mm}
\end{table}

\vspace{0.04in}


\vspace{0.05in}
\noindent\textbf{Impact of Different Parameters ($m$, $\zeta$, Distance).}  \label{sec:impact}
To investigate the impact of various parameters on identity leakage, we explored the variations in forged bit accuracy under different numbers of users ($m$), multiplicative coefficients ($\zeta$), and distances between secret watermark pairs. 

To explore the effect of $m$ and $\zeta$,  
we selected four sets of parameter combinations. The results are shown in Figure~\ref{fig:different params}. 
Generally, increasing $m$ results in higher forged bit accuracy.   
When $m=30$, the accuracy is comparable to, or even surpasses, that obtained with $m=100$. 
Similarly, increasing $\zeta$ leads to higher forged bit accuracy; however, as illustrated in Figure \ref{fig:id_forgery} this also results in a greater loss of image quality. Consequently, in our primary experiments, we have selected the settings $m=30$ and $\zeta=1$, which represent a balanced trade-off between the effectiveness of the attack and the preservation of image quality.

\subsection{Robustness Enhancement in W-IR} 

\textbf{\Rmnum{1}. Robustness Enhancement in W-ER:}
The robustness of StegaStamp$^+$ and HiDDeN$^+$ has been assessed in \cite{stegastamp,zhu2018hidden}, revealing that both methods demonstrate significant enhancements in robustness after adversarial training. Subsequent studies have further evaluated and compared the robustness of these two watermarking techniques. For instance, research findings~\cite{lei2024diffusetrace,gu2023anti} indicate that StegaStamp$^+$ exhibits superior robustness compared to HiDDeN$^+$.

\begin{table*}[!t]
\centering
\scriptsize
\setlength\tabcolsep{2.5pt}
\caption{Certified robustness, authentication effectiveness, and visual quality of W-CR without residual information loss $\CL_{RIL}$.
}
\vspace{-2mm}
\begin{tabular}{ccccccccccccccccc}
\toprule
 & \multicolumn{4}{c}{Clean vanilla} & \multicolumn{6}{c}{W-CR: Additive Gaussian noise} & \multicolumn{6}{c}{W-CR: Affine transformation} \\
 \cmidrule(r){2-5} \cmidrule(r){6-11} \cmidrule(r){12-17} 
\begin{tabular}[c]{@{}c@{}}Dataset\\ (Model)\end{tabular} & PSNR↑ & SSIM↑ & \begin{tabular}[c]{@{}c@{}}Clean \\ Bit Acc.↑\end{tabular} & \begin{tabular}[c]{@{}c@{}}Clean \\ Acc.↑\end{tabular} & \begin{tabular}[c]{@{}c@{}}Noise\\ ($\sigma$)\end{tabular} & PSNR↑ & SSIM↑ & \begin{tabular}[c]{@{}c@{}}Clean\\  Bit Acc.↑\end{tabular} & \begin{tabular}[c]{@{}c@{}}Clean\\  Acc.↑\end{tabular} & \begin{tabular}[c]{@{}c@{}}Certified\\  Acc.↑\end{tabular} & \begin{tabular}[c]{@{}c@{}}Noise\\ ($\sigma$)\end{tabular} & PSNR↑ & SSIM↑ & \begin{tabular}[c]{@{}c@{}}Clean\\  Bit Acc.↑\end{tabular} & \begin{tabular}[c]{@{}c@{}}Clean\\  Acc.↑\end{tabular} & \begin{tabular}[c]{@{}c@{}}Certified \\ Acc.↑\end{tabular} \\
\midrule
\multirow{3}{*}{\begin{tabular}[c]{@{}c@{}}COCO\\ (StegaStamp)\end{tabular}} & \multirow{3}{*}{30.79} & \multirow{3}{*}{0.9314} & \multirow{3}{*}{99.82\%} & \multirow{3}{*}{100.0\%} & 0.10 & 27.22 & 0.8835 & 99.99\% & 100.0\% & 96.60\% & 0.01 & 28.24 & 0.9111 & 99.98\% & 100.0\% & 94.60\% \\
 &  &  &  &  & 0.25 & 23.97 & 0.8217 & 100.0\% & 100.0\% & 98.80\% & 0.02 & 27.18 & 0.9013 & 99.99\% & 100.0\% & 94.80\% \\
 &  &  &  &  & 0.50 & 19.92 & 0.7241 & 100.0\% & 100.0\% & 99.40\% & 0.03 & 26.03 & 0.8948 & 99.99\% & 100.0\% & 95.80\% \\
 \midrule
\multirow{3}{*}{\begin{tabular}[c]{@{}c@{}}CelebA\\ (StegaStamp)\end{tabular}} & \multirow{3}{*}{31.39} & \multirow{3}{*}{0.9360} & \multirow{3}{*}{99.97\%} & \multirow{3}{*}{100.0\%} & 0.10 & 30.01 & 0.9089 & 100.0\% & 100.0\% & 100.0\% & 0.01 & 33.33 & 0.9505 & 99.99\% & 100.0\% & 100.0\% \\
 &  &  &  &  & 0.25 & 25.60 & 0.8485 & 100.0\% & 100.0\% & 100.0\% & 0.02 & 32.89 & 0.9489 & 100.0\% & 100.0\% & 100.0\% \\
 &  &  &  &  & 0.50 & 19.48 & 0.6509 & 99.90\% & 100.0\% & 100.0\% & 0.03 & 30.57 & 0.9412 & 100.0\% & 100.0\% & 100.0\% \\
 \midrule
\multirow{3}{*}{\begin{tabular}[c]{@{}c@{}}COCO\\ (HiDDeN)\end{tabular}} & \multirow{3}{*}{32.84} & \multirow{3}{*}{0.9418} & \multirow{3}{*}{83.40\%} & \multirow{3}{*}{100.0\%} & 0.10 & 27.66 & 0.8500 & 81.48\% & 99.67\% & 100.0\% & 0.01 & 29.65 & 0.9140 & 82.76\% & 100.0\% & 100.0\% \\
 &  &  &  &  & 0.25 & 24.09 & 0.7659 & 79.01\% & 99.92\% & 100.0\% & 0.02 & 29.59 & 0.9106 & 81.09\% & 100.0\% & 100.0\% \\
 &  &  &  &  & 0.50 & 21.76 & 0.7009 & 74.60\% & 100.0\% & 100.0\% & 0.03 & 29.59 & 0.9100 & 80.04\% & 100.0\% & 100.0\% \\
 \midrule
\multirow{3}{*}{\begin{tabular}[c]{@{}c@{}}CelebA\\ (HiDDeN)\end{tabular}} & \multirow{3}{*}{35.51} & \multirow{3}{*}{0.9604} & \multirow{3}{*}{88.89\%} & \multirow{3}{*}{100.0\%} & 0.10 & 26.24 & 0.8104 & 86.90\% & 100.0\% & 100.0\% & 0.01 & 32.08 & 0.9426 & 89.87\% & 100.0\% & 100.0\% \\
 &  &  &  &  & 0.25 & 22.52 & 0.7131 & 80.01\% & 100.0\% & 100.0\% & 0.02 & 31.67 & 0.9423 & 88.37\% & 100.0\% & 100.0\% \\
 &  &  &  &  & 0.50 & 20.85 & 0.7011 & 85.09\% & 100.0\% & 100.0\% & 0.03 & 32.10 & 0.9428 & 87.58\% & 100.0\% & 100.0\% \\
 \midrule
Average & 32.63 & 0.9424 & 93.02\% & 100.0\% &  & 24.11 & 0.78157 & 90.58\% & 99.97\% & 99.57\% &  & 30.24 & 0.9258 & 92.47\% & 100.0\% & 98.77\%  \\
\bottomrule
\end{tabular}
\label{tab:certified_acc}
\vspace{-2mm}
\end{table*}

\begin{figure}[t]
    \centering
    \begin{subfigure}[b]{0.85\columnwidth}
        \centering
        \includegraphics[width=\textwidth]{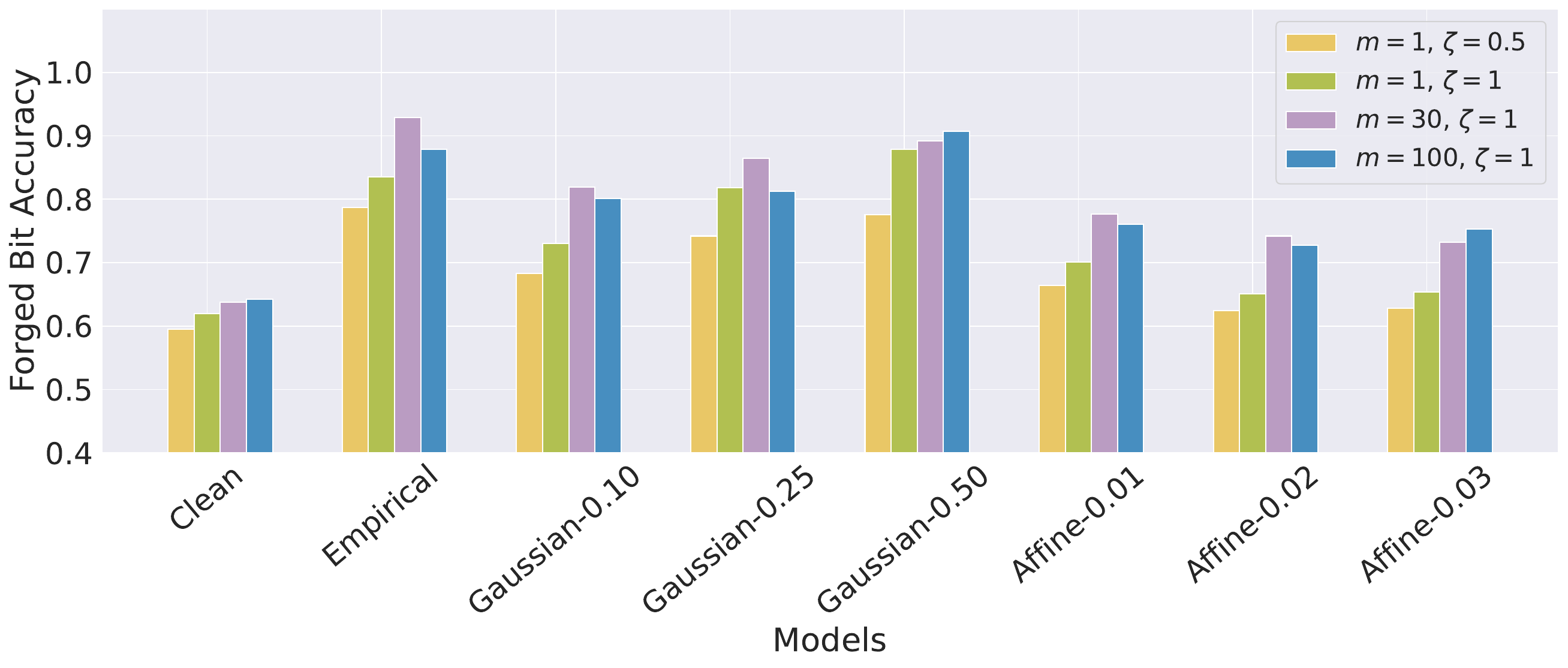}
    \end{subfigure}
\vspace{-3mm}
\caption{Impact of $m$ and $\zeta$ on COCO (StegaStamp).} 
\label{fig:different params}
\vspace{-4mm}
\end{figure}


\vspace{0.05in}
\textbf{\Rmnum{2}. Robustness Enhancement in W-CR:}
\label{sec:res-WRS-noIB}
Here we report the certified robustness results of W-CR without the designed $\CL_{RIL}$. Please refer to Appendix~\ref{sec:res-WCR-IB} for detailed robustness results of W-CR with $\CL_{RIL}$. The results indicate that our $\CL_{RIL}$ maintains the certified robustness of W-CR without $\CL_{RIL}$. 

\vspace{0.05in}
\noindent\textbf{Certified Robustness Results.}
Table~\ref{tab:certified_acc} summarizes the robustness, authentication, and visual quality of W-CR against additive Gaussian noise and affine transformation on different datasets, models, and noise magnitudes. To our knowledge, W-CR is the first scheme in the domain of post-processing image watermarking that offers provable robustness against pixel value and coordinate perturbations. We have the following key observations: 
1) Compared to the clean model, W-CR sacrifices only a marginal bit accuracy and nearly no clean accuracy. 
In addition, it achieves nearly 100\% certified accuracy: the average certified accuracy of W-CR stands at $99.57\%$ and $98.77\%$ under additive Gaussian noise and affine transformations across all datasets and models, respectively. 
2) Comparing different datasets, CelebA outperforms COCO in terms of model robustness, authentication performance, and visual quality. This is due to the larger scale of  CelebA and the uniformity of its image content (i.e., all are face images), which facilitates model convergence. 
3) As the noise magnitude increases, there is a decline in the 
authentication performance and visual quality. However, different from classic classification 
where model accuracy is significantly impacted by robust training, the clean (bit) accuracy of watermarking is less affected by robust training, though with a slight decrease in watermarked image quality. 


\begin{figure}
    \centering
    \begin{subfigure}[b]{0.48\columnwidth}
        \centering
        \includegraphics[width=\textwidth]{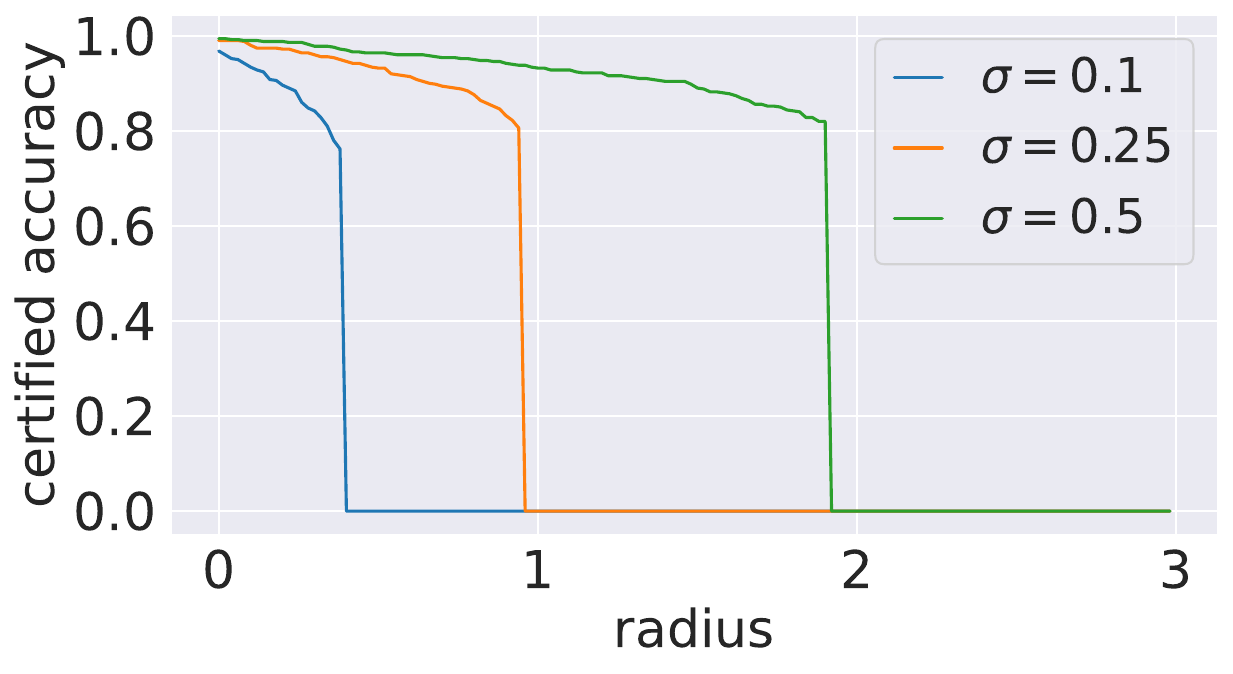}
    \end{subfigure}
    \begin{subfigure}[b]{0.48\columnwidth}
        \centering
        \includegraphics[width=\linewidth]{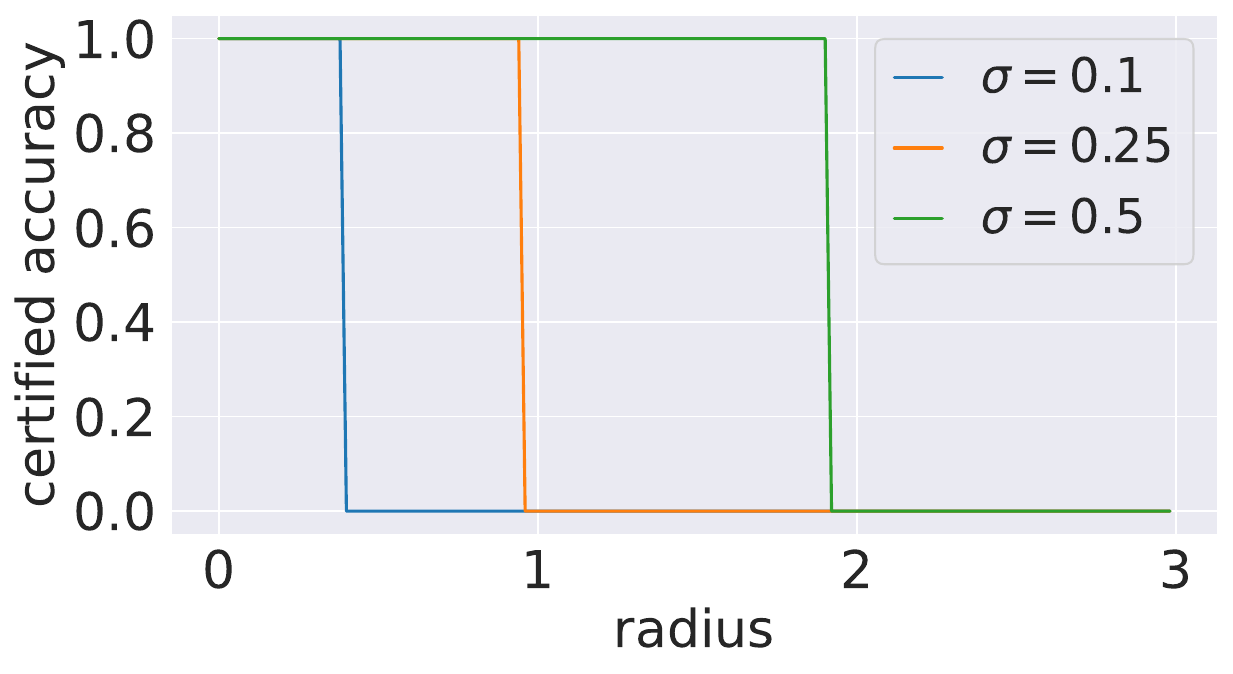}
    \end{subfigure}
    
    \begin{subfigure}[b]{0.48\columnwidth}
        \centering
        \includegraphics[width=\textwidth]{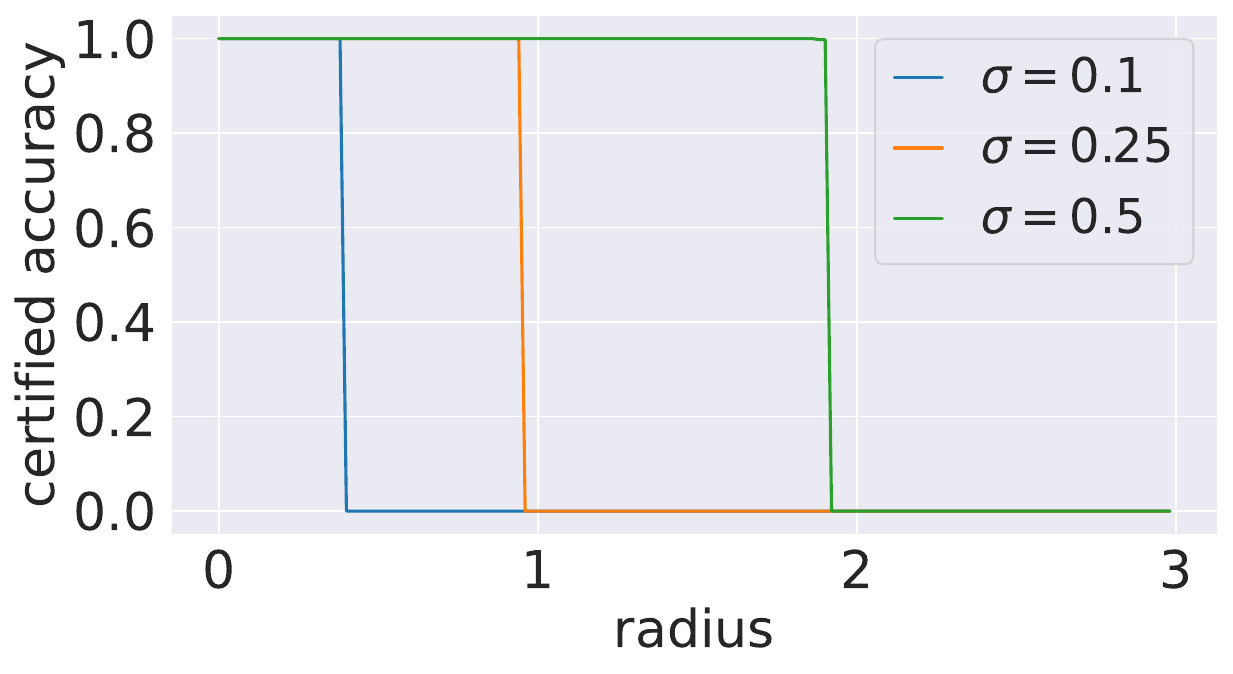}
    \end{subfigure}
    \begin{subfigure}[b]{0.48\columnwidth}
        \centering
        \includegraphics[width=\linewidth]{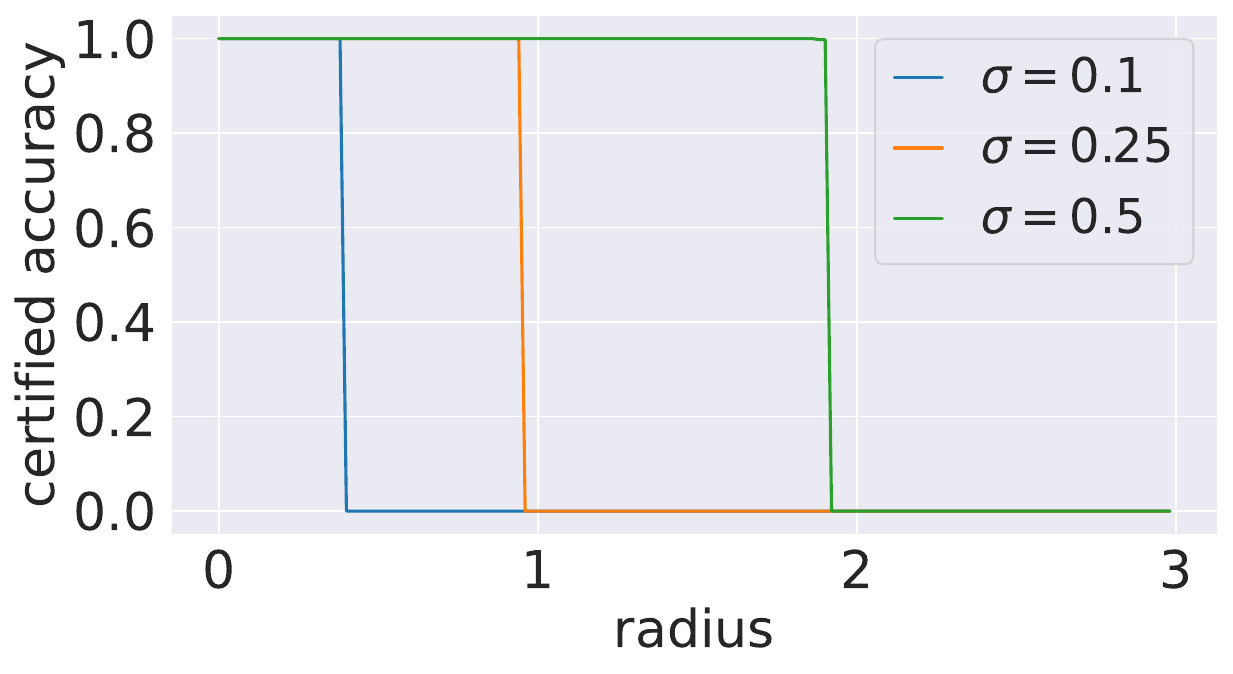}
    \end{subfigure}   

\vspace{-3mm}
\caption{Certified accuracy at different radii against additive Gaussian noise. Datasets: left-COCO and right-CelebA; Models: top-StegaStamp and bottom-HiDDeN.} 
\label{fig:additive_certify_result_GN}
\vspace{-4mm}
\end{figure}

\vspace{0.05in}
\noindent\textbf{Certified Accuracy under Different Radii.}
Figures~\ref{fig:additive_certify_result_GN} and \ref{fig:affine_certify_result_Affine} depict the impact of certified radii at varying noise magnitudes on the certified accuracy against additive Gaussian noise and affine transformation. The results show the certified accuracy decreases as the radius increases, and it drops to $0$ abruptly at a certain threshold of the radius. Moreover,
increasing the noise magnitude is shown to enlarge the certified radius. At the same radius, higher magnitudes noise achieves the highest certified accuracy; however, this is at the expense of image quality (see PSNR and SSIM in Table~\ref{tab:certified_acc}). We also observe that on HiDDeN,  a $100\%$ certified accuracy under various radii is achieved across datasets, perturbations, and noise levels. This implies the test images for certification are all certified robust under the considered setting on HiDDeN. 


    


\subsection{Identity Protection in W-IR} \label{sec:result_mitigation}


\begin{table*}[!th]
\centering
\scriptsize
\setlength\tabcolsep{5pt}
\caption{Comparison of clean, W-ER, and W-CR for identity leakage, with and without  $\CL_{RIL}$. A lower value ($\downarrow$) indicates less identity leakage. The \textbf{bold} and \textbf{\color[HTML]{478ec0} blue} indicate the highest and second-highest levels of identity protection.}
\vspace{-2mm}
\begin{tabular}{cccccccccc}
\toprule
 & \multicolumn{1}{c}{} & \multicolumn{2}{c}{Clean vanilla} & \multicolumn{2}{c}{W-ER} & \multicolumn{2}{c}{W-CR: Gaussian noise} & \multicolumn{2}{c}{W-CR: Affine} \\
 \cmidrule(r){3-4} \cmidrule(r){5-6}  \cmidrule(r){7-8} \cmidrule(r){9-10} 
Metrics & Dataset (Model) & w/o $\CL_{RIL}$ & w/ $\CL_{RIL}$ & w/o $\CL_{RIL}$ & w/ $\CL_{RIL}$ & w/o $\CL_{RIL}$ & w/ $\CL_{RIL}$ & w/o $\CL_{RIL}$ & w/ $\CL_{RIL}$ \\
 \midrule
 & COCO (StegaStamp) & \textbf{\color[HTML]{478ec0} 65.44\%} & \textbf{61.48\%} & 98.99\% & 94.89\% & 91.10\% & 88.42\% & 81.34\% & 71.18\% \\
 & CelebA (StegaStamp) & \textbf{\color[HTML]{478ec0} 84.80\%} & \textbf{76.75\%} & 100.0\% & 95.00\% & 98.61\% & 92.95\% & 94.29\% & 92.96\% \\
 & COCO (HiDDeN) & 54.46\% & \textbf{\color[HTML]{478ec0} 50.68\%} & 55.93\% & 55.45\% & 56.23\% & \textbf{50.13\%} & 64.90\% & 53.93\% \\
 & CelebA (HiDDeN) & 57.48\% & \textbf{\color[HTML]{478ec0} 50.76\%} & 58.78\% & 58.40\% & 61.21\% & \textbf{46.46\%} & 57.73\% & {  52.60\%} \\
\cmidrule(r){2-10}
\multirow{-5}{*}{\begin{tabular}[c]{@{}c@{}}Identity forgery attack:\\ Forged bit accuracy (↓)\end{tabular}} & {  Average} & \textbf{\color[HTML]{478ec0} 65.55\%} & \textbf{59.92\%} & 78.43\% & 75.93\% & 76.79\% & 69.49\% & 74.57\% & 67.67\% \\
 \midrule
 & COCO (StegaStamp) & \textbf{\color[HTML]{478ec0} 83.00\%} & \textbf{79.00\%} & 97.00\% & {  83.00\%} & 90.00\% & 86.00\% & 93.00\% & 90.00\% \\
 & CelebA (StegaStamp) & \textbf{\color[HTML]{478ec0} 76.00\%} & \textbf{74.00\%} & 98.00\% & 78.00\% & 88.00\% & 81.00\% & 88.00\% & 78.00\% \\
 & COCO (HiDDeN) & \textbf{\color[HTML]{478ec0} 63.00\%} & \textbf{61.00\%} & 66.00\% & 64.00\% & 69.00\% & 65.00\% & 65.00\% & 62.00\% \\
 & CelebA (HiDDeN) & 67.00\% & \textbf{61.00\%} & 69.00\% & 65.00\% & 80.00\% & 70.00\% & 68.00\% & \textbf{\color[HTML]{478ec0} 64.00\%} \\
\cmidrule(r){2-10}
\multirow{-5}{*}{\begin{tabular}[c]{@{}c@{}}Identity extraction attack:\\ Attack bit accuracy (↓)\end{tabular}} &  Average & \textbf{\color[HTML]{478ec0} 72.25\%} & \textbf{68.75\%} & 82.50\% & 72.50\% & 81.75\% & 75.50\% & 78.50\% & 73.50\% \\
\bottomrule
\end{tabular}
\label{tab:identity_acc_loss}
\vspace{-4mm}
\end{table*}

According to the results in Section~\ref{sec:leakage_t}, W-ER (w/o $\CL_{RIL}$) shows no significant variation in identity leakage across different noise levels. Consequently, we analyze the effects of $\CL_{RIL}$ on identity protection under low-level noise, which exhibited the best authentication effectiveness. 
Table~\ref{tab:identity_acc_loss} illustrates the effectiveness of our proposed residual information loss $\CL_{RIL}$, in mitigating identity leakage across various settings. The findings indicate that $\CL_{RIL}$ effectively reduces identity leakage under three types of attacks involving clean models, W-ER, and W-CR. Among these models, the clean model provides the highest identity protection; however, it lacks robustness enhancement. 
Notably, W-CR (w/ $\CL_{RIL}$) achieves a level of identity protection comparable to or even exceeding that of the clean model (w/o $\CL_{RIL}$). For instance, in identity linking attacks, the average silhouette score for W-CR under Gaussian noise is $0.4181$, indicating the second-highest level of identity protection. This suggests that W-CR maintains robustness without increasing identity leakage compared to clean vanilla, successfully balancing identity protection with robustness. In the case of W-ER, although identity leakage is typically highest w/o $\CL_{RIL}$, our proposed $\CL_{RIL}$ significantly mitigates this leakage. For example, under the identity extraction attack, the average attack bit accuracy of W-ER w/ $\CL_{RIL}$ is comparable to that of clean models (w/o $\CL_{RIL}$).

\section{Related Work}
\subsection{Image Watermarking}

\noindent\textbf{Post-processing Methods}
embed the watermark message into images after generation.
Conventional digital image watermarking predominantly relies on image processing techniques to embed watermark messages in the images' spatial or transform domain~\cite{an2024benchmarking}. In spatial domain watermarking, the watermark message is embedded by making slight adjustments to pixel values. 
For instance, Chopra et al.~\cite{chopra2012lsb} propose embedding messages to the least significant bit of pixel values.
Frequency domain watermarking requires transforming images from the spatial to the frequency domain. 
Al-Haj et al.~\cite{DWT-DCT} combine DWT with DCT techniques
to capitalize on both methods' advantages, which is employed by open-source generative model Stable Diffusion~\cite{rombach2022high}. However, these methods show vulnerability against image perturbations such as JPEG compression and additive noise~\cite{stablesignature,jiang2023evading,zhu2018hidden,huLearningbasedImageSteganography2024}.

Neural networks enhance post-processing methods by improving robustness and visual quality, often using encoder-decoder architectures for effective watermark embedding and extraction~\cite{stegastamp}.
For instance, Zhu et al.~\cite{zhu2018hidden} introduce a discriminator to enhance invisibility and a noise simulation layer 
to bolster robustness against mimic real-world perturbations.
Fang et al.~\cite{fang2023flow} apply invertible neural networks to watermarking, leveraging their performance in image generation and high-resolution tasks to achieve high quality watermarked images. 
Jia et al.~\cite{jia2021mbrs} propose MBRS to enhance robustness against JPEG compression by incorporating real and simulated JPEG compression during training. 

\vspace{0.03in}

\noindent\textbf{In-processing Methods} 
are commonly utilized in generative AI detection and traceability, where the watermark message is embedded directly during the image generation process.
Inspired by model fingerprinting, Yu et al.~\cite{YU1} propose embedding artificial fingerprints (i.e., watermark messages) into the training data of generative models, demonstrating that these watermark messages transfer from the training data to the model's output. 
Yu et al.~\cite{yu2021responsible} and Lukas et al.~\cite{lukas2023ptw} streamline the process by embedding watermark messages directly into model weights and fine-tuning pre-trained generative models.
Fernandez et al.~\cite{stablesignature} achieve lower training overhead by fine-tuning only the decoder part of stable diffusion. Wen et al.~\cite{wen2023tree} embed watermarks directly in the Fourier space of latent variables in a diffusion model, with extraction via diffusion inversion.
However, in-processing methods are limited in their application to non-generated images, primarily applicable to the detection of AI-generated content~\cite{an2024benchmarking} and model intellectual property protection\cite{zhao2023recipe}.

\subsection{Image Watermarking Adversarial Attacks}

\noindent\textbf{Tampering Attacks} 
manipulate the watermarked image using different transformations to prevent the correct decoding of the watermark message.
An effective and straightforward method involves applying rule-based transformations to the watermarked image and enhancing the attack's effectiveness by adjusting the transformation parameters.
For example, 
Tancik et al.~\cite{stegastamp} incorporate noise layers such as crop, Gaussian blur, affine, etc.
to simulate real-world interferences to images. 
Ma et al.~\cite{CIN} include additional operations like 
hue adjustment and dropout to enhance robustness. 
These attacks can generate numerous adversarial examples by setting various transformation parameters at a low cost, without the knowledge of watermarking neural networks~\cite{an2024benchmarking}.

Another approach involves employing finely tuned optimization algorithms to search for adversarial perturbations. For instance, Jiang et al.~\cite{jiang2023evading} propose to use PGD~\cite{PGD} and HopSkipJump~\cite{chen2020hopskipjumpattack} to search for adversarial examples. An et al.~\cite{an2024benchmarking} and Hu et al.~\cite{TransferAttackImage} trained surrogate models simulating a watermark decoder to determine the presence of a watermark in an image. 
However, they typically involve algorithms of higher complexity, requiring additional time for model training and adversarial examples generation~\cite {InvisibleImageWatermarks2023}.



\vspace{0.03in}

\noindent\textbf{Reconstruction Attacks} 
aim to recreate images without watermarks based on the content of watermarked images.
These attacks view the addition of a watermark message as noise, which can be removed from the watermarked image using well-trained reconstruction models, such as super-resolution models~\cite{lukas2023ptw}, image compression models\cite{balle2018variational,cheng2020learned,esser2021taming}, and pre-trained Denoising Diffusion Implicit Model~\cite{song2020denoising, an2024benchmarking}.
However, these methods heavily rely on the reconstruction capability of the models and often fail to remove the watermark while maintaining acceptable visual quality~\cite{lukas2023ptw,stablesignature}.

\section{Conclusion}

We reveal that post-processing image watermarking is vulnerable to both adversarial attacks and identity leakage-- adversaries can exploit leaked identity information to craft watermarked images that fool the watermark owner (even with invisible watermarks). Additionally, we observe that robust models can exacerbate such identity leakage. To address the issues, we develop W-IR, the first image watermarking framework that ensures both identity protection and robustness. 
We perform extensive experiments to demonstrate the effectiveness of the identity leakage attack and the efficacy of our W-IR in providing robustness and identity protection.



\appendices






\bibliographystyle{IEEEtran}
\bibliography{watermark,alan}






\section{Proof for Residual Information Loss}
\label{appendix:RIL_proof}

\noindent\textbf{Theorem~3.}
To deal with the estimation of mutual information in Eq.~\ref{eq:residual_objective}, we introduce the following theorem as an estimate of residual information objective to mitigate identity leakage:
\begin{equation}
\begin{aligned}
\max I(w;t) - I(z;t)  \Leftarrow \max f_{KL}[\mathbb{P}_w \| \mathbb{P}_z] - f_{KL}[\mathbb{P}_z \| \mathbb{P}_w].
\end{aligned}
\end{equation}

\begin{proof}
    
In accordance with the definition of mutual information, given three random variables $w, z,$ and $t$, the difference of mutual information can be expressed as follows:

\begin{equation*}
\small
    \begin{aligned}
    I(w;t) - I(z;t)=&\iint p(w,t) \log \frac{p(w,t)}{p(w)p(t)} dwdt-\\&\iint p(z,t) \log \frac{p(z,t)}{p(z)p(t)} dzdt
    \end{aligned}
\end{equation*}
\begin{equation}
    \begin{aligned}
    = &\iint p(w)p(t|w) \log \frac{p(t|w)}{p(t)} dwdt-\\&\iint p(z)p(t|z) \log \frac{p(t|z)}{p(t)} dzdt\\
    = &\iint p(w)p(t|w) \log \frac{p(t|w)p(t|z)}{p(t)p(t|z)} dwdt-\\&\iint p(z)p(t|z) \log \frac{p(t|z)p(t|w)}{p(t)p(t|w)} dzdt.
    \end{aligned}
    \label{eq:origin_mi_objective}
\end{equation}

By factorizing the double integrals in Eq.~\ref{eq:origin_mi_objective} into another two components, we show the following:

\begin{equation}
\small
\begin{aligned}
 &\iint p(w)p(t|w) \log \frac{p(t|w)p(t|z)}{p(t)p(t|z)} dwdt\\
= & \iint p(w)p(t|w) \log \frac{p(t|w)}{p(t|z)} dwdt +\\
&\iint p(w)p(t|w) \log \frac{p(t|z)}{p(t)} dwdt\\
= & \iint p(w)p(t|w) \log \frac{p(t|w)}{p(t|z)} dwdt +\\
&\int \left(\int p(w)p(t|w) dw \right) \log \frac{p(t|z)}{p(t)} dt\\
= & \iint p(w)p(t|w) \log \frac{p(t|w)}{p(t|z)} dwdt +\int p(t) \log \frac{p(t|z)}{p(t)} dt\\
= & \int p(w) f_{KL}[p(t|w)||p(t|z)] dw + \int p(t) \log \frac{p(t|z)}{p(t)} dt,
\end{aligned}
\end{equation}

where $f_{KL}$ represents KL-divergence. Conduct similar factorization for the second term in Eq.~\ref{eq:origin_mi_objective}, we have:
\begin{equation}
\small
\begin{aligned}
&\iint p(z)p(t|z) \log \frac{p(t|z)p(t|w)}{p(t)p(t|w)} dzdt\\
= & \iint p(z)p(t|z) \log \frac{p(t|z)}{p(t|w)} dzdt +\\
&\iint p(z)p(t|z) \log \frac{p(t|w)}{p(t)} dzdt\\
= & \iint p(z)p(t|z) \log \frac{p(t|z)}{p(t|w)} dzdt +\\
&\int \left(\int p(z)p(t|z) dz \right) \log \frac{p(t|w)}{p(t)} dt\\
= & \iint p(z)p(t|z) \log \frac{p(t|z)}{p(t|w)} dzdt +\int p(t) \log \frac{p(t|w)}{p(t)} dt\\
= & \int p(z) f_{KL}[p(t|z)||p(t|w)] dz + \int p(t) \log \frac{p(t|w)}{p(t)} dt
\end{aligned}
\end{equation}

Therefore, our objective is equivalent to:
\begin{equation*}
\begin{aligned}
&I(w;t) - I(z;t)  \\=&\int p(w) f_{KL}[p(t|w)||p(t|z)] dw + \int p(t) \log \frac{p(t|z)}{p(t)} dt - \\& \int p(z) f_{KL}[p(t|z)||p(t|w)] dz - \int p(t) \log \frac{p(t|w)}{p(t)} dt\\
\end{aligned}
\end{equation*}

\begin{equation}
\begin{aligned}
= &\int p(w) f_{KL}[p(t|w)||p(t|z)] dw - \\&\int p(z) f_{KL}[p(t|z)||p(t|w)] dz + \int p(t) \log \frac{p(t|z)p(t)}{p(t)p(t|w)} dt \\
= &\underbrace{\int p(w) f_{KL}[p(t|w)||p(t|z)] dw}_{V_1} - \\&\underbrace{\int p(z) f_{KL}[p(t|z)||p(t|w)] dz}_{V_2} + \underbrace{\int p(t) \log \frac{p(t|z)}{p(t|w)} dt}_{Q}.
\end{aligned}
\end{equation}
We can see that:
\begin{equation}
\begin{aligned}
V_1-V_2=\mathbb{E}_{w \sim \mathbb{E}_E(\CW \mid \CT,\CX)}\{f_{KL}[\mathbb{P}_w \| \mathbb{P}_z] - f_{KL}[\mathbb{P}_z \| \mathbb{P}_w]\},
\end{aligned}
\end{equation}
where $\mathbb{P}_w=p(t|w), \mathbb{P}_z=p(t|z)$ denote the probability distributions. Note that, in our system, the residual image $z$ is the difference between the watermarked image $w$ and the original image $x$. Our goal is that $w$ contains all information of secret watermark $t$ while $z$ does not contain any information of $t$, i.e., $t$ and $z$ are orthogonal to each other. Therefore, in maximizing $V_1-V_2$, $t$ and $z$ become progressively orthogonal to each other, which allows to explicitly approximate $p(t|z)$ as $ p(t)$, so $Q \approx \int p(t) \log \frac{p(t)}{p(t|w)} dt = f_{KL}[p(t)||p(t|w)]$. Based on the non-negativity of KL-divergence, we further have:
\begin{equation}
\begin{aligned}
\max I(w;t) - I(z;t)  \Leftarrow \max f_{KL}[\mathbb{P}_w \| \mathbb{P}_z] - f_{KL}[\mathbb{P}_z \| \mathbb{P}_w].
\end{aligned}
\end{equation}
At the extreme, where $z$ does not contain any information about $t$, it is further shown that $w$ is sufficient for $t$. 
\end{proof}

\section{Authentication Algorithm}
\begin{algorithm}[!h]
\small
\caption{Authentication Algorithm (W-CR)}  
\begin{algorithmic}[1]
\Require Test sample's pixel value $w$, pixel coordinate $v$, authentication model $h$, perturbation type $P$, pixel transformation $\phi_P(\cdot, \epsilon)$, coordinate transformation $\theta_P(\cdot, \rho)$, secret watermark $t$, verification threshold $\tau$, $N$, $N_0$, $\alpha$.

\State $w$, $v\gets $ watermarked image pixel and coordinate
\State $\texttt{counts0} \gets \textsc{SampleUnderNoise}(h, \phi_P(w,\epsilon), \theta_P(v,\rho),$ $ t, \tau, N_0)$
\State $\hat{y}_A \gets$ top index in $\texttt{counts0}$, where $\hat{y}_A\in\{0,1\}$
\State $\texttt{counts} \gets \textsc{SampleUnderNoise}(h, \phi_P(w,\epsilon), \theta_P(v,\rho),$ $ t, \tau, N)$
\State $\underline{p_A} \gets \textsc{LowerConfBound}(\texttt{counts}[\hat{y}_A], N, 1-\alpha)$
\State \textbf{if} $\underline{p_A} > \frac{1}{2}$ \textbf{return} prediction $\hat{y}_A$ and radius $R=\sigma \Phi^{-1}(\underline{p_A})$
\State \textbf{else} \textbf{return} ABSTAIN
\end{algorithmic}
\label{alg:certify}
\end{algorithm}

\section{W-IR Training Algorithm} \label{sec:W-IR_train}

\begin{algorithm}[!h]
\small
\caption{Training Algorithm for W-IR}  
\begin{algorithmic}[1]
\Require Training dataset $\CD$, the original image $x\subset \CX$, secret watermark $t\subset \CT$, perturbation layer $l_{p}$, encoder $E$, decoder $D$.
\State Initialize encoder parameters $\psi_{E}$ and decoder parameters $\psi_{D}$. 
\For {epoch $\kappa=1,2,3,... K$}
    \For{each image $x$ in $\CD$}
    \State \textbf{Phase 1: Train Encoder and Decoder with $\CL$}
        \State Sample random secret watermark $t$ from $\CT$
        \State  $w \gets E(x, t)$ \Comment{Watermarked image}
        \State  $t_D \gets D(l_p(w))$
        \State Calculate loss $\CL(t^{\prime}, t, w, x)$ using Eq.~\ref{eq:L_loss}
        \State Update parameters $\psi_{E}$ and $\psi_{D}$ to minimize $\CL$.
\Statex
    \State \textbf{Phase 2: Train Encoder with $\CL_{RIL}$}
        \State Get residual image $z \gets  w - x$ 
        \State  $t_{w} \gets D(w)$
        \State  $t_{z} \gets D(z)$
        \State Calculate loss $\CL_{RIL}(t_{w},t_{z})$ using Eq.~\ref{eq:RIL_loss}
        \State Update only parameters $\psi_{E}$ to minimize $\CL_{RIL}$
    \EndFor
\EndFor
\State \textbf{Output:} Trained encoder and decoder parameters $\psi_{E}$, $\psi_{D}$.

\end{algorithmic}
\label{alg:W-IR_train}
\end{algorithm}

\section{Discussion and Limitations}

\noindent\textbf{The Knowledge of Residual Image.} We assume that the adversary in identity linking and forgery attacks has access to residual images of watermarked images. While the direct obtaining of residual images may be impractical in real-world scenarios, which is also beyond the scope of this work, alternative approaches exist. For instance, watermark removal methods such as VAE reconstruction~\cite{stablesignature} can provide estimates of the original images, from which residual images can be derived. To validate our attacks under such conditions, we conduct experiments under COCO (StegaStamp), employing VAE reconstruction to estimate original images and generate residual images for identity leakage assessment. Table~\ref{tab:VAEestimate} shows that despite the performance degradation, we can still successfully implement identity leakage attacks using VAE reconstruction estimates.

\vspace{0.03in}

\noindent\textbf{Identity Leakage Attacks in In-processing Watermarks.}
We investigate identity leakage in in-processing watermarks~\cite{wen2023tree,Gaussianshading}, where the embedded information strongly correlates with image semantics. We utilize Tree-Rings~\cite{wen2023tree}, a method that embeds watermark patterns in the Fourier domain of diffusion model noise. 
We generate watermarked images using four distinct patterns, each representing a different user. We use prompts sampled from the COCO dataset captions. Given that Tree-Rings does not embed conventional bit string watermarks, our investigation concentrates on identity linking and forgery attacks. The identity linking attack produces a modest Silhouette Score of 0.21, likely because the semantic watermark fundamentally alters the image content, causing residuals to retain substantial semantic information that hinders effective clustering. However, when conducting identity forgery attacks on the COCO dataset with access to a user's watermarked images, we achieve 100\% forgery accuracy, demonstrating significant vulnerabilities in identity protection.

\section{Impact of Knowing User Identity} \label{sec:know_id}
In identity forgery attacks, the adversary uses $k$-means clustering to cluster residual images without knowing their user identity.  
Here we further explore the attack when the adversary knows the user-specific residual images
and the results are shown in Table \ref{tab:ablation}. Note that when the identity is known, there is no need for the adversary to use the clustering algorithm.  
We observe that, with additive Gaussian noise, our attack achieves similar forged bit accuracy on knowing the user identity or not. For instance, on CelebA (StegaStamp) with additive Gaussian noise $\sigma=0.1$, the forged bit accuracy differs by only 0.02\%. 
With affine transformation, on average, the forged bit accuracy without knowing the identity is 2.97\% lower than that with knowing the identity. 
The results validate that the proposed k-mean clustering-based attack is effective, and knowing the user identity of residual images can further make the identity forgery attack more serious. 


\begin{table}[]
\centering
\scriptsize
\setlength\tabcolsep{2pt}
\caption{Results of identity forgery attacks where the user identity is known to the adversary or not. 
}
\vspace{-2mm}
\resizebox{\linewidth}{!}{
\begin{tabular}{cccccc}
\toprule
\multicolumn{2}{c}{}  & \multicolumn{2}{c}{COCO (StegaStamp)} & \multicolumn{2}{c}{CelebA (StegaStamp)} \\
\cmidrule(r){3-4} \cmidrule(r){5-6} 
\multicolumn{1}{c}{User identity} & & $\times$ & $\checkmark$  & $\times$ & $\checkmark$ \\
\midrule
\multicolumn{1}{l}{}  & \begin{tabular}[c]{@{}c@{}}Noise\\ ($\sigma$)\end{tabular} & \begin{tabular}[c]{@{}c@{}}Forged \\Bit Acc. \end{tabular} & \begin{tabular}[c]{@{}c@{}}Forged \\Bit Acc. \end{tabular} & \begin{tabular}[c]{@{}c@{}}Forged \\Bit Acc. \end{tabular} & \begin{tabular}[c]{@{}c@{}}Forged \\ Bit Acc. \end{tabular} \\

\midrule
Clean vanilla & - & 65.44\% & 78.72\% & 84.80\% & 94.49\% \\
\midrule
{W-ER} & - & 98.99\% & 98.96\% & 100.0\% & 100.0\% \\
\midrule
\multirow{3}{*}{\begin{tabular}[c]{@{}c@{}}W-CR:\\ Additive \\Gaussian noise\end{tabular}} & 0.10 & 91.10\% & 91.52\% & 98.61\% & 98.59\% \\
 & 0.25 & 95.22\% & 95.14\% & 98.74\% & 98.76\% \\
 & 0.50 & 96.66\% & 96.61\% & 92.13\% & 92.06\% \\
 \midrule
\multirow{3}{*}{\begin{tabular}[c]{@{}c@{}}W-CR:\\ Affine\\ transformation \end{tabular}} & 0.01 & 81.34\% & 87.26\% & 94.29\% & 94.63\% \\
 & 0.02 & 77.97\% & 86.98\% & 95.39\% & 95.36\% \\
 & 0.03 & 81.79\% & 84.34\% & 91.56\% & 91.58\% \\
\bottomrule
\end{tabular} \label{tab:ablation}
}
\end{table}


\section{Robustness Results on W-CR with $\CL_{RIL}$} \label{sec:res-WCR-IB}  

Table~\ref{tab:certified_acc_mi} summarizes the certified accuracy of W-CR under different settings. The results indicate that our residual information loss can maintain
the certified robustness of the W-CR without the residual information loss (refer to Table~\ref{tab:certified_acc}). For instance, their average certified accuracy are respectively 99.43\% vs 99.27\% under additive Gaussian noise, and 98.73\% vs 98.28\% under affine transformations, respectively. This validates that the designed residual information loss almost does not affect the 
model's robust training. 


We also report the certified accuracy across various radii under different datasets, models, and noise perturbations in Figure~\ref{fig:additive_certify_result_mi} and Figure~\ref{fig:affine_certify_result_mi}. The results also demonstrate that the introduced residual information loss does not affect the certified radius. 

\vspace{-0.05in}

\begin{table}[h]
\centering
 \small
\setlength\tabcolsep{2pt}
\caption{Certified robustness of W-CR with residual information loss $\CL_{RIL}$.}
\vspace{-2mm}
\resizebox{\linewidth}{!}{
\begin{tabular}{ccccc}
\toprule
 & \multicolumn{2}{c}{Additive Gaussian noise} & \multicolumn{2}{c}{Affine transformation} \\
\cmidrule(r){2-3} \cmidrule(r){4-5}
Dataset (Model) & Noise ($\sigma$) & Certified Acc.↑ & Noise ($\sigma$) & Certified Acc.↑ \\
\midrule
\multirow{3}{*}{\begin{tabular}[c]{@{}c@{}}COCO\\ (StegaStamp)\end{tabular}} & 0.10 & 98.80\% & 0.01 & 98.40\% \\
 & 0.25 & 97.60\% & 0.02 & 98.20\% \\
 & 0.50 & 98.00\% & 0.03 & 97.40\% \\
 \midrule
\multirow{3}{*}{\begin{tabular}[c]{@{}c@{}}CelebA\\ (StegaStamp)\end{tabular}} & 0.10 & 100.00\% & 0.01 & 100.00\% \\
 & 0.25 & 100.00\% & 0.02 & 100.00\% \\
 & 0.50 & 100.00\% & 0.03 & 100.00\% \\
 \midrule
\multirow{3}{*}{\begin{tabular}[c]{@{}c@{}}COCO\\ (HiDDeN)\end{tabular}} & 0.10 & 100.00\% & 0.01 & 100.00\% \\
 & 0.25 & 100.00\% & 0.02 & 100.00\% \\
 & 0.50 & 100.00\% & 0.03 & 100.00\% \\
 \midrule
\multirow{3}{*}{\begin{tabular}[c]{@{}c@{}}CelebA\\ (HiDDeN)\end{tabular}} & 0.10 & 100.00\% & 0.01 & 100.00\% \\
 & 0.25 & 100.00\% & 0.02 & 100.00\% \\
 & 0.50 & 100.00\% & 0.03 & 100.00\% \\
 \midrule
Average &  & 99.25\% &  & 99.50\% 
\\ \bottomrule
\end{tabular}
}\vspace{-0.15in}
\label{tab:certified_acc_mi}
\end{table}

\begin{table*}[htb]
\centering
\scriptsize
\setlength\tabcolsep{1.5pt}
\caption{
Authentication effectiveness and visual quality of different watermarking methods with $\CL_{RIL}$.} 
\vspace{-2mm} \resizebox{\linewidth}{!}{
\begin{tabular}{ccccccccccccccccccc}
\toprule
\multicolumn{1}{l}{} & \multicolumn{4}{c}{Clean vanilla} & \multicolumn{4}{c}{W-ER} & \multicolumn{5}{c}{ W-CR: Additive Gaussian noise} & \multicolumn{5}{c}{ W-CR: Affine transformation} \\
\cmidrule(r){2-5} \cmidrule(r){6-9} \cmidrule(r){10-14} \cmidrule(r){15-19} 
\begin{tabular}[c]{@{}c@{}}Dataset\\ (Model)\end{tabular} & PSNR↑ & SSIM↑ & \begin{tabular}[c]{@{}c@{}}Clean \\ Bit Acc.↑\end{tabular} & \begin{tabular}[c]{@{}c@{}}Clean \\ 
Acc.↑\end{tabular} & PSNR↑ & SSIM↑ & \begin{tabular}[c]{@{}c@{}}Clean \\ Bit Acc.↑\end{tabular} & \begin{tabular}[c]{@{}c@{}}Clean \\ Acc.↑\end{tabular} & \begin{tabular}[c]{@{}c@{}}Noise\\ ($\sigma$)\end{tabular} & PSNR↑ & SSIM↑ & \begin{tabular}[c]{@{}c@{}}Clean\\  Bit Acc.↑\end{tabular} & \begin{tabular}[c]{@{}c@{}}Clean\\  Acc.↑\end{tabular} & \begin{tabular}[c]{@{}c@{}}Noise\\ ($\sigma$)\end{tabular} & PSNR↑ & SSIM↑ & \begin{tabular}[c]{@{}c@{}}Clean\\  Bit Acc.↑\end{tabular} & \begin{tabular}[c]{@{}c@{}}Clean\\  Acc.↑\end{tabular} \\
\midrule
\multirow{3}{*}{\begin{tabular}[c]{@{}c@{}}COCO\\ (Stega-\\Stamp)\end{tabular}} & \multirow{3}{*}{27.66} & \multirow{3}{*}{0.9041} & \multirow{3}{*}{99.88\%} & \multirow{3}{*}{100.0\%} & \multirow{3}{*}{25.77} & \multirow{3}{*}{0.7767} & \multirow{3}{*}{100.0\%} & \multirow{3}{*}{100.0\%} 
& 0.10 & 24.78 & 0.8775 & 99.98\% & 100.0\% & 0.01 & 27.84 & 0.9144 & 99.95\% & 100.0\% \\
&  &  &  &  &  &  &  &  & 0.25 & 23.69 & 0.8240 & 100.0\% & 100.0\% & 0.02 & 26.97 & 0.9039 & 99.98\% & 100.0\% \\
&  &  &  &  &  &  &  &  & 0.50 & 20.03 & 0.7363 & 99.99\% & 100.0\% & 0.03 & 25.03 & 0.8909 & 99.98\% & 100.0\% \\
 \midrule
\multirow{3}{*}{\begin{tabular}[c]{@{}c@{}}CelebA\\ (Stega-\\Stamp)\end{tabular}} & \multirow{3}{*}{27.48} & \multirow{3}{*}{0.8975} & \multirow{3}{*}{99.69\%} & \multirow{3}{*}{100.0\%} & \multirow{3}{*}{27.46} & \multirow{3}{*}{0.7609} & \multirow{3}{*}{99.97\%} & \multirow{3}{*}{100.0\%} 
& 0.10 & 23.37 & 0.8202 & 99.99\% & 100.0\% & 0.01 & 33.61 & 0.9554 & 99.98\% & 100.0\% \\
&  &  &  &  &  &  &  &  & 0.25 & 23.20 & 0.7992 & 100.0\% & 100.0\% & 0.02 & 31.14 & 0.9525 & 99.99\% & 100.0\% \\
&  &  &  &  &  &  &  &  & 0.50 & 20.38 & 0.6457 & 99.98\% & 100.0\% & 0.03 & 26.61 & 0.8929 & 99.98\% & 100.0\% \\
 \midrule
\multirow{3}{*}{\begin{tabular}[c]{@{}c@{}}COCO\\ (HiDDeN)\end{tabular}} & \multirow{3}{*}{32.27} & \multirow{3}{*}{0.9397} & \multirow{3}{*}{79.23\%} & \multirow{3}{*}{100.0\%} & \multirow{3}{*}{27.94} & \multirow{3}{*}{0.8752} & \multirow{3}{*}{81.08\%} & \multirow{3}{*}{100.0\%} 
& 0.10 & 25.71 & 0.8034 & 85.34\% & 99.92\% & 0.01 & 29.77 & 0.9154 & 81.28\% & 100.0\% \\
&  &  &  &  &  &  &  &  & 0.25 & 21.83 & 0.6949 & 80.23\% & 99.83\% & 0.02 & 29.75 & 0.9125 & 81.52\% & 100.0\% \\
&  &  &  &  &  &  &  &  & 0.50 & 19.51 & 0.6044 & 79.96\% & 100.0\% & 0.03 & 29.88 & 0.9113 & 80.34\% & 100.0\% \\
 \midrule\multirow{3}{*}{\begin{tabular}[c]{@{}c@{}}CelebA\\ (HiDDeN)\end{tabular}} & \multirow{3}{*}{35.14} & \multirow{3}{*}{0.9599} & \multirow{3}{*}{88.45\%} & \multirow{3}{*}{100.0\%} & \multirow{3}{*}{32.26} & \multirow{3}{*}{0.9472} & \multirow{3}{*}{77.82\%} & \multirow{3}{*}{100.0\%} 
 & 0.10 & 26.20 & 0.8105 & 83.77\% & 100.0\% & 0.01 & 32.07 & 0.9419 & 89.13\% & 100.0\% \\
 &  &  &  &  &  &  &  &  & 0.25 & 22.68 & 0.7178 & 77.37\% & 100.0\% & 0.02 & 31.61 & 0.9414 & 87.88\% & 100.0\% \\
 &  &  &  &  &  &  &  &  & 0.50 & 21.21 & 0.7199 & 76.67\% & 99.86\% & 0.03 & 22.66 & 0.9330 & 78.78\% & 99.88\% \\
 \midrule
Average & 30.64 & 0.9253 & 91.81\% & 100.0\% & 28.36 & 0.8400 & 89.72\% & 100.0\% &  & 22.72 & 0.7545 & 90.27\% & 99.97\% &  & 28.91 & 0.9221 & 91.57\% & 99.99\%
\\ \bottomrule
\end{tabular}
} 
\label{tab:auth_acc_mi}
\end{table*}

\begin{figure}[th]
    \centering
    \begin{subfigure}[b]{0.45\columnwidth}
        \centering
        \includegraphics[width=\textwidth]{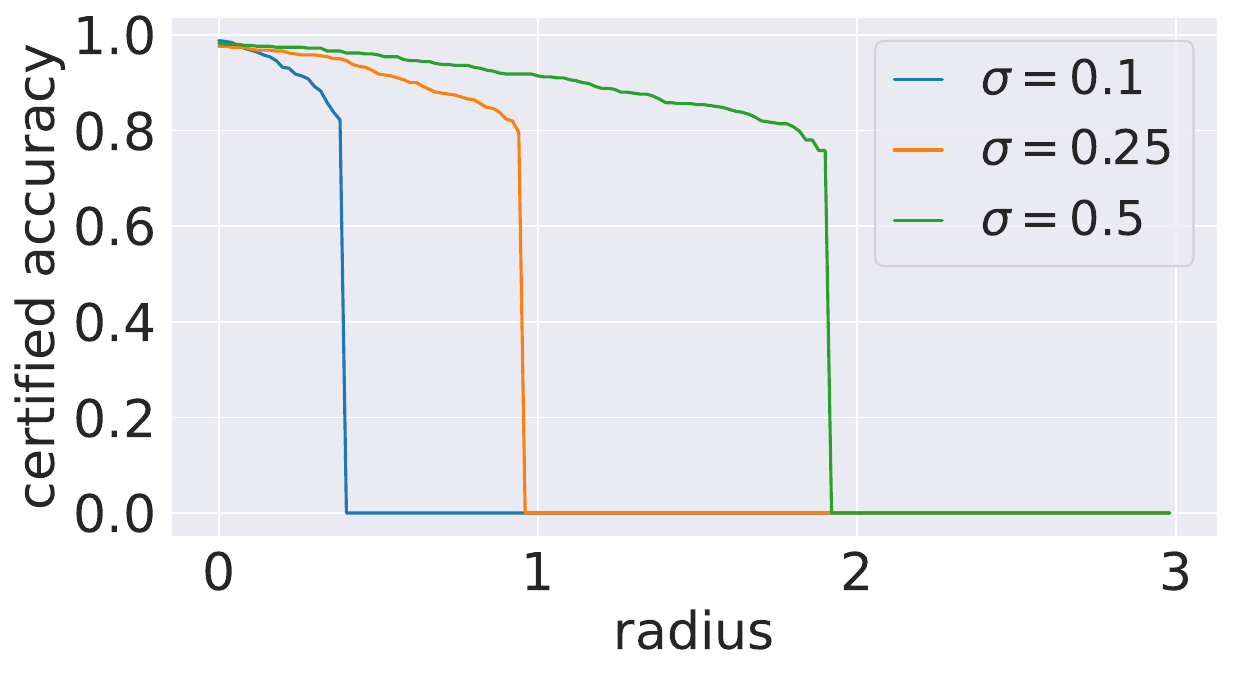}
    \end{subfigure}
    \begin{subfigure}[b]{0.45\columnwidth}
        \centering
        \includegraphics[width=\linewidth]{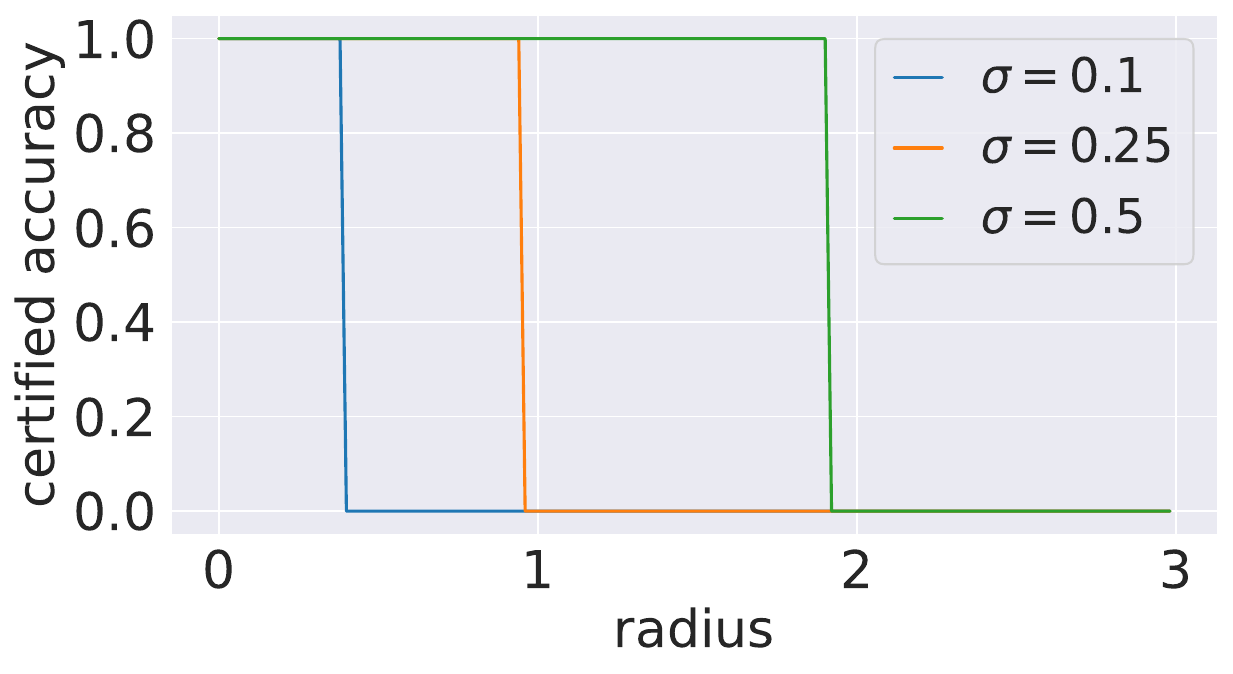}
    \end{subfigure}
    
    \begin{subfigure}[b]{0.45\columnwidth}
        \centering
        \includegraphics[width=\textwidth]{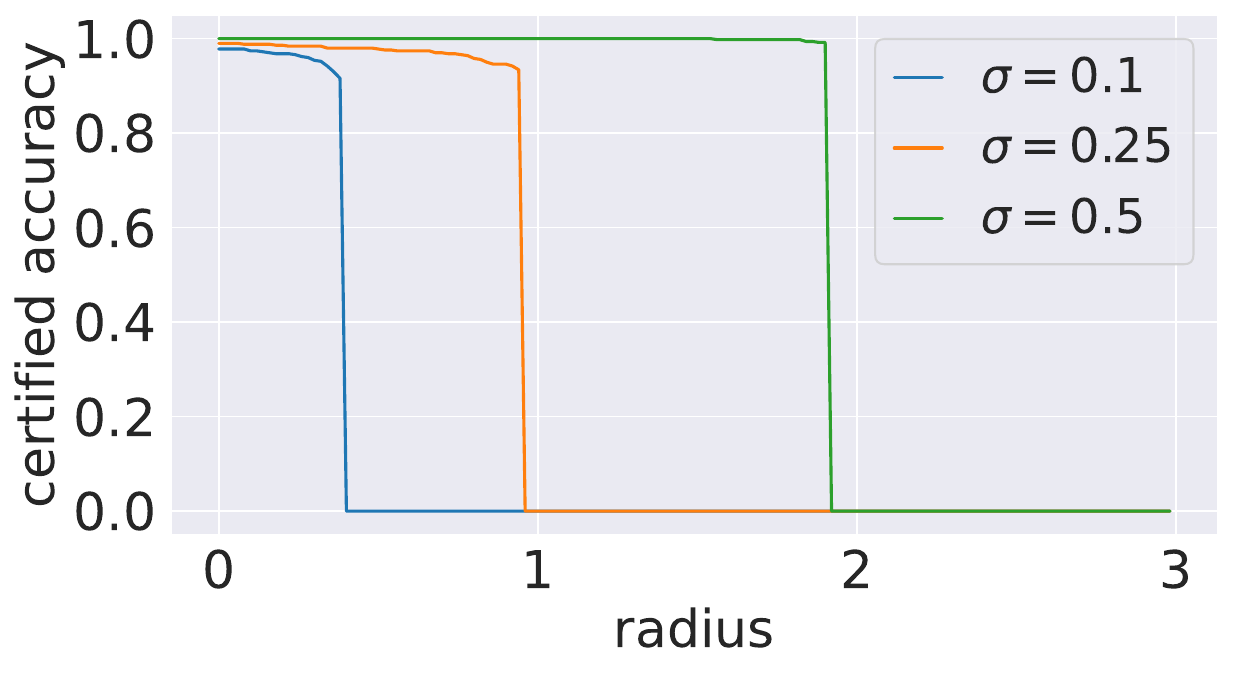}
    \end{subfigure}
    \begin{subfigure}[b]{0.45\columnwidth}
        \centering
        \includegraphics[width=\linewidth]{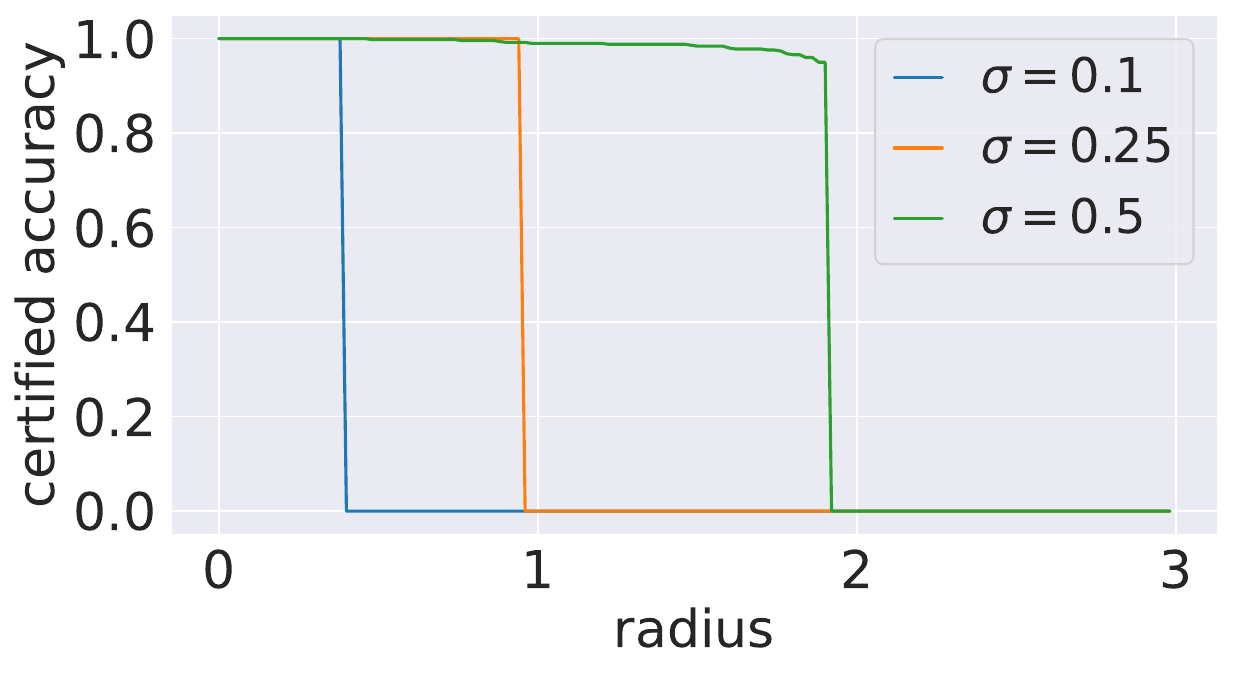}
    \end{subfigure}   

\vspace{-2mm}
\caption{Certified accuracy (w/ $\CL_{RIL}$) at different radii against additive Gaussian noise. Datasets: left-COCO and right-CelebA; Models: top-StegaStamp and bottom-HiDDeN.} 
\label{fig:additive_certify_result_mi}
\vspace{-4mm}
\end{figure}

\begin{figure}[th]
    \centering
    \begin{subfigure}[b]{0.45\columnwidth}
        \centering
        \includegraphics[width=\textwidth]{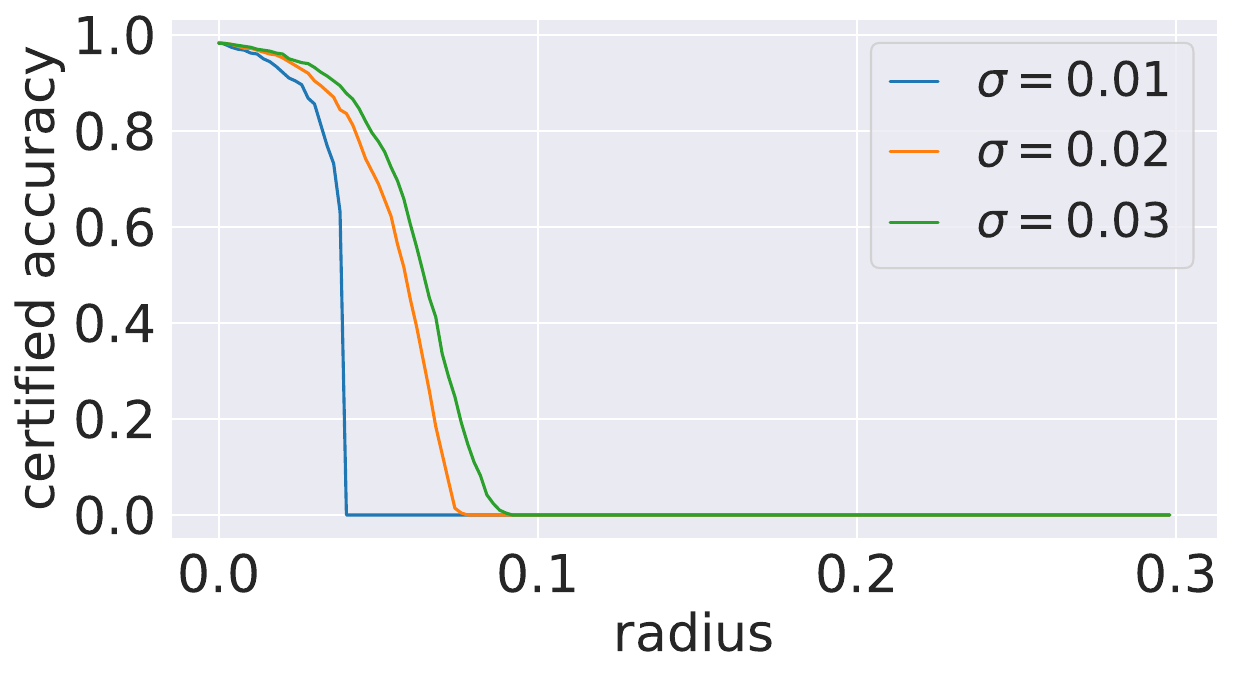}
    \end{subfigure}
    \begin{subfigure}[b]{0.45\columnwidth}
        \centering
        \includegraphics[width=\linewidth]{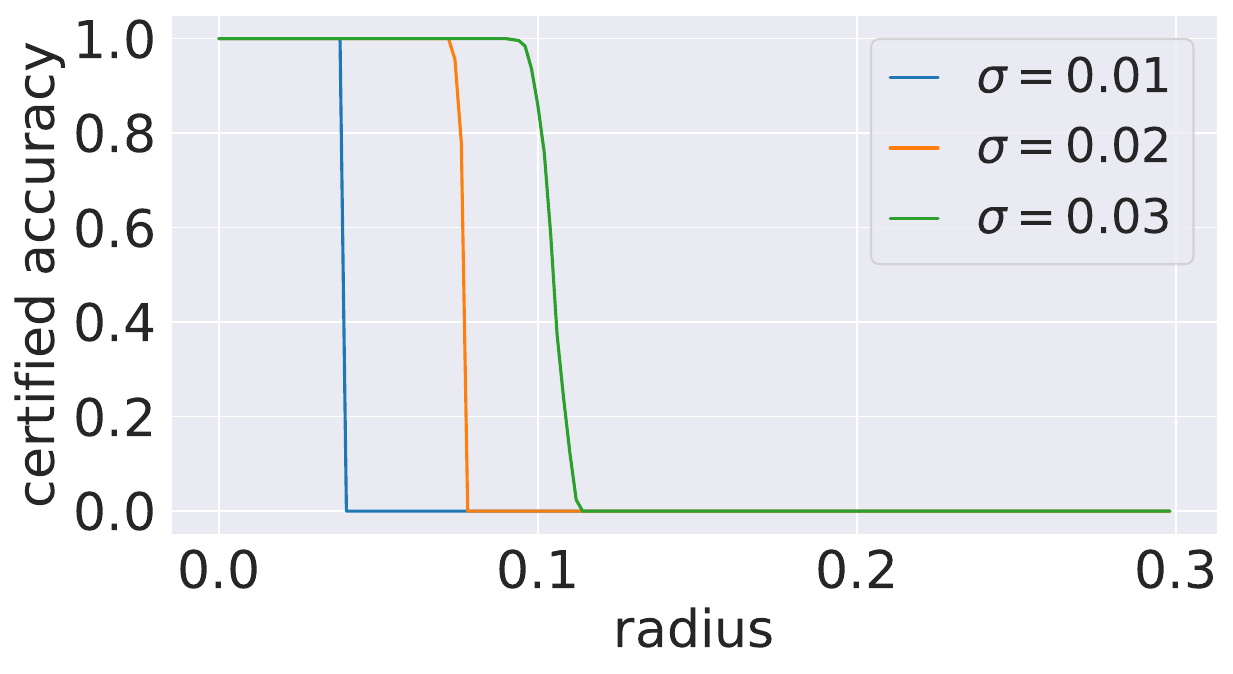}
    \end{subfigure}
    
    \begin{subfigure}[b]{0.45\columnwidth}
        \centering
        \includegraphics[width=\textwidth]{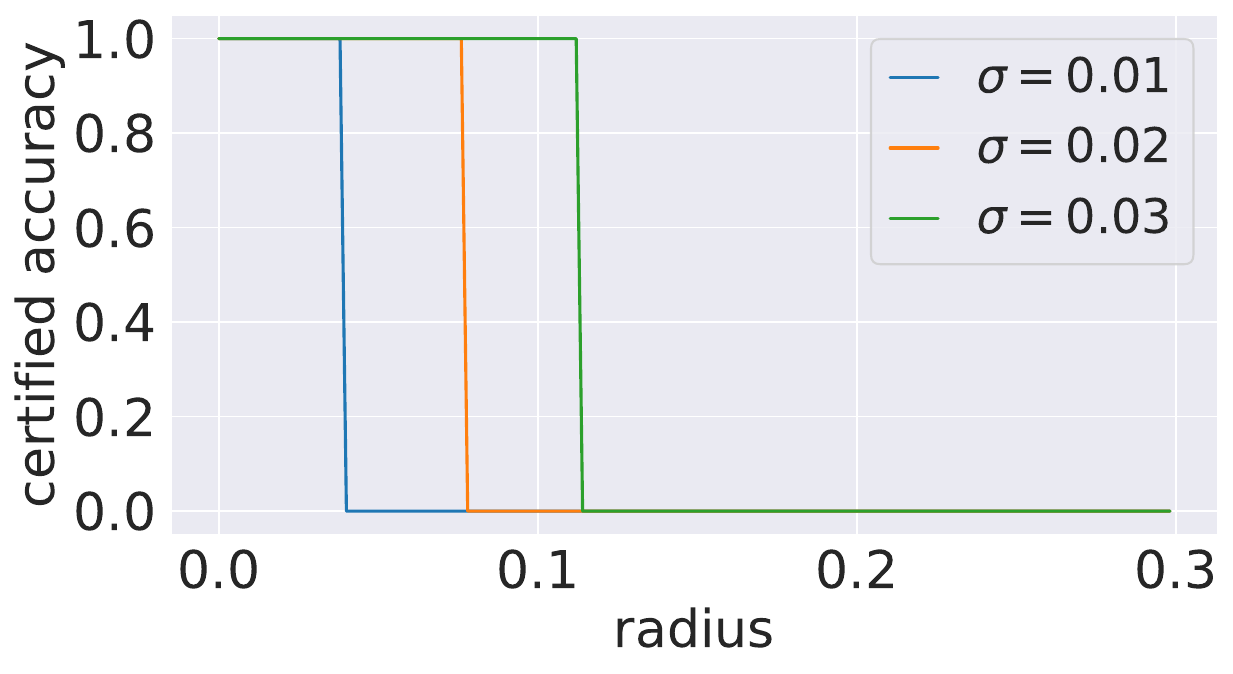}
    \end{subfigure}
    \begin{subfigure}[b]{0.45\columnwidth}
        \centering
        \includegraphics[width=\linewidth]{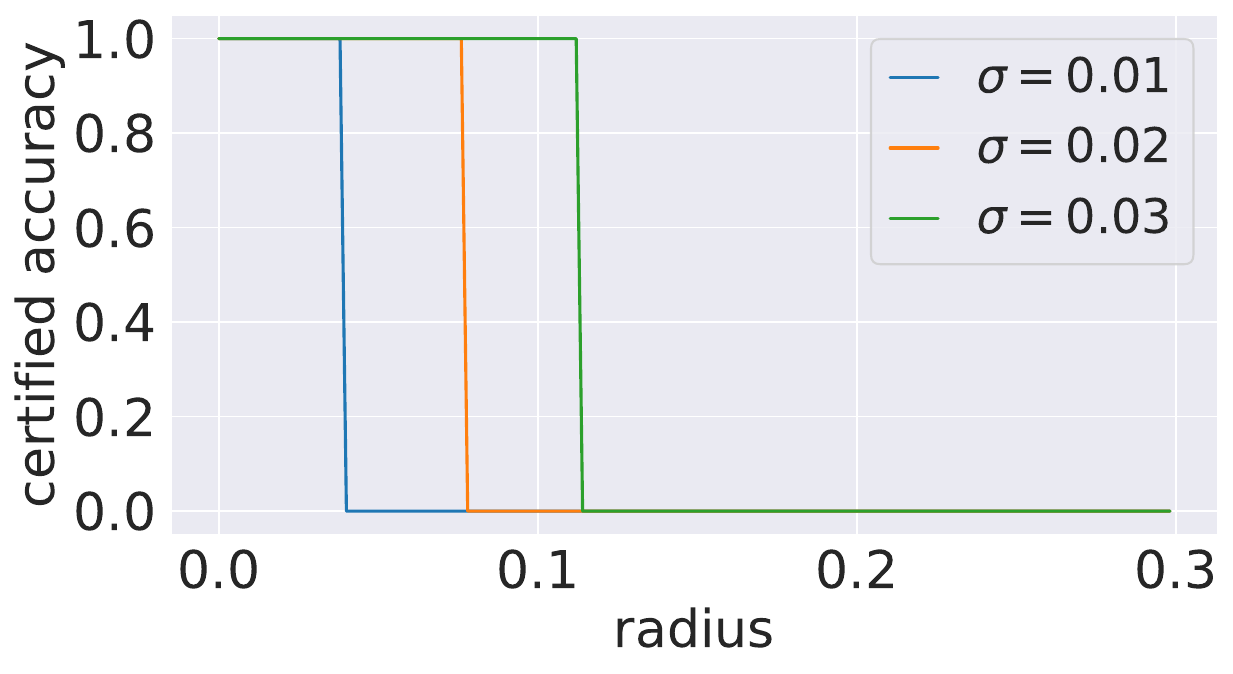}
    \end{subfigure}   

\vspace{-2mm}
\caption{Certified accuracy (w/ $\CL_{RIL}$) at different radii against affine transformation. Datasets: left-COCO and right-CelebA; Models: top-StegaStamp and bottom-HiDDeN.} 
\label{fig:affine_certify_result_mi}
\end{figure}

\section{Additional Experimental Results}


\begin{figure}[th]
    \centering
    \begin{subfigure}[b]{0.48\columnwidth}
        \centering
        \includegraphics[width=\textwidth]{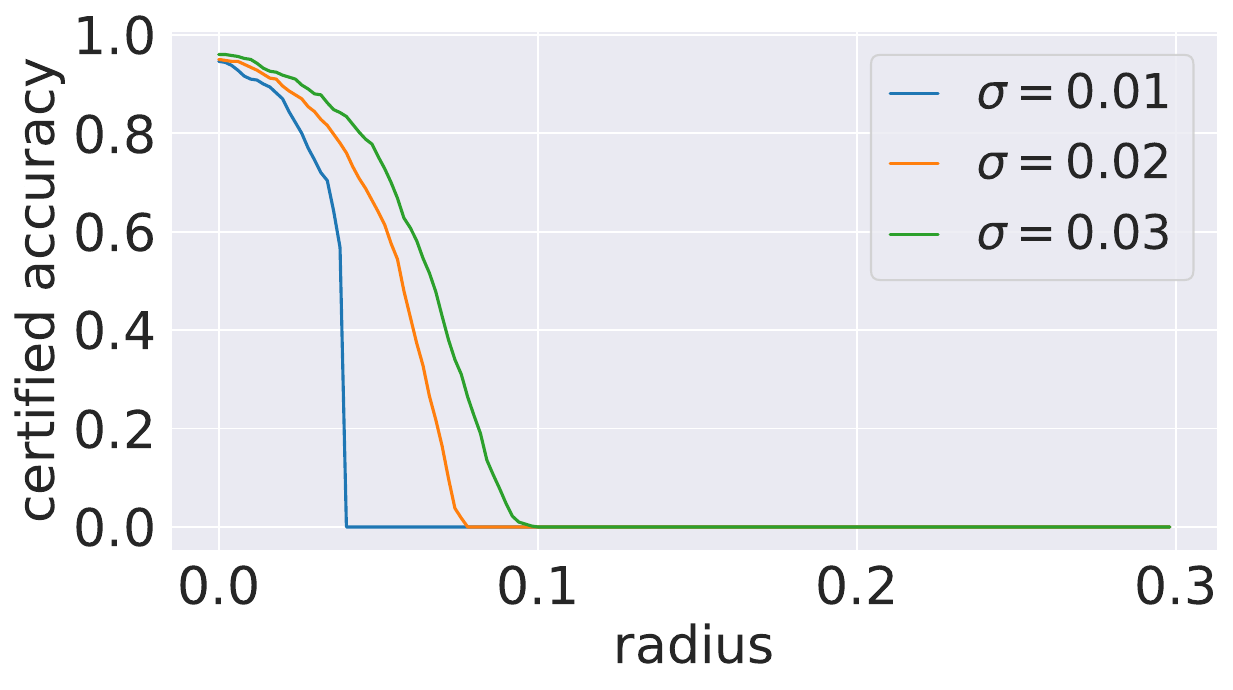}
    \end{subfigure}
    \begin{subfigure}[b]{0.48\columnwidth}
        \centering
        \includegraphics[width=\linewidth]{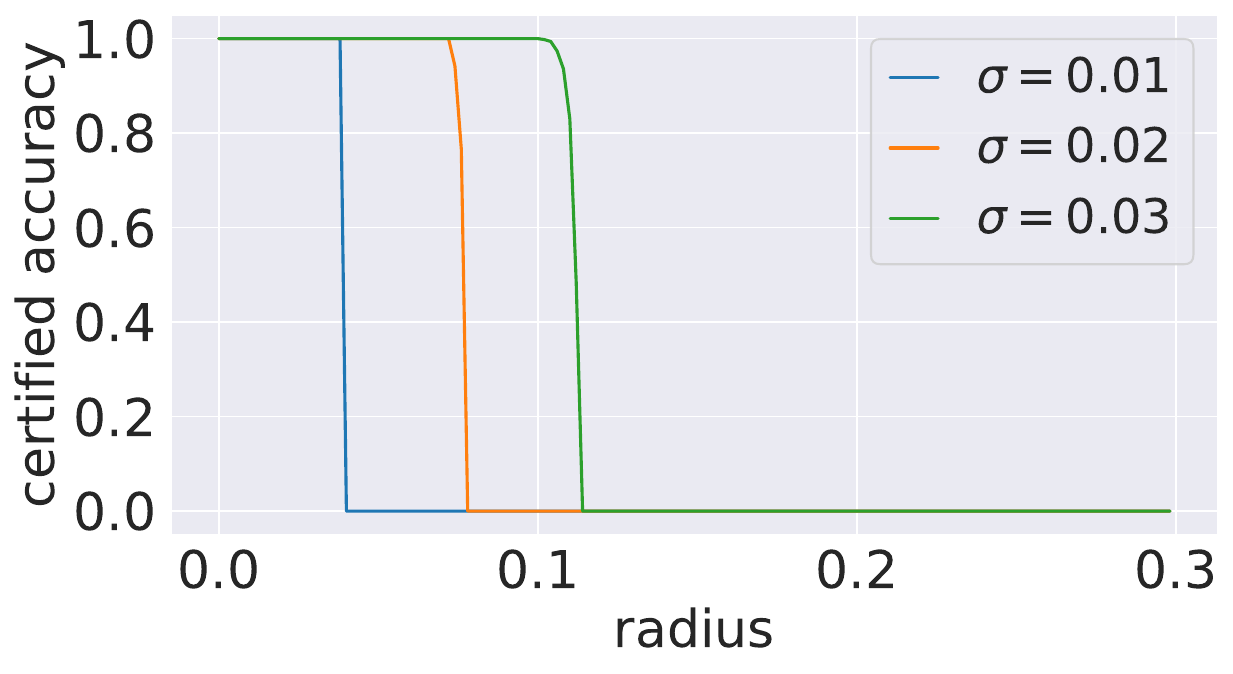}
    \end{subfigure}
    
    \begin{subfigure}[b]{0.48\columnwidth}
        \centering
        \includegraphics[width=\textwidth]{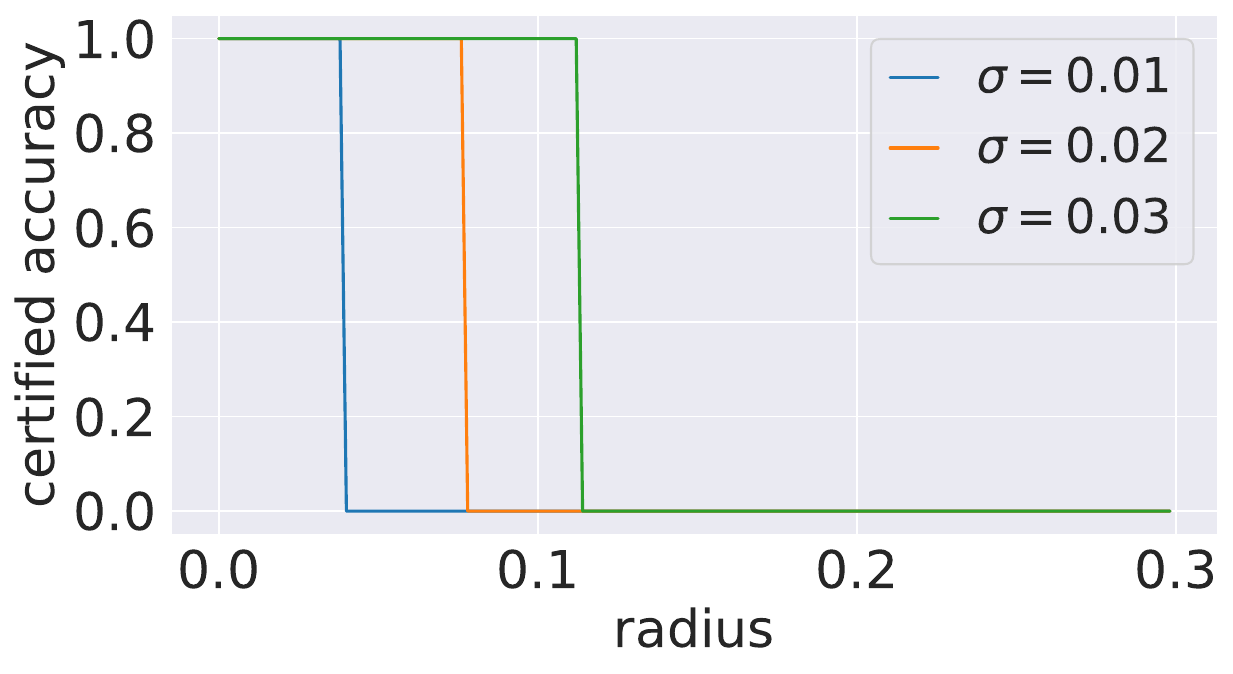}
    \end{subfigure}
    \begin{subfigure}[b]{0.48\columnwidth}
        \centering
        \includegraphics[width=\linewidth]{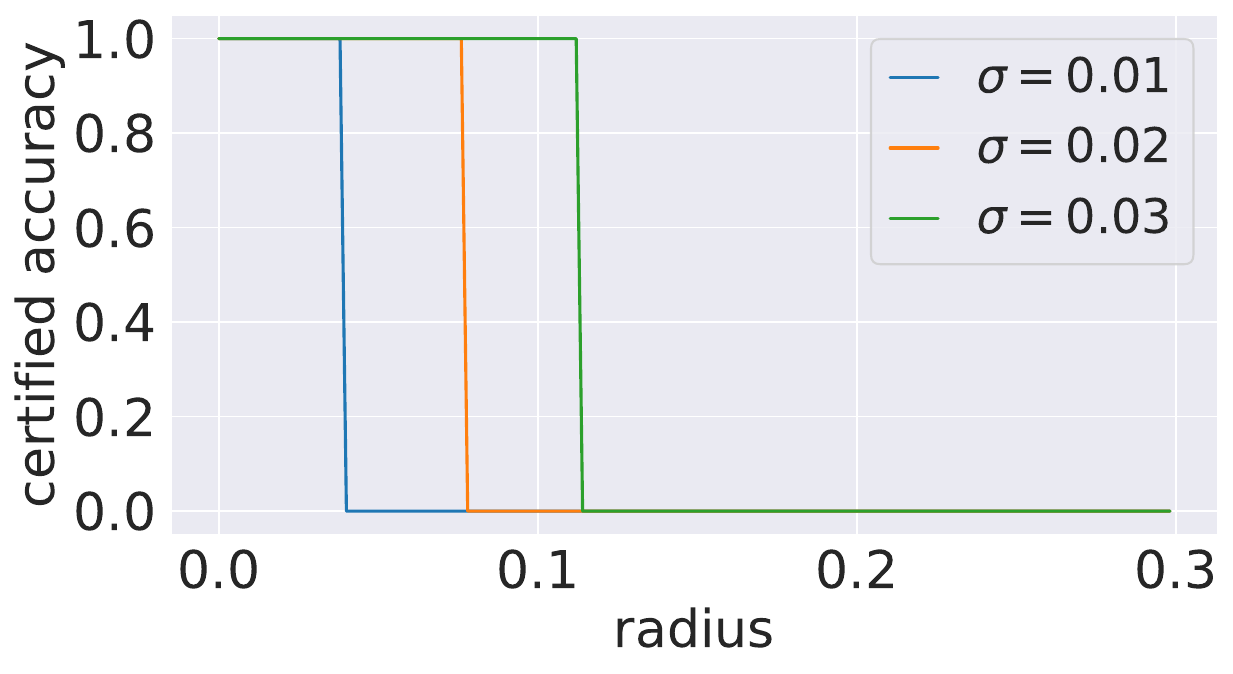}
    \end{subfigure}   

\vspace{-2mm}
\caption{Certified accuracy (w/o $\CL_{RIL}$) at different radii against affine transformation. Datasets: left-COCO and right-CelebA; Models: top-StegaStamp and bottom-HiDDeN.} 
\label{fig:affine_certify_result_Affine}
\end{figure}

\begin{figure}[th]
    \centering
    \begin{subfigure}[b]{0.78\columnwidth}
        \centering
        \includegraphics[width=\linewidth]{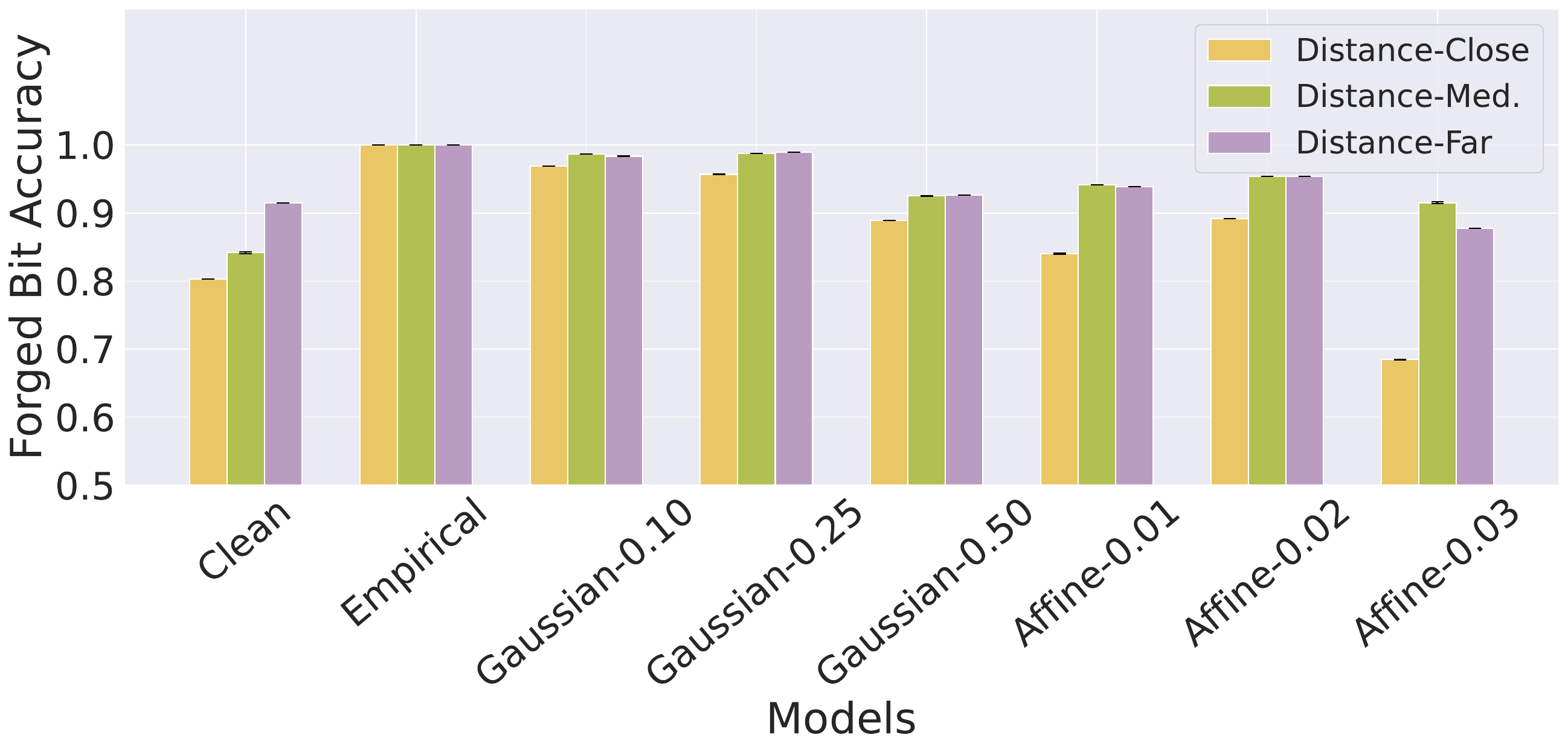}
    \end{subfigure}
\vspace{-2mm}
\caption{Impact of secret watermark pairs' distances under CelebA (StegaStamp).}
\label{fig:forgery_result_CelebA}
\end{figure}

Our robust training (both W-ER and W-CR) maintains $O(1)$ complexity vs. vanilla training, with minor overhead on noise injection. During standard inference, W-ER maintains $O(1)$ complexity, identical to vanilla inference. For certification inference, W-CR operates at $O(N/k)$ complexity ($N$:noise sampling count; $k$:inference batch size). Table~\ref{tab:computation_cost} shows that both W-ER and W-CR are efficient in large-scale applications.

\begin{table}[H]
\centering
\scriptsize
\setlength\tabcolsep{2pt}
\caption{Training and inference time (seconds) of W-IR, including W-ER and W-CR.}
\vspace{-2mm}
\resizebox{\linewidth}{!}{
\begin{tabular}{cccccc}
\toprule
Dataset(Model) & \multicolumn{1}{l}{Phase} & \begin{tabular}[c]{@{}c@{}}Clean \\ vanilla \end{tabular} & \begin{tabular}[c]{@{}c@{}}W-ER \\ without $L_{RIL}$\end{tabular} & \begin{tabular}[c]{@{}c@{}}W-CR \\ without $L_{RIL}$\end{tabular} & \begin{tabular}[c]{@{}c@{}}W-CR \\ with $L_{RIL}$\end{tabular} \\
\midrule
\multirow{2}{*}{\begin{tabular}[c]{@{}c@{}}COCO\\ (StegaStamp)\end{tabular}} & Training & 4073 & 4105 & 4106 & 4394 \\
 & Inference & \textless{}1 & \textless{}1 & 4.46 & 4.46 \\
\midrule
\multirow{2}{*}{\begin{tabular}[c]{@{}c@{}}Imagenet\\ (StegaStamp)\end{tabular}} & Training & 4923 & 5184 & 4903 & 5832 \\
 & Inference & \textless{}1 & \textless{}1 & 12.7 & 12.79
 \\ \bottomrule
\end{tabular}
}\vspace{-2mm}
\label{tab:computation_cost}
\end{table}

\begin{table}[H]
\centering
\scriptsize
\setlength\tabcolsep{2pt}
\caption{Accuracy of W-CR and W-ER against unseen adversarial examples.}
\vspace{-2mm}
\resizebox{\linewidth}{!}{
\begin{tabular}{ccccccc}
\toprule
 & Clean vanilla & W-ER & \multicolumn{2}{c}{W-CR (Gaussian noise)} & \multicolumn{2}{c}{W-CR (Affine)} \\
\cmidrule(r){2-2} \cmidrule(r){3-3} \cmidrule(r){4-5} \cmidrule(r){6-7}
Dataset(Model) & Acc.↑ & Acc.↑ & Noise($\sigma$) & Certified Acc.↑ & Noise($\sigma$) & Certified Acc.↑ \\
\midrule
\multirow{3}{*}{\begin{tabular}[c]{@{}c@{}}CelebA\\ (StegaStamp)\end{tabular}} & \multirow{3}{*}{0\%} & \multirow{3}{*}{10.00\%} & 0.10 & 100\% & 0.01 & 100\% \\
 &  &  & 0.25 & 100\% & 0.02 & 100\% \\
 &  &  & 0.50 & 100\% & 0.03 & 100\%
 \\ \bottomrule
\end{tabular}
}
\label{tab:unseen_attacks}
\end{table}


To evaluate W-CR's robustness against unseen adversarial examples, we combine pixel and coordinate transformations to generate 100 novel adversarial samples and evaluate their classification accuracy under CelebA (StegaStamp). The results in Table~\ref{tab:unseen_attacks} demonstrate that the W-ER model shows minimal resistance (10\%) to unseen attacks. W-CR models across various noise levels exhibit high defensive capabilities against these adversarial examples.

\begin{table}[!ht]
\centering
\scriptsize
\setlength\tabcolsep{4pt}
\caption{Comparison of identity leakage of original and estimated residual images under COCO (StegaStamp). The original refers to the known original image, and `Estimate' uses VAE to estimate the original image. }
\vspace{-2mm}
\resizebox{\linewidth}{!}{
\begin{tabular}{cccccc}
\toprule
 &  & \multicolumn{2}{c}{Silhouette Score} & \multicolumn{2}{c}{Forged Bit Accuracy} \\
 \cmidrule(r){3-4} \cmidrule(r){5-6}
Perturbations & Noise($\sigma$) & \begin{tabular}[c]{@{}c@{}}Original \end{tabular} & \begin{tabular}[c]{@{}c@{}}Estimate\end{tabular} & \begin{tabular}[c]{@{}c@{}}Original \end{tabular} & \begin{tabular}[c]{@{}c@{}}Estimate\end{tabular} \\
\midrule
\multirow{3}{*}{\begin{tabular}[c]{@{}c@{}} Additive \\Gaussian noise \end{tabular}} & 0.10 & 0.5519 & 0.3637 & 91.10\% & 71.74\% \\
 & 0.25 & 0.7416 & 0.2885 & 95.22\% & 66.34\% \\
 & 0.50 & 0.8327 & 0.3311 & 96.66\% & 68.46\% \\
 \midrule
\multirow{3}{*}{\begin{tabular}[c]{@{}c@{}}Affine \\ transformation\end{tabular}} & 0.01 & 0.3661 & 0.3772 & 81.34\% & 71.13\% \\
 & 0.02 & 0.3792 & 0.3604 & 77.97\% & 72.69\% \\
 & 0.03 & 0.3869 & 0.3775 & 81.79\% & 65.54\%
 \\
 \bottomrule
\end{tabular}
}
\label{tab:VAEestimate}
\end{table}

\end{document}